\documentclass{article}

\usepackage{arxiv}

\usepackage[utf8]{inputenc} 
\usepackage[T1]{fontenc}    
\usepackage{hyperref}       

\hypersetup{
    colorlinks=true,   
    linkcolor=blue,    
    citecolor=blue,    
    filecolor=blue,    
    urlcolor=blue      
}

\usepackage{url}            
\usepackage{booktabs}       
\usepackage{amsfonts}       
\usepackage{nicefrac}       
\usepackage{microtype}      
\usepackage{physics}
\usepackage{lipsum}
\usepackage{natbib}   
\usepackage{graphicx}
\usepackage{tikz}
\usepackage{placeins}
\usepackage{scalerel}
\usepackage{mathpartir}
\usepackage{blindtext}
\usepackage{scrextend}
\addtokomafont{labelinglabel}{\sffamily}

\usepackage{tcolorbox}

\usetikzlibrary{arrows}

\usepackage{amsmath,amssymb,amsfonts}
\usepackage{algorithm}
\usepackage{algorithmicx} 
\usepackage{algpseudocode} 
\usepackage{subcaption}

\usetikzlibrary{quantikz2}
\usepackage{amsmath}
\usepackage{amsthm}
\usepackage{bm}

\newtheorem{definition}{Definition}[section]
\newtheorem{theorem}{Theorem}[section]
\newtheorem{corollary}{Corollary}[section]

\newtheorem{lemma}{Lemma}[section]

\graphicspath{ {./images/} }

\title{Statistical modeling of Quantum error propagation}

\author{
 Zhuoyang Ye \\
  Computer Science\\
  University of California Los Angeles\\
  Los Angeles, CA 15213 \\
  \texttt{yezhuoyang@cs.ucla.edu} \\
}

\begin{document}
\maketitle
\begin{abstract}
In this paper, I design a new statistical abstract model for studying quantum error propagation. For each circuit, I give the algorithm to construct the Error propagation space-time graph(\textbf{EPSTG}) graph as well as the bipartite reverse spanning graph (\textbf{RSG}). Then I prove that the problem of finding an error pattern is $\mathcal{P}$ while calculate the error number distribution is $\textit{NP-complete}$. I invent the new measure for error propagation and show that for widely used transversal $CNOT$ circuit in parallel, the shift of distribution is bounded by $\frac{n}{27}$, where $n$ is the number of physical qubits. The consistency between the result of qiskit simulation and my algorithm justify the correctness of my model. Applying the framework to random circuit, I show that there is severe unbounded error propagation when circuit has global connection. We also apply my framework on parallel transversal logical $CNOT$ gate in surface code, and demonstrate that the error threshold will decrease from $0.231$ to $0.134$ per cycle. 
\end{abstract}


\section{Introduction}

Quantum computers, by virtue of quantum superposition and quantum entanglement, have many algorithmic applications with exponential speedup that might revolutionize the human world and society, such as factoring a large integer number \citep{shor1999polynomial} simulating complex quantum material \citep{childs2012hamiltonian}, solving linear equations \citep{Harrow_2009}, designing new drugs \citep{Santagati_2024}. The vast amount of success in the design of new quantum algorithms has inspired physicists to build a new physical quantum computer and apply these new quantum algorithms on their platform. Such physical systems include Nuclear magnetic resonance system(NMR)\citep{jones1998implementation}, superconducting Josephson junction\citep{clarke2008superconducting}, trapped ions\citep{cirac1995quantum}, linear optics\citep{knill2001scheme},semiconductor and quantum dots\citep{awschalom2002semiconductor}, as well as  Rydberg neutral atoms\citep{saffman2010quantum}. Among all competing physical qubits, trapped ion has already achieved the best single qubit gate fidelity of over $99.9999\%$ \citep{Harty_2014}, which pave the way for the future large scale quantum computer robust to noise. Some people even believe now that quantum supremacy over classical computer has already been achieved on random circuit sampling \citep{arute2019quantum} and Boson sampling tasks \citep{tillmann2013experimental}.

Despite the success in improving quantum hardware and inventing algorithm, the progress in implementing verifiable fault tolerant algorithm in real quantum computer is still very slow, the physical realization of a scalable fault-tolerant quantum computer is still considered by most people to be a fantasy. In 2001, a Shor's algorithm for factoring $15$ has been implemented on Nuclear magnetic resonance NMR \citep{Vandersypen_2001}, where hundreds of radio frequency(RF) pulses were applied on a room temperature molecule with seven spin-1/2 nuclei. However, the record for factoring Shor's algorithm has never been broken. One of the reason is the astronomical resource required to implement the algorithm fault tolerantly\citep{Gidney2021howtofactorbit}.

The biggest challenge for quantum computing is quantum error correction(QEC). Specifically, because of unavoidable decoherence and thermal fluctuation in all quantum system, we have to design a system to detect and correct errors.The famous threshold theorem for quantum computing has been shown to guarantee the feasibility to reduce error rate per physical gate when the physical error rate is below some constant threshold value \cite{gottesman2009introductionquantumerrorcorrection}. Under the same framework of stabilizer formalism \citep{gottesman1998theory}, many quantum error correction codes have been designed, with different structure and design of stabilizers made up of pauli $\hat{Z}$ and $\hat{X}$ operators. Among the first stabilizer code is the Calderbank-Shor-Steane(CSS) codes \citep{calderbank1996good,shor1996fault}. Toric code \citep{KITAEV20032} is the first efficient fault tolerant design of logical qubit, which also inspire the invention of the famous surface code \citep{fowler2012surface}.  

In the process of designing and verifying quantum error correction code, it is essential to build a model for random quantum noise with some ways of measure to characterize and compare the behavior of each code under different physical platform. Normally, to estimate the quantum error threshold, people assume the error model they use is a random pauli bit-flip model with constant error rate. Such simple statistical estimation, however, has not take into consideration the effect of the Clifford circuit on the noise.

A naive way of estimating the error threshold of surface code is given by calculating the most likely misidentification of error syndrome in \citep{fowler2012surface}. For a $d \times d$ surface code with code distance $d$, when there are more than $d_e=\frac{d+1}{2}$ errors in one row, we can no longer discriminate the error with it's dual error. There are $d$ rows and $\frac{d!}{(d_e-1)!d_e!}$ possible ways for such misidentification. Thus, the logical error rate is given by:

\begin{equation}
    P_L=d \frac{d!}{(d_e-1)!d_e!} P_e^{d_e}
    \label{eq:statisticalestimation}
\end{equation}

\begin{figure}[h!]
    \centering
    \includegraphics[width=0.7\linewidth]{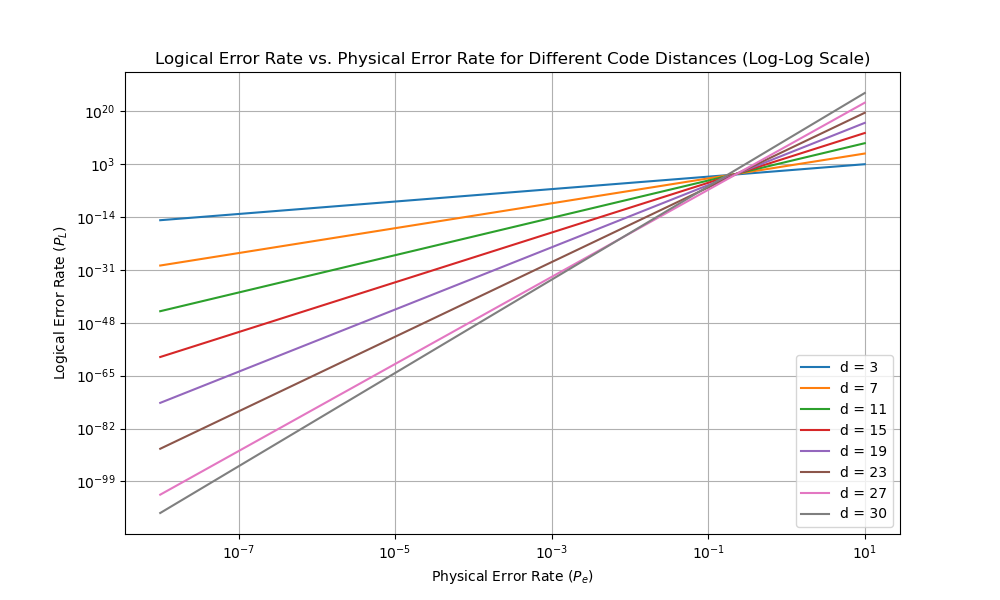}
    \caption{The threshold of surface code estimated by \autoref{eq:statisticalestimation} is roughly $10^{-1}$.}
    \label{fig:surfacethreshold}
\end{figure}

However, this estimation doesn't take into account the complexity introduced by error propagation.

As pointed out by Gottesnman \citep{gottesman2009introductionquantumerrorcorrection},  \begin{quote}\textit{The biggest obstacle which we must overcome in order to create a fault-tolerant
protocol is that of error propagation.}\end{quote}

\begin{figure}[h!]
\centering
\begin{tikzpicture}

    \node[draw, thick, inner sep=0.05cm] (box1) at (0, 0) {
        \begin{tikzpicture}
            \node[anchor=base] (leftPlot1) {
            \resizebox{0.17\columnwidth}{!}{
                \begin{quantikz}
                     Q1        &   \slice{$T=1$} & \targ{}& \slice{\quad $T=2$}    & \ctrl{3} & \slice{\quad $T=3$}  & \ctrl{4} & \slice{\quad $T=4$} &  &\slice{\quad $T=5$}&\ctrl{5}&  \\
                     Q2         & \gate[style={fill=red!80}]{X} & \ctrl{-1} &  &  &  & & & & & &\\
                     Q3        &  &  &   &          &  & &  &  \targ{}           & & &\\
                     Q4        &  &  &   & \targ{}  &  & &  & \ctrl{-1} & & &\\
                     Q5        &  &  &   &  &   & \targ{}& & & & &\\
                     Q6        &  &  &   &  &   & & &   & & \targ{} &
                \end{quantikz}
            }};
            \node[anchor=base, right=0.0cm of leftPlot1] (rightPlot1) {
            \resizebox{0.12\columnwidth}{!}{
                \begin{tikzpicture}[->,>=stealth',shorten >=1pt,auto,node distance=3cm,
                                    thick,main node/.style={circle,draw,font=\sffamily\normalsize\bfseries, minimum size=1cm}]
                  \node[main node] (Q1) {Q1};
                  \node[main node,fill=red] (Q2) [below left of=Q1] {Q2};
                  \node[main node] (Q3) [below right of=Q2] {Q3};
                  \node[main node] (Q4) [below right of=Q1] {Q4};
                  \node[main node] (Q5) [above right of=Q1] {Q5};
                  \node[main node] (Q6) [above left of=Q1] {Q6};
                  \path[every node/.style={font=\sffamily\small}]
                    (Q1) edge node [right] {$T=2$} (Q4)
                        edge node [left] {$T=3$} (Q5)
                        edge node [left] {$T=5$} (Q6)
                    (Q2) edge node [left] {$T=1$} (Q1)
                    (Q4) edge node [left] {$T=4$} (Q3);
                \end{tikzpicture}
            }};
        \end{tikzpicture}
    };

    \node[draw, thick, inner sep=0.05cm, right=0.5cm of box1] (box2) {
        \begin{tikzpicture}
            \node[anchor=base] (leftPlot2) {
            \resizebox{0.17\columnwidth}{!}{
                \begin{quantikz}
                 Q1        &   \slice{$T=1$} & \targ{}& \gate[style={fill=red!80}]{X}\slice{\quad $T=2$}    & \ctrl{3} & \slice{\quad $T=3$}  & \ctrl{4} & \slice{\quad $T=4$} &  &\slice{\quad $T=5$}&\ctrl{5}&  \\
                 Q2         & \gate[style={fill=red!80}]{X} & \ctrl{-1} & \gate[style={fill=red!80}]{X} &  &  & & & & & &\\
                   Q3        &  &  &   &          &  & &  &  \targ{}           & & &\\
                   Q4        &  &  &   & \targ{}  &  & &  & \ctrl{-1} & & &\\
                   Q5        &  &  &   &  &   & \targ{}& & & & &\\
                   Q6        &  &  &   &  &   & & &   & & \targ{} &
                \end{quantikz}
            }};
            \node[anchor=base, right=0.0cm of leftPlot2] (rightPlot2) {
            \resizebox{0.12\columnwidth}{!}{
                \begin{tikzpicture}[->,>=stealth',shorten >=1pt,auto,node distance=2.5cm,
                                    thick,main node/.style={circle,draw,font=\sffamily\Large\bfseries}]
                  \node[main node,fill=red] (Q1) {Q1};
                  \node[main node,fill=red] (Q2) [below left of=Q1] {Q2};
                  \node[main node] (Q3) [below right of=Q2] {Q3};
                  \node[main node] (Q4) [below right of=Q1] {Q4};
                  \node[main node] (Q5) [above right of=Q1] {Q5};
                  \node[main node] (Q6) [above left of=Q1] {Q6};
                  \path[every node/.style={font=\sffamily\small}]
                    (Q1) edge node [right] {$T=2$} (Q4)
                        edge node [left] {$T=3$} (Q5)
                        edge node [left] {$T=5$} (Q6)
                    (Q2) edge node [left] {$T=1$} (Q1)
                    (Q4) edge node [left] {$T=4$} (Q3);
                \end{tikzpicture}
            }};
        \end{tikzpicture}
    };

    \node[draw, thick, inner sep=0.05cm, right=0.5cm of box2] (box3) {
        \begin{tikzpicture}
            \node[anchor=base] (leftPlot3) {
            \resizebox{0.17\columnwidth}{!}{
                \begin{quantikz}
                 Q1        &   \slice{$T=1$} & \targ{}& \gate[style={fill=red!80}]{X}\slice{\quad $T=2$}    & \ctrl{3} &  \gate[style={fill=red!80}]{X}\slice{\quad $T=3$}  & \ctrl{4} & \slice{\quad $T=4$} &  &\slice{\quad $T=5$}&\ctrl{5}&  \\
                 Q2         & \gate[style={fill=red!80}]{X} & \ctrl{-1} & \gate[style={fill=red!80}]{X} &  &   \gate[style={fill=red!80}]{X}& & & & & &\\
                   Q3        &  &  &   &          &  & &  &  \targ{}           & & &\\
                   Q4        &  &  &   & \targ{}  &  \gate[style={fill=red!80}]{X} & &  & \ctrl{-1} & & &\\
                   Q5        &  &  &   &  &   & \targ{}& & & & &\\
                   Q6        &  &  &   &  &   & & &   & & \targ{} &
                \end{quantikz}
            }};
            \node[anchor=base, right=0.0cm of leftPlot3] (rightPlot3) {
            \resizebox{0.12\columnwidth}{!}{
                \begin{tikzpicture}[->,>=stealth',shorten >=1pt,auto,node distance=2.5cm,
                                    thick,main node/.style={circle,draw,font=\sffamily\Large\bfseries}]
                  \node[main node,fill=red] (Q1) {Q1};
                  \node[main node,fill=red] (Q2) [below left of=Q1] {Q2};
                  \node[main node] (Q3) [below right of=Q2] {Q3};
                  \node[main node,fill=red] (Q4) [below right of=Q1] {Q4};
                  \node[main node] (Q5) [above right of=Q1] {Q5};
                  \node[main node] (Q6) [above left of=Q1] {Q6};
                  \path[every node/.style={font=\sffamily\small}]
                    (Q1) edge node [right] {$T=2$} (Q4)
                        edge node [left] {$T=3$} (Q5)
                        edge node [left] {$T=5$} (Q6)
                    (Q2) edge node [left] {$T=1$} (Q1)
                    (Q4) edge node [left] {$T=4$} (Q3);
                \end{tikzpicture}
            }};
        \end{tikzpicture}
    };

    \node[draw, thick, inner sep=0.05cm, below=1cm of box1] (box4) {
        \begin{tikzpicture}
            \node[anchor=base] (leftPlot4) {
            \resizebox{0.17\columnwidth}{!}{
    \begin{quantikz}
     Q1        &   \slice{$T=1$} & \targ{}& \gate[style={fill=red!80}]{X}\slice{\quad $T=2$}    & \ctrl{3} &  \gate[style={fill=red!80}]{X}\slice{\quad $T=3$}  & \ctrl{4} &\gate[style={fill=red!80}]{X} \slice{\quad $T=4$} &  &\slice{\quad $T=5$}&\ctrl{5}&  \\
     Q2         & \gate[style={fill=red!80}]{X} & \ctrl{-1} & \gate[style={fill=red!80}]{X} &  &   \gate[style={fill=red!80}]{X}& &\gate[style={fill=red!80}]{X} & & & &\\
       Q3        &  &  &   &          &  & &  &  \targ{}           & & &\\
       Q4        &  &  &   & \targ{}  &  \gate[style={fill=red!80}]{X} & & \gate[style={fill=red!80}]{X} & \ctrl{-1} & & &\\
       Q5        &  &  &   &  &   & \targ{}& \gate[style={fill=red!80}]{X}& & & &\\
       Q6        &  &  &   &  &   & & &   & & \targ{} &
    \end{quantikz}
            }};
            \node[anchor=base, right=0.0cm of leftPlot4] (rightPlot4) {
            \resizebox{0.12\columnwidth}{!}{
    \begin{tikzpicture}[->,>=stealth',shorten >=1pt,auto,node distance=2.5cm,
                        thick,main node/.style={circle,draw,font=\sffamily\Large\bfseries}]
      \node[main node,fill=red] (Q1) {Q1};
      \node[main node,fill=red] (Q2) [below left of=Q1] {Q2};
      \node[main node] (Q3) [below right of=Q2] {Q3};
      \node[main node,fill=red] (Q4) [below right of=Q1] {Q4};
      \node[main node,fill=red] (Q5) [above right of=Q1] {Q5};
      \node[main node] (Q6) [above left of=Q1] {Q6};

      \path[every node/.style={font=\sffamily\small}]
        (Q1) edge node [right] {$T=2$} (Q4)
            edge node [left] {$T=3$} (Q5)
            edge node [left] {$T=5$} (Q6)
        (Q2) edge node [left] {$T=1$} (Q1)
        (Q4) edge node [left] {$T=4$} (Q3);
    \end{tikzpicture}
            }};
        \end{tikzpicture}
    };

    \node[draw, thick, inner sep=0.05cm, right=0.5cm of box4] (box5) {
        \begin{tikzpicture}
            \node[anchor=base] (leftPlot5) {
            \resizebox{0.17\columnwidth}{!}{
    \begin{quantikz}
     Q1        &   \slice{$T=1$} & \targ{}& \gate[style={fill=red!80}]{X}\slice{\quad $T=2$}    & \ctrl{3} &  \gate[style={fill=red!80}]{X}\slice{\quad $T=3$}  & \ctrl{4} &\gate[style={fill=red!80}]{X} \slice{\quad $T=4$} &   &\gate[style={fill=red!80}]{X}\slice{\quad $T=5$}&\ctrl{5}&  \\
     Q2         & \gate[style={fill=red!80}]{X} & \ctrl{-1} & \gate[style={fill=red!80}]{X} &  &   \gate[style={fill=red!80}]{X}& &\gate[style={fill=red!80}]{X} & & \gate[style={fill=red!80}]{X}& &\\
       Q3        &  &  &   &          &  & &  &  \targ{}  &  \gate[style={fill=red!80}]{X}& &\\
       Q4        &  &  &   & \targ{}  &  \gate[style={fill=red!80}]{X} & & \gate[style={fill=red!80}]{X} & \ctrl{-1} & \gate[style={fill=red!80}]{X} & &\\
       Q5        &  &  &   &  &   & \targ{}& \gate[style={fill=red!80}]{X}& & \gate[style={fill=red!80}]{X}& &\\
       Q6        &  &  &   &  &   & & &   & & \targ{} &
    \end{quantikz}
            }};
            \node[anchor=base, right=0.0cm of leftPlot5] (rightPlot5) {
            \resizebox{0.12\columnwidth}{!}{
    \begin{tikzpicture}[->,>=stealth',shorten >=1pt,auto,node distance=2.5cm,
                        thick,main node/.style={circle,draw,font=\sffamily\Large\bfseries}]
      \node[main node,fill=red] (Q1) {Q1};
      \node[main node,fill=red] (Q2) [below left of=Q1] {Q2};
      \node[main node,fill=red] (Q3) [below right of=Q2] {Q3};
      \node[main node,fill=red] (Q4) [below right of=Q1] {Q4};
      \node[main node,fill=red] (Q5) [above right of=Q1] {Q5};
      \node[main node] (Q6) [above left of=Q1] {Q6};

      \path[every node/.style={font=\sffamily\small}]
        (Q1) edge node [right] {$T=2$} (Q4)
            edge node [left] {$T=3$} (Q5)
            edge node [left] {$T=5$} (Q6)
        (Q2) edge node [left] {$T=1$} (Q1)
        (Q4) edge node [left] {$T=4$} (Q3);
    \end{tikzpicture}
            }};
        \end{tikzpicture}
    };

    \node[draw, thick, inner sep=0.05cm, right=0.5cm of box5] (box6) {
        \begin{tikzpicture}
            \node[anchor=base] (leftPlot6) {
            \resizebox{0.17\columnwidth}{!}{
    \begin{quantikz}
     Q1        &   \slice{$T=1$} & \targ{}& \gate[style={fill=red!80}]{X}\slice{\quad $T=2$}    & \ctrl{3} &  \gate[style={fill=red!80}]{X}\slice{\quad $T=3$}  & \ctrl{4} &\gate[style={fill=red!80}]{X} \slice{\quad $T=4$} &   &\gate[style={fill=red!80}]{X}\slice{\quad $T=5$}&\ctrl{5}& \gate[style={fill=red!80}]{X} \\
     Q2         & \gate[style={fill=red!80}]{X} & \ctrl{-1} & \gate[style={fill=red!80}]{X} &  &   \gate[style={fill=red!80}]{X}& &\gate[style={fill=red!80}]{X} & & \gate[style={fill=red!80}]{X}& &\gate[style={fill=red!80}]{X}\\
       Q3        &  &  &   &          &  & &  &  \targ{}  &  \gate[style={fill=red!80}]{X}& &\gate[style={fill=red!80}]{X}\\
       Q4        &  &  &   & \targ{}  &  \gate[style={fill=red!80}]{X} & & \gate[style={fill=red!80}]{X} & \ctrl{-1} & \gate[style={fill=red!80}]{X} & &\gate[style={fill=red!80}]{X}\\
       Q5        &  &  &   &  &   & \targ{}& \gate[style={fill=red!80}]{X}& & \gate[style={fill=red!80}]{X}& & \gate[style={fill=red!80}]{X}\\
       Q6        &  &  &   &  &   & & &   & & \targ{} & \gate[style={fill=red!80}]{X}
    \end{quantikz}
            }};
            \node[anchor=base, right=0.0cm of leftPlot6] (rightPlot6) {
            \resizebox{0.12\columnwidth}{!}{
    \begin{tikzpicture}[->,>=stealth',shorten >=1pt,auto,node distance=2.5cm,
                        thick,main node/.style={circle,draw,font=\sffamily\Large\bfseries}]
      \node[main node,fill=red] (Q1) {Q1};
      \node[main node,fill=red] (Q2) [below left of=Q1] {Q2};
      \node[main node,fill=red] (Q3) [below right of=Q2] {Q3};
      \node[main node,fill=red] (Q4) [below right of=Q1] {Q4};
      \node[main node,fill=red] (Q5) [above right of=Q1] {Q5};
      \node[main node,fill=red] (Q6) [above left of=Q1] {Q6};

      \path[every node/.style={font=\sffamily\small}]
        (Q1) edge node [right] {$T=2$} (Q4)
            edge node [left] {$T=3$} (Q5)
            edge node [left] {$T=5$} (Q6)
        (Q2) edge node [left] {$T=1$} (Q1)
        (Q4) edge node [left] {$T=4$} (Q3);
    \end{tikzpicture}
            }};
        \end{tikzpicture}
    };
\end{tikzpicture}
\caption{Propagation graph of the left circuit between $T=1$ to $T=5$.}
\label{fig:randompropag}
\end{figure}
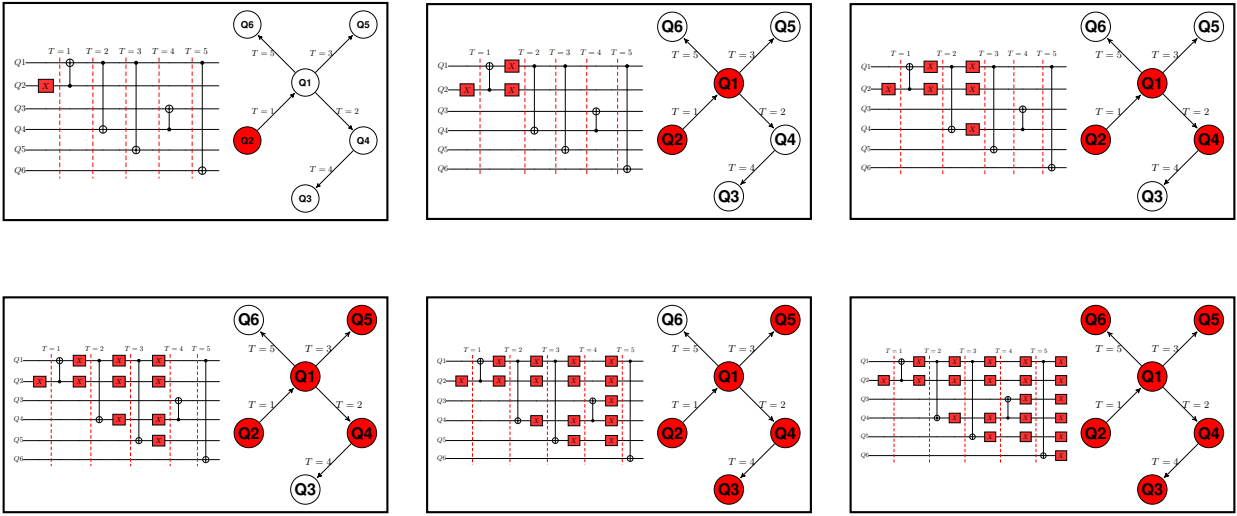

Undersanding quantum error propagation in a general quantum circuit is critical in the development of fault-tolerant quantum computer. The motivation of this work, is try to answer the following questions with respect to this problem:
\begin{enumerate}
    \item How to numerically or mathematically quantify the extent of error propagation and how to compare that quality between different circuits?
    \item What is the time complexity we need to calculate the above measurement in the general case?
    \item What is the optimal way to design a quantum error correction circuit with bounded error propagation?
    \item Can we bound the error propagation so the threshold conclusion is not affected?
\end{enumerate}

With the help of a successful model of error propagation, we might not only build a classical compiler with some abstract interpretation  \citep{abstractInter} to verify the correctness and threshold of different QEC codes, but also search for the best solution for logical gate design and QEC paradigm with classical method automatically, such as artificial intelligence  \citep{Wang_2022}  and SAT solver \citep{tan2024satscalpellatticesurgery}, with the awareness of potential error propagation.

We care about the time complexity of this problem because currently it has been estimated that $200$ million noisy qubits are required to factor a $2048$ bit RSA integer\citep{Gidney2021howtofactorbit} with surface code. We should reasonably expect that the real fault-tolerant quantum computer in the future will have a large scale of physical qubits, with extremely complicated circuit structure and topology after compilation and optimization. Thus, we could only understand the propagation of the error if the problem itself is computationally tractable.

From the view of designing better quantum error correction code, a new concept called \textbf{circuit-level errors} is invented \citep{Pryadko2020maximumlikelihood} to describe the noise model in a given quantum circuit. The \textbf{circuit-level errors} has been used in studying the optimal decoding of surface code \citep{heim2016optimalcircuitleveldecodingsurface} as well as the quantum memor\citep{Bravyi_2024}. Another very similar concept is the \textbf{detector-error model}\citep{derks2024designingfaulttolerantcircuitsusing}, where the quantum errors are analyze after propagation to the detected qubits.

There is much research about how to understand and estimate the propagation of quantum error. Mathematically, a theoretical upper bound for error propagation is provided by using the Frobenius norm \citep{yu2022analysiserrorpropagationquantum}.
A rescheduling algorithm has been proposed in which the author used gate error propagation paths to estimate the output state fidelity of
the quantum circuit \citep{saravanan2022paulierrorpropagationbasedgate}. 
The propagation of errors has also been studied in simulation of the Hubbard model \citep{flannigan2022propagationerrorsquantitativequantum}. A fault path tracer method has also been introduced \citep{faultpathtracer} to understand the propagation of errors caused by the circuit structure. Another similar idea is gate-flow error, applied in transversal CNOT gate \citep{kim2024transversalcnotgatemulticycle}.

Despite all the above efforts, there is no commonly accepted paradigm for analyzing the statistical model of quantum error. In many literatures, the problem of error propagation is just briefly mentioned without careful analysis\citep{Kubica_2018}. Also, we still lack a clear and precise statiscal measure for error propagation.

In this paper, we construct a new abstract framework for the error propagation problem, inspired by all previous work. 

In my framework, the circuit is first converted to a space-time error propagation graph, and then a binary graph which characterizes the error propagation behavior.  We use the term \textbf{shift of the error distribution} as a measure for the extent of error propagation, and apply the measure to some special cases. Further, we study the complexity of calculating the exact error distribution before and after error propagation, and prove that the problem is \textit{NP-complete} in the worse case by reducing $\textbf{MAX-XORSAT}$ to the problem itself.

\section{Assumption to simplify the problem in our model}
We assume that our circuit is only made up of gates $CX$ and that our error model is simply a single-qubit bitflip model. We discretize the calculation into $T$ time windows. The error model is determined by a constant probability $p$, described as the independent probability that a bit-flip error will occur within one window:

\begin{equation}
    \ket{\psi}  \rightarrow \begin{cases}
                & \hat{X} \ket{\psi}  \qquad \text{With probability $p$}\\
                & \ket{\psi} \qquad \text{With probability $1-p$}
    \end{cases}
    \label{eq:bitflipmodel}
\end{equation}

We will further show that the method developed for the circuit and error model under the assumption is actually adaptive to the general case. The dynamic of all kinds of error propagation is similar to the propagation of this simple case. And because of the linearity of quantum mechanics, the actual noise model can be described by a linear combination of error models isomorphism with the simplest bitflip model.

A general quantum model acting on density matrix $\rho$ can be described by a quantum noise channel $\mathcal{E}$ with a set of Kraus operators $\{\mathcal{K}_i \}$ that form a complete positive map:
\begin{equation}
    \mathcal{E}(\hat{\rho})=\sum_{i=1}^n \hat{\mathcal{K}}_i \hat{\rho} \hat{\mathcal{K}}_i^\dagger, \qquad \sum_{i=1}^n  \hat{\mathcal{K}}_i \hat{\mathcal{K}}_i^\dagger=\hat{I}
\end{equation}

\section{Propagation of quantum errors}

The simplest example of error propagation is that $CX$ gate can propagate the bitflip error on the control qubit.

\begin{figure}[h!]
\centering
\fbox{
   \begin{minipage}{0.4\textwidth}
     \centering
\resizebox{\columnwidth}{!}
{
\begin{quantikz}
 \ket{0}        & \gate[style={fill=red!80}]{X}\slice{$T=1$} & \ctrl{1}  & \slice{$T=2$}    &   \\
 \ket{0}        &  & \targ{}   &  &  
\end{quantikz}
$\Rightarrow$
\begin{quantikz}
 \ket{0}        &  \slice{$T=1$}  & \ctrl{1}  & \slice{$T=2$} &\gate[style={fill=red!80}]{X}  &    \\
 \ket{0}        &  & \targ{}  & & \gate[style={fill=red!80}]{X} &  
\end{quantikz}
}
\end{minipage}
}
\fbox{
   \begin{minipage}{0.4\textwidth}
     \centering
\resizebox{\columnwidth}{!}
{
\begin{quantikz}
 \ket{0}        & \gate[style={fill=blue!80}]{Z}\slice{$T=1$} & \ctrl{1}  & \slice{$T=2$}    &   \\
 \ket{0}        &  & \targ{}   &  &  
\end{quantikz}
$\Rightarrow$
\begin{quantikz}
 \ket{0}        &  \slice{$T=1$}  & \ctrl{1}  & \slice{$T=2$} &\gate[style={fill=blue!80}]{Z}  &    \\
 \ket{0}        &  & \targ{}  & & &  
\end{quantikz}
}
\end{minipage}
}
\fbox{
   \begin{minipage}{0.4\textwidth}
     \centering
\resizebox{\columnwidth}{!}
{
\begin{quantikz}
 \ket{0}        & & \ctrl{1}  & \slice{$T=2$}    &   \\
 \ket{0}        & \gate[style={fill=blue!80}]{Z}\slice{$T=1$}   & \targ{}   &  &  
\end{quantikz}
$\Rightarrow$
\begin{quantikz}
 \ket{0}        &  \slice{$T=1$}  & \ctrl{1}  & \slice{$T=2$} &\gate[style={fill=blue!80}]{Z}  &    \\
 \ket{0}        &  & \targ{}  & & \gate[style={fill=blue!80}]{Z} &  
\end{quantikz}
}
\end{minipage}
}
\fbox{
   \begin{minipage}{0.4\textwidth}
     \centering
\resizebox{\columnwidth}{!}
{
\begin{quantikz}
 \ket{0}        & \slice{$T=1$} & \ctrl{1}  & \slice{$T=2$}    &   \\
 \ket{0}        & \gate[style={fill=red!80}]{X}  & \targ{}   &  &  
\end{quantikz}
$\Rightarrow$
\begin{quantikz}
 \ket{0}        &  \slice{$T=1$}  & \ctrl{1}  & \slice{$T=2$} &   &    \\
 \ket{0}        &  & \targ{}  & & \gate[style={fill=red!80}]{X} &  
\end{quantikz}
}
\end{minipage}
}
\caption{Propagation of a bitflip $X$ error by $CX$ gate. At time step $T=1$, one bitflip error occurs at the first qubit. The error propagate through $CX$ gate and becomes $2$ bitflip errors. }
\label{fig:ErrorPropagationRule}
\end{figure}
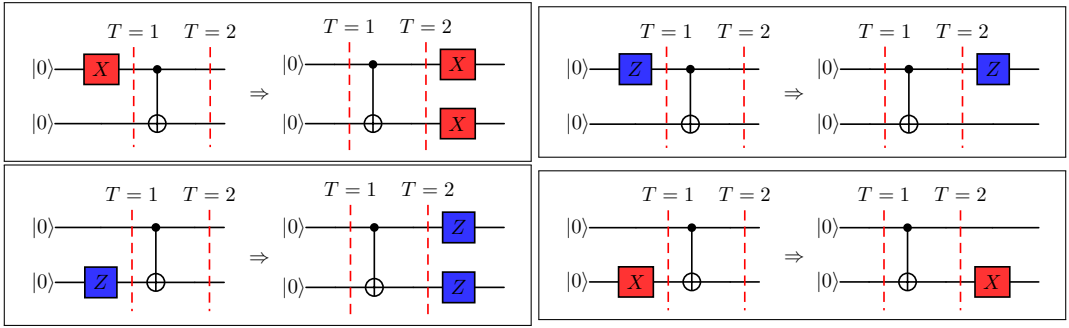

In the two qubit case with $T=2$, there are $2^4=16$ possible arrangement of errors. We can calculate the probability of each case. We use $P(k)$ to denote the probability of $k$ bitflip errors. $P(0),P(1),P(2)$ are shown in \autoref{eq:pCnot2q}.

\begin{equation}
    \begin{split}
           P(0)&=(1-p)^4+p^2(1-p)^2+2 \times p^3(1-p) \\
           P(1)&=3p(1-p)^3+3p^2(1-p)^2+p^3(1-p)+p^4\\
           P(2)&=p(1-p)^3+2p^2(1-p)^2+p^3(1-p)
    \end{split}
    \label{eq:pCnot2q}
\end{equation}

On the other hand, we can also calculate the error distribution without the propagation of $CNOT$ gate. Since the total time step is $T=2$, an error happens at $T=2$ only if one error happens within one of the time windows.

\begin{equation}
        \begin{split}
             P(0)&=(2p^2-2p+1)^2\\
             P(1)&=2\times(2p-2p^2)(2p^2-2p+1)\\
             P(2)&=(2p-2p^2)^2
        \end{split}
        \label{eq:p2q}
\end{equation}

In the general case, the error model after propagation can be quite complicated.  Now let's do a simple calculation comparing the error number distribution before and after the propagation.

\subsection{General model for pauli error propagation by clifford circuit}

In QEC code, we build circuit mostly by Clifford gates. In this section, we study the general model of pauli error propagation by Clifford gates.

For a pauli error $\hat{\sigma}$ and a Clifford gate $\hat{G}$, by Gottesman-Knill theorem, the conjugate mapping of $\hat{G}$ transform the pauli error $\hat{\sigma}$ to a new pauli error $\hat{\sigma}'$: 
\begin{equation}
   \hat{G}^{-1} \hat{\sigma}\hat{G}= \hat{\sigma}'
   \label{eq:gottesman}
\end{equation}
After multiply $\hat{G}$ on the left on both sides of the equation, we get:
\begin{equation}
    \hat{\sigma}\hat{G}=\hat{G}\hat{\sigma}'
    \label{eq:errorprop}
\end{equation}
\autoref{eq:errorprop} implies that the pauli error propagation is exactly the same as conjugate mapping by Clifford gates.

The conjugative operation of hadamard gate can propagate $\hat{X}$, $\hat{Y}$ and $\hat{Z}$ as:
\begin{equation}
    \hat{H}\hat{X}\hat{H}=\hat{Z} , \qquad \hat{H}\hat{Y}\hat{H}=-\hat{Y},\qquad \hat{H}\hat{Z}\hat{H}=\hat{X}
    \label{eq:Basischange}
\end{equation}
The $\hat{X}$, $\hat{Y}$,$\hat{Z}$ gates under the conjugate transformation of phase gate $\hat{P}$ is given as
\begin{equation}
    \hat{P}\hat{X}\hat{P}^\dagger=\hat{Y} , \qquad \hat{P}\hat{Y}\hat{P}^\dagger=-\hat{X} ,\qquad \hat{P}\hat{Z}\hat{P}^\dagger=\hat{Z}
    \label{eq:phaseconj}
\end{equation}

$CNOT$ gate propagate all two qubit pauli error as: 

\begin{equation}
\begin{split}
  \hat{CNOT}_{i,j}(X_i \otimes I_j)  \hat{CNOT}_{i,j}&= \hat{X}_i \otimes \hat{X}_j\\
  \hat{CNOT}_{i,j}(\hat{I}_i \otimes \hat{X}_j)  \hat{CNOT}_{i,j}&= \hat{I}_i \otimes \hat{X}_j\\
  \hat{CNOT}_{i,j}(\hat{X}_i \otimes \hat{X}_j)  \hat{CNOT}_{i,j}&=\hat{X}_i \otimes \hat{I}_j\\
  \hat{CNOT}_{i,j}(\hat{Z}_i \otimes \hat{I}_j)  \hat{CNOT}_{i,j}&= \hat{Z}_i \otimes \hat{I}_j\\
  \hat{CNOT}_{i,j}(\hat{I}_i \otimes \hat{Z}_j)  \hat{CNOT}_{i,j}&= \hat{Z}_i \otimes \hat{Z}_j\\
  \hat{CNOT}_{i,j}(\hat{Z}_i \otimes \hat{Z}_j)  \hat{CNOT}_{i,j}&=\hat{I}_i \otimes \hat{Z}_j\\
  \hat{CNOT}_{i,j}(\hat{Y}_i \otimes \hat{I}_j)  \hat{CNOT}_{i,j}&= \hat{Y}_i \otimes \hat{X}_j \\
  \hat{CNOT}_{i,j}(\hat{I}_i \otimes \hat{Y}_j)  \hat{CNOT}_{i,j}&=\hat{Z}_i \otimes \hat{Y}_j\\
  \hat{CNOT}_{i,j}(\hat{Y}_i \otimes \hat{Y}_j)  \hat{CNOT}_{i,j}&=-\hat{X}_i \otimes \hat{Z}_j\\ 
  \hat{CNOT}_{i,j}(\hat{X}_i \otimes \hat{Y}_j)  \hat{CNOT}_{i,j}&=\hat{Y}_i \otimes \hat{Z}_j\\
  \hat{CNOT}_{i,j}(\hat{Y}_i \otimes \hat{X}_j)  \hat{CNOT}_{i,j}&=\hat{Y}_i \otimes \hat{I}_j\\
  \hat{CNOT}_{i,j}(\hat{X}_i \otimes \hat{Z}_j)  \hat{CNOT}_{i,j}&=-\hat{Y}_i \otimes \hat{Y}_j   \\
  \hat{CNOT}_{i,j}(\hat{Z}_i \otimes \hat{X}_j)  \hat{CNOT}_{i,j}&=\hat{Z}_i \otimes \hat{X}_j\\
  \hat{CNOT}_{i,j}(\hat{Y}_i \otimes \hat{Z}_j)  \hat{CNOT}_{i,j}&=\hat{X}_i \otimes \hat{Y}_j\\
  \hat{CNOT}_{i,j}(\hat{Z}_i \otimes \hat{Y}_j)  \hat{CNOT}_{i,j}&=\hat{I}_i \otimes \hat{Y}_j \\
\end{split} 
\label{eq:cnotconj}
\end{equation} 

This simulation of pauli error propagation can be done by tableau simulation first proposed by Scott Aaronson.
Here we only care about the propagation of $Z$ and $X$ error.

\section{Statement of the computational problem}

The problem of quantum error distribution under the previous assumption goes as follows:
\textbf{Given a quantum circuit of $n$ qubit composed of only $m$ $CX$ gates with $T$ time windows and the error model with only single bitflip errors of bitflip probability $p$, what is the distribution of the probability of single qubit error after propagation?}

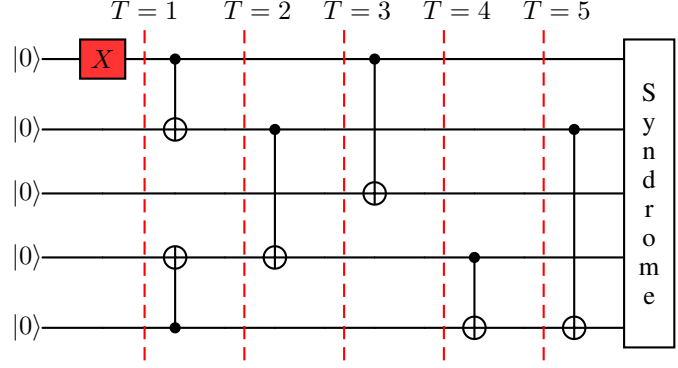
\begin{figure}[h!]
\centering
\begin{quantikz}
 \ket{0}        & \gate[style={fill=red!80}]{X}\slice{$T=1$} & \ctrl{1}& \slice{\quad $T=2$}    &  & \slice{\quad $T=3$}  & \ctrl{2} & \slice{\quad $T=4$} &  &\slice{\quad $T=5$}&& \gate[5,disable auto
height]{\verticaltext{Syndrome}} \\
 \ket{0}        &  & \targ{} &  & \ctrl{2} &  & & & & & \ctrl{3}&\\
  \ket{0}        &  &  &   &  &  & \targ{} & & & & &\\
   \ket{0}        &  & \targ{}&    & \targ{}  &  & &  & \ctrl{1} & & &\\
   \ket{0}        &  & \ctrl{-1}&    &  &   & & & \targ{} & & \targ{}&
\end{quantikz}
\caption{General case of error propagation. We want to understand the effective error model after the propagation, because QEC code only guarantee to decode and correct the errors when the syndromes are measured.}
\label{fig:ErrorPropagationGeneral}
\end{figure}

The answer to the problem is important to error threshold analysis. Because in the fault tolerant quantum computing, we design the QEC code which should be able to correct all error patterns where the total number of single qubit errors where we run the decoding should be smaller than the code distance.

The simulation of error propagation once we fixed the circuit and the errors is trivial. 

However, the calculation of probability is computationally hard, in the general case. Because the probability space has size $2^{nT}$, and we have to simulate the propagation for each independent case.

\section{Error distribution analysis for some special cases}

In this section, we analyze some special cases with simple structure for error distribution. 

Now we use the notation $Q_k[t]$ for the qubit index $k$ in the time window $t$. The error distribution can be visualized by such as diagram that is composed of all vertices $Q_k[t]$ for $1 \leq  k \leq n$ and $1 \leq t \leq T$. There is a directed edge $(Q_i[t_1],Q_j[t_2])$ if an only if a random bitflip error at at qubit $Q_i$ in time window $t_1$ can propagate to qubit $Q_j$ through some directed path. An example of the directed graph is shown in \autoref{fig:propagateDiagram}.

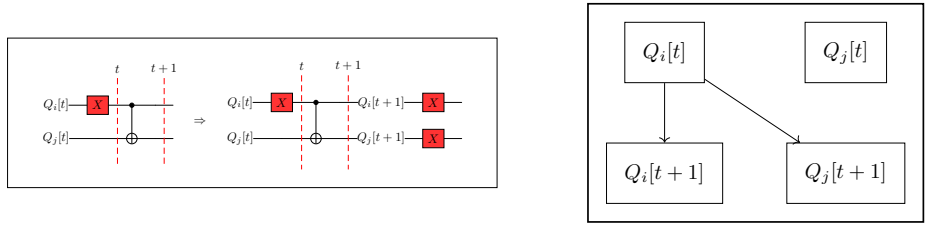
\begin{figure}[h!]
\centering

\begin{minipage}{0.4\textwidth}
     \centering
\resizebox{\columnwidth}{!}
{
\begin{tikzpicture}
 \node[draw, thick, inner sep=0.5cm] (firstBox) {
\begin{quantikz}
 Q_i[t]        & \gate[style={fill=red!80}]{X}\slice{$t$} & \ctrl{1}  & \slice{$t+1$}    &   \\
  Q_j[t]        &  & \targ{}   &  &  
\end{quantikz}
$\Rightarrow$
\begin{quantikz}
 Q_i[t]        & \gate[style={fill=red!80}]{X} \slice{$t$}  & \ctrl{1}  & \slice{$t+1$}  &Q_i[t+1]&\gate[style={fill=red!80}]{X}  &     \\
Q_j[t]        &  & \targ{}  &  & Q_j[t+1]&\gate[style={fill=red!80}]{X} &  
\end{quantikz}
};
\end{tikzpicture}
}
\end{minipage}
\begin{minipage}{0.4\textwidth}
     \centering
\resizebox{0.7\columnwidth}{!}
{
    \begin{tikzpicture}
        \node[draw, thick, inner sep=0.3cm] (thirdBox) {
            \begin{tikzpicture}
                \tikzstyle{every node}=[draw, minimum size=1cm, inner sep=0.3cm]

                \node (Q1) at (0,2) {$Q_i[t]$};
                \node (Q2) at (3,2) {$Q_j[t]$};
                \node (Q3) at (0,0) {$Q_i[t+1]$};
                \node (Q4) at (3,0) {$Q_j[t+1]$};

                \draw[->] (Q1) -- (Q3);
                \draw[->] (Q1) -- (Q4);
            \end{tikzpicture}
        };
    \end{tikzpicture}
}
\end{minipage}
\caption{Illustration of how to convert the propagation of bitflip error to a new directed graph.}
\label{fig:propagateDiagram}
\end{figure}

The input of the problem is a quantum circuit $\textbf{C}$ with $n$ qubits composed of $CNOT$ gates, the bitflip error probability $p$ in one gate cycle. The output is the estimate distribution of the number of bitflip error after the circuit is performed.

\subsection{Error distribution when $T=1$}

The distribution of error when $T=1$ satisfies a binomial distribution. Once we determine the number of error $k$, the number of possible error pattern is the combinatorial value $\binom{n}{k}=\frac{n!}{k!(n-k)!}$, with probability $p^k(1-p)^{n-k}$ for each case, as summarized in \autoref{tab:errorDistribution}. This is exactly a binomial distribution. When $n \rightarrow \infty$, by large number theorem, the distribution is close to a Gaussian distribution.

\begin{table}[h!]
    \centering
    \begin{tabular}{|c|c|c|}
        \hline
        Number of error & Number of pattern & Probability\\
          \hline
       $0$  & $1$ & $(1-p)^n$\\
         \hline        
       $1$  & $n$ & $np(1-p)^{n-1}$\\
          \hline
       $2$  & $\binom{n}{2}$ & $\binom{n}{2}p^2(1-p)^{n-2}$\\
         \hline
       $\cdots$  & $\cdots$ & $\cdots$\\
         \hline
       $k$  & $\binom{n}{k}$  & $\binom{n}{k}p^k(1-p)^{n-k}$ \\
         \hline
      $\cdots$  & $\cdots$ & $\cdots$\\
      \hline
       $n$  & $1$ & $p^n$\\
       \hline
    \end{tabular}
    \caption{The distribution of the probability of error number form a binomial distribution.}
    \label{tab:errorDistribution}
\end{table}

The expectation of the number of error is:
\begin{equation}
    E(k)=pn
\end{equation}

\FloatBarrier

\subsection{Error distribution without any $CX$ gates}

Another naive case is analyzing the error distribution when the circuit is empty without any structure. Unlike in the previous case where $T=1$, now the bit-flip errors can accumulate over time and sometimes cancel each other out by the rule $\hat{X}\hat{X}=\hat{I}$. In fact, at time $T$, a bit-flip error still exists only when the qubit is flipped for an odd number of times, as given in \autoref{eq:PoddPeven}.

\begin{figure}[h!]
\centering
\begin{minipage}{0.48\textwidth}
     \centering
\resizebox{0.8\columnwidth}{!}
{
\begin{quantikz}
 Q_1[1]        & \gate[style={fill=red!80}]{X}\slice{$T=1$} & & \slice{\quad $T=2$}    &  & \slice{\quad $T=3$}  & & \slice{\quad $T=4$} &  &\slice{\quad $T=5$}& &\\
  Q_2[1]        &   &  &  &\gate[style={fill=red!80}]{X} &  &\gate[style={fill=red!80}]{X} & & & & &\\
  Q_3[1]         & \gate[style={fill=red!80}]{X} &  &   &  &  & & & & & &\\
   Q_4[1]        &  &\gate[style={fill=red!80}]{X}  &   &   &  &\gate[style={fill=red!80}]{X} &  & & & &\\
   Q_5[1]        &  &\gate[style={fill=red!80}]{X} &    &  &   & & & \gate[style={fill=red!80}]{X}& & &
\end{quantikz}
}
\end{minipage}
\begin{minipage}{0.5\textwidth}
     \centering
    \includegraphics[width=\linewidth]{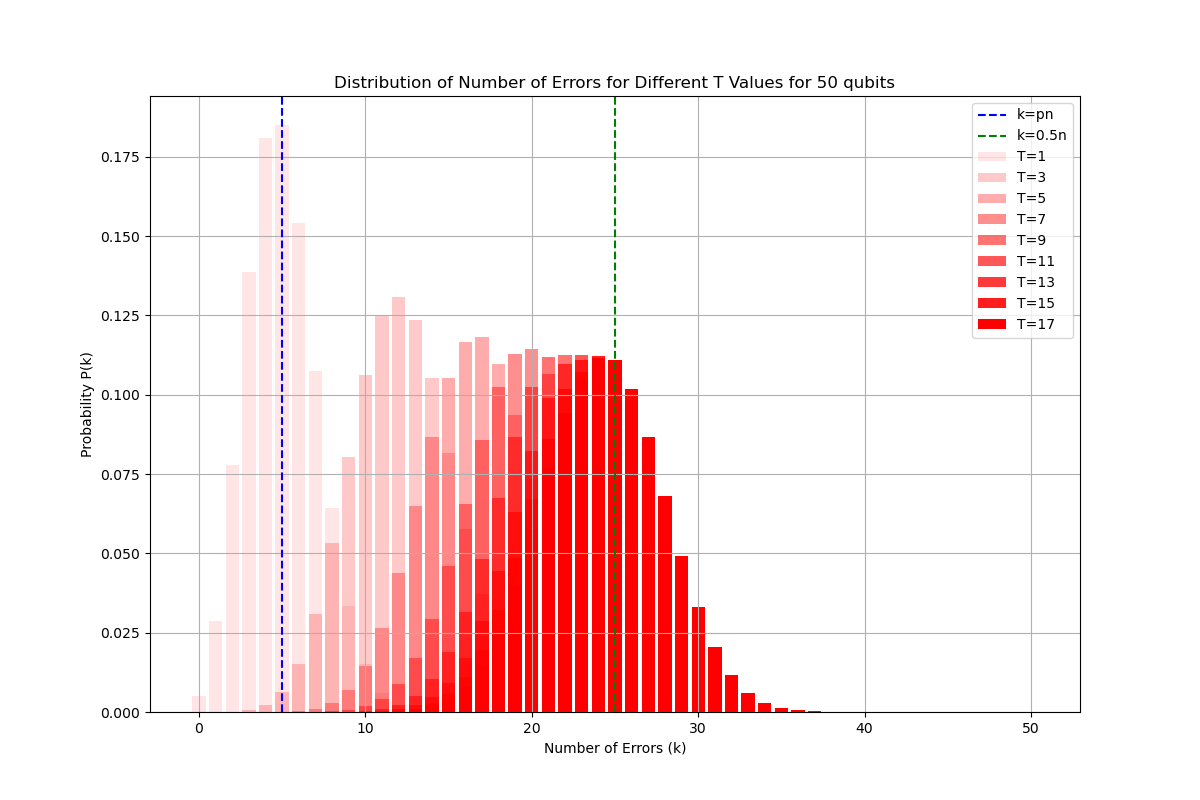}
\end{minipage}
    \caption{The shift of error distribution of $50$ qubits with increasing $T$. When $T$ is small, the expectation value is roughly $E(k)=pn$. However, when $T$ increase, the number of error is dominated by the random parity in \autoref{eq:PoddPeven}, but not the error probability $p$ anymore, the distribution converge to a distribution with expectation value $E(k)=\frac{n}{2}$.}
    \label{fig:50QubitNoCNOT}
\end{figure}

\begin{equation}
        P(\text{even})=\sum_{k=0}^{\lfloor T/2 \rfloor} \binom{T}{2k}p^{2k}(1-p)^{T-2k},\quad
        P(\text{odd})=\sum_{k=0}^{\lfloor T/2 \rfloor} \binom{T}{2k+1}p^{2k+1}(1-p)^{T-2k-1}\\        
\label{eq:PoddPeven}
\end{equation}

The distribution of error number, is then
\begin{equation}
    P(k)=\binom{n}{k} P(\text{even})^{n-k}P(\text{odd})^k
\end{equation}

A simulation of the distribution is shown in \autoref{fig:50QubitNoCNOT}.

\FloatBarrier

\subsection{Error distribution for transversal CNOT gate set}

In a circuit of $n=2r$ qubits with only transversal CNOT gate set satisfies the propagation rule.

\begin{equation}
      Q_k[k]  \xlongrightarrow{\text{Propagate to}} Q_{k+r}[k+1],  \quad \forall  1 \leq k \leq r
      \label{eq:TransversalCNOTrules}
\end{equation}

The propagation graph of \autoref{eq:TransversalCNOTrules} is shown in \autoref{fig:trasverseCNOT}.  The error distributions of the first $r=\frac{n}{2}$ qubits are not affected by the circuit, which still obey \autoref{eq:PoddPeven}. Every qubit $Q_{k+r}$, on the other hand, is only affected by $Q_{k}$. As a result, we can separate our analysis of the error distribution into  different pair of qubits $(Q_k,Q_{k+r})$. The locations of the random errors for each pair $(Q_k,Q_{k+r})$ are divided into three areas:
\begin{enumerate}
    \item The errors occur on qubit $Q_{k+r}$ for $1\leq t\leq T$. We call this area Area 1.
    \item The errors occur on qubit $Q_k$ for $k+1 \leq t \leq n$.We call this area Area 2.
    \item The errors occur on qubit $Q_k$ for $1 \leq t \leq k$.We call this area Area 3.
\end{enumerate}
In the first two areas, errors occur independently on each qubit and do not affect the probability of the other. In the third area, the error can propagate from $Q_k$ to $Q_{k+r}$ and affect the result of $Q_{k+r}$. The $8$ possible result of the error parity of the three areas determine the final error numbers.

\begin{table}[h!]
    \centering
    \begin{tabular}{|c|c|c|c|c|}
         \hline
   Index  & Area 1    &   Area 2 & Area 3 & \#Errors \\
         \hline
   0&  Even    & Even  &  Even & 0\\
         \hline
   1&  Even     &  Even &  Odd & 2\\
         \hline
   2&   Even    & Odd & Even & 1\\
         \hline
   3& Even     & Odd & Odd & 1\\
         \hline
   4& Odd     & Even & Even & 1\\
         \hline
   5&  Odd    & Even & Odd & 1\\
         \hline
   6&  Odd    & Odd & Even & 2\\
         \hline
   7&  Odd    & Odd & Odd & 0\\
         \hline
    \end{tabular}
    \caption{The parity of errors in three different areas and the resultant error number for a qubit pair $(Q_k,Q_{k+r})$ in a transversal circuit.}
    \label{tab:transversalParity}
\end{table}

For Area 1:

\begin{equation}
    P_1(\text{Even})=\sum_{i=0}^{\lfloor T/2 \rfloor} \binom{T}{2i}p^{2i}(1-p)^{T-2i},\quad
    P_1(\text{Odd})=\sum_{i=0}^{\lfloor T/2 \rfloor} \binom{T}{2i+1}p^{2i+1}(1-p)^{T-2i-1}
\end{equation}

For Area 2:

\begin{equation}
    P_2(\text{Even})=\sum_{i=0}^{\lfloor (n-k)/2 \rfloor} \binom{n-k}{2i}p^{2i}(1-p)^{n-k-2i},\quad
    P_2(\text{Odd})=\sum_{i=0}^{\lfloor (n-k)/2 \rfloor} \binom{n-k}{2i+1}p^{2i+1}(1-p)^{n-k-2i-1}
\end{equation}

For Area 3:

\begin{equation}
    P_3(\text{Even})=\sum_{i=0}^{\lfloor k/2 \rfloor} \binom{k}{2i}p^{2i}(1-p)^{k-2i},\quad
    P_3(\text{Odd})=\sum_{i=0}^{\lfloor k/2 \rfloor} \binom{k}{2i+1}p^{2i+1}(1-p)^{k-2i-1}
\end{equation}

The final distribution on qubit $Q_{k}$ and $Q_{k+r}$ is given by:

\begin{equation}
\begin{split}
P(0)&=P_1(\text{Even})P_2(\text{Even})P_3(\text{Even})+P_1(\text{Odd})P_2(\text{Odd})P_3(\text{Odd})\\
P(1)&=P_1(\text{Even})P_2(\text{Odd})P_3(\text{Even})+P_1(\text{Even})P_2(\text{Odd})P_3(\text{Odd})\\
&+P_1(\text{Odd})P_2(\text{Even})P_3(\text{Even})+P_1(\text{Odd})P_2(\text{Even})P_3(\text{Odd})\\
P(2)&=P_1(\text{Even})P_2(\text{Even})P_3(\text{Odd})+P_1(\text{Odd})P_2(\text{Odd})P_3(\text{Even})
\end{split}
\label{eq:transverbasicProb}
\end{equation}

For the whole circuit , we use $a_1,a_2,\cdots a_r$ as the variables representing the number of bitflip errors occur on qubit pairs $(Q_1,Q_{1+r}),(Q_2,Q_{2+r}),(Q_3,Q_{3+r}),\cdots, (Q_r,Q_{2r})$. For any possible pattern of the final error syndrome with error number $K$, the variables satisfy  
\begin{equation}
\begin{cases}
        &a_1 + a_2 + \cdots a_r=k \\ 
        & a_i \in \{0,1,2 \}, \quad  \forall 1 \leq i \leq r
\end{cases}
\label{eq:Probability}
\end{equation}
We use notation $\overline{P}(r,k)$ as the probability of the condition in \autoref{eq:Probability} to be true, which has the folloing induction formula:

\begin{equation}
\begin{split}
    \overline{P}(1,k)&= \begin{cases}
        &  P(k) \qquad   0 \leq k \leq 2\\
        &  0    \qquad     3  \leq k \leq 2r
    \end{cases}\\
    \overline{P}(r,k)&=\overline{P}(r-1,k)P(0)+\overline{P}(r-1,k-1)P(1)+\overline{P}(r-1,k-2)P(2)
\end{split}
\label{eq:transversalInduction}
\end{equation}

By \autoref{eq:transversalInduction}, we can easily get an algorithm for calculating dynamic programming $\overline{P}(r,k)$ with time complexity $O(rk)$ and space complexity $O(r)$. The pseudocode for the algorithm goes as follows:

\begin{algorithm}
\flushleft
  \caption{Algorithm to calculate the error number distribution for transversal CNOT circuit}
  \label{alg:transversalCNOT}
  \begin{algorithmic}[1] 
    \State Initialize array $P[k]$ for $0 \leq  k \leq 2$  \Comment{Store the basic probability with calculated equation.}
    \State Initialize array $X[k]$ for $0 \leq  k \leq 2r$  \Comment{Store the error distribution of previous step}
    \State Initialize array $Y[k]$ for $0 \leq  k \leq 2r$  \Comment{Store the error distribution of current step}    
    \For{$k \in 0,1,2$}
        \State $P[k] \gets \text{Result of \autoref{eq:transverbasicProb}}$    \Comment{Initialize $X[k]$ according to \autoref{eq:transversalInduction}}
    \EndFor
    \For{$k \in 0,1,2$}
        \State $X[k] \gets P[k]$    \Comment{Initialize $X[k]$ according to \autoref{eq:transversalInduction}}
    \EndFor
    \For{$k \in 3,4,5,\cdots,2r$}
        \State $X[k] \gets 0$  \Comment{Initialize $X[k]$ according to \autoref{eq:transversalInduction}}
    \EndFor

    \For{$i \in 2,3,4,r$}
        \For{$j \in 0,1,2,\cdots,2r$}
            \State $Y[j]=X[j]P[0]+X[j-1]P[1]+X[j-2]P[2]$ \Comment{Update $Y[k]$ according to \autoref{eq:transversalInduction}}
        \EndFor
        \For{$j \in 0,1,2,\cdots,2r$}
            \State $X[j]=Y[j]$  \Comment{Reset $X[k]$ as the previous $Y[k]$.}
        \EndFor        
    \EndFor
  \end{algorithmic}
\end{algorithm}

A simulation result of the error distribution for $1000$ qubits with $p=0.001$ is provided in \autoref{fig:distributionShift}. This clear evidence of increased number of errors demonstrates the negative effect of error propagations in general circuit.

\begin{figure}[h!]
\centering
\begin{minipage}{0.4\textwidth}
     \centering
\resizebox{0.8\columnwidth}{!}
{
\begin{quantikz}
 \ket{0}        & \gate[style={fill=red!80}]{X}\slice{$T=1$} & \ctrl{3}& \slice{\quad $T=2$}    &  & \slice{\quad $T=3$}  &   & \slice{\quad $T=4$} &  &&&  \\
 \ket{0}        &  & &  & \ctrl{3} &  & & & & & \\
 \ket{0}       &   &         &  &   &  & \ctrl{3} & & & & \\
 \ket{0}        &  & \targ{} &  &          & & &  &  & & \\
 \ket{0}        &  &  &      &\targ{}   &  & & &  & &   \\
 \ket{0}        &  &  &      &          &  & \targ{}& &    & &    
\end{quantikz}
}
\end{minipage}
\begin{minipage}{0.4\textwidth}
     \centering
    \includegraphics[width=\linewidth]{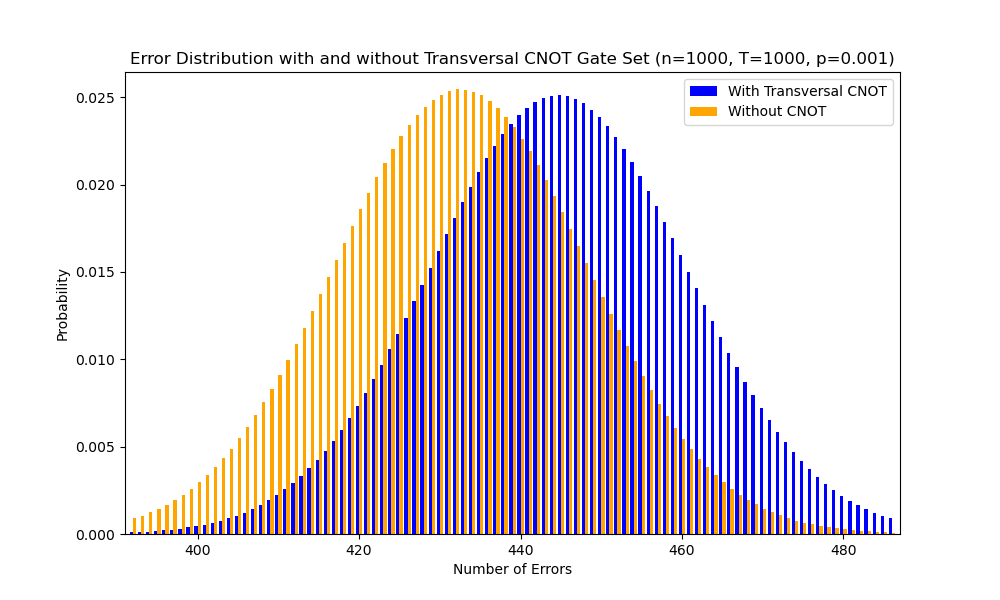}
\end{minipage}
    \caption{The simulation result of comparing the distribution of the number of error with and without transversal $CNOT$ gate. It is obvious that the distribution is shifted by the transversal $CNOT$ gate.}
\label{fig:trasverseCNOT}
\end{figure}

\FloatBarrier

\section{Exact solution based on propagation graph}

We use the symbol $\mathbf{C}$ for the circuit we are considering. Any $CNOT$ gate in the circuit is denoted as $CNOT(i,j)[t]$, where $i$ is the control qubit index and $j$ is the target qubit index and $t$ is the time index, denoting when it is acted on the circuit.

\begin{definition}(Error propagation space-time graph(\textbf{EPSTG}))
    \label{df:propgraph}
    Given an $n$ qubit quantum circuit $\mathbf{C}$ composed of $m$ $CNOT$ gates with total time windows $T$, the bit-flip probability $p$ in each time window,  we define the propagation graph time space graph $\mathbf{G}(\mathbf{C},n,m,T)[t_1,t_2,p], 1 \leq t_1 < t_2 \leq T$, which is a directed graph defined as: \begin{enumerate}
        \item A set of vertices  $\mathcal{V}(\mathbf{C},n,m,T)[t_1,t_2]=\{Q_k[t] | 1 \leq k \leq n, t_1 \leq t < t_2 \}$, each assigned to a qubit with the same index at time $t$. There are $n(t_2-t_1)$ vertices in total.
        \item A directed edge set $\mathcal{E}(\mathbf{C},n,m,T)[t_1,t_2]=\{(Q_i[t],Q_j[t+1]) \quad | \quad \exists t_1 \leq t < t_2 , s.t. \text{Either } i=j \quad \text{or} \quad CNOT(i,j)[t] \in \mathbf{C} \}$. The direction of $(Q_i[t],Q_j[t+1])$ is from $Q_i$ to $Q_j$. 
    \end{enumerate}
\end{definition}

\begin{figure}[!htb]
\fbox{
   \begin{minipage}{0.48\textwidth}
     \centering
\resizebox{0.7\columnwidth}{!}
{
\begin{quantikz}
  Q_1        & \slice{$T=1$} & \ctrl{3}& \slice{\quad $T=2$}    &  & \slice{\quad $T=3$}  &   & \slice{\quad $T=4$} &  &&  \\
 Q_2         &  & &  & \ctrl{3} &  & & & & & \\
 Q_3       &   &         &  &   &  & \ctrl{3} & & & & \\
  Q_4        &  & \targ{} &  &          & & &  &  & & \\
  Q_5        &  &  &      &\targ{}   &  & & &  & &   \\
  Q_6       &  &  &      &          &  & \targ{}& &    & &    
\end{quantikz}
}
     \caption{An example of six qubit transversal circuit $\textbf{C}$. Now the set of vertices is $\mathcal{V}(\mathbf{C},n,m,T)[t_1,t_2]=\{Q_k[t] | 1 \leq k \leq 6, 1 \leq t < 4\}$}\label{fig:Circuit}
   \end{minipage}}\hfill
\fbox{
   \begin{minipage}{0.48\textwidth}
     \centering
\resizebox{\columnwidth}{!}
{
\begin{tikzpicture}
  \tikzstyle{every node}=[draw, minimum size=1cm, inner sep=0.3cm]

  \node (Q1) at (0,0) {$Q_1[1]$};
  \node (Q2) at (0,-1.5) {$Q_1[2]$};
  \node (Q3) at (0,-3) {$Q_1[3]$};
  \node (Q4) at (0,-4.5) {$Q_1[4]$};

  \node (Q6) at (2,0) {$Q_2[1]$};
  \node (Q7) at (2,-1.5) {$Q_2[2]$};
  \node (Q8) at (2,-3) {$Q_2[3]$};
  \node (Q9) at (2,-4.5) {$Q_2[4]$};

  \node (Q11) at (4,0) {$Q_3[1]$};
  \node (Q12) at (4,-1.5) {$Q_3[2]$};
  \node (Q13) at (4,-3) {$Q_3[3]$};
  \node (Q14) at (4,-4.5) {$Q_3[4]$};

  \node (Q16) at (6,0) {$Q_4[1]$};
  \node (Q17) at (6,-1.5) {$Q_4[2]$};
  \node (Q18) at (6,-3) {$Q_4[3]$};
  \node (Q19) at (6,-4.5) {$Q_4[4]$};

  \node (Q21) at (8,0) {$Q_5[1]$};
  \node (Q22) at (8,-1.5) {$Q_5[2]$};
  \node (Q23) at (8,-3) {$Q_5[3]$};
  \node (Q24) at (8,-4.5) {$Q_5[4]$};

  \node (Q26) at (10,0) {$Q_6[1]$};
  \node (Q27) at (10,-1.5) {$Q_6[2]$};
  \node (Q28) at (10,-3) {$Q_6[3]$};
  \node (Q29) at (10,-4.5) {$Q_6[4]$};

  \draw[->] (Q1) -- (Q2);
  \draw[->] (Q2) -- (Q3);
  \draw[->] (Q3) -- (Q4);
  
  \draw[->] (Q1) -- (Q17);

  \draw[->] (Q6) -- (Q7);
  \draw[->] (Q7) -- (Q8);
  \draw[->] (Q8) -- (Q9);

  \draw[->] (Q7) -- (Q23);

  \draw[->] (Q11) -- (Q12);
  \draw[->] (Q12) -- (Q13);
  \draw[->] (Q13) -- (Q14);

  \draw[->] (Q13) -- (Q29);

  \draw[->] (Q16) -- (Q17);
  \draw[->] (Q17) -- (Q18);
  \draw[->] (Q18) -- (Q19);

  \draw[->] (Q21) -- (Q22);
  \draw[->] (Q22) -- (Q23);
  \draw[->] (Q23) -- (Q24);

  \draw[->] (Q26) -- (Q27);
  \draw[->] (Q27) -- (Q28);
  \draw[->] (Q28) -- (Q29);
\end{tikzpicture}
}
\caption{An example of \textbf{EPSTG} graph of six qubit transversal circuit.}\label{fig:EPSTG}
   \end{minipage}
}
\end{figure}
An example of quantum circuit $\textbf{C}$ and the corresponding $\textbf{EPSTG}$ is demonstrated in \autoref{fig:Circuit} and \autoref{fig:EPSTG}.

By \autoref{df:propgraph}, it is obvious that $\mathbf{G}(\mathbf{C},n,m,T)[t_1,t_2,p] \subseteq \mathbf{G}(\mathbf{C},n,m,T)[t_1',t_2',p]$ for any $1\leq t_1' \leq t_1 < t_2 \leq t_2' \leq n$.

The edges set in propagation graph is the abstraction of the behavior bitflip error propagation by $CNOT$ gate. An error at qubit $Q_i$ will  stay in the circuit until the next time window, embodied by the directed edge $(Q_i[t],Q_i[t+1])$. When it exists at the control qubit before a CNOT gate $CNOT(i,j)$, it can be propagated from $Q_i[t]$ to $Q_j[t+1]$ as well. Note that we are not considering the behavior of errors cancelling each others at this stage, this is formalized in the next lemma.

\begin{lemma}(Propagation rules)
    \label{lem:CNOTprop}
    In a propagation graph $\mathbf{G}(\mathbf{C},n,m,T)[t,t+1,p]$, any error occurring at $Q_i[t]$ can propagate to $Q_j[t+1]$ if and only if $(Q_i[t],Q_j[t+1]) \in \mathcal{E}(\mathbf{C},n,m,T)[t_1,t_2]$.
    \label{lem:propCNOT}
\end{lemma}
\begin{proof}
By the definition of directed edge set in \autoref{df:propgraph},$(Q_i[t],Q_j[t+1]) \in \mathcal{E}(\mathbf{C},n,m,T)[t_1,t_2]$ on if $i=j$ or there is a $CNOT(i,j)$ gate connecting $Q_i$ and $Q_j$ at time $t$. In the first case, when $i=j$, it is trivial that the error can be passed to the same qubit itself. In the second case, from \autoref{fig:ErrorPropagationRule} we know that the error can be propagated from the control qubit to the target qubit. 
\end{proof}

\begin{figure}[!htb]
\fbox{
   \begin{minipage}{0.48\textwidth}
     \centering
\resizebox{0.9\columnwidth}{!}
{
\begin{quantikz}
  Q_1  & \gate[style={fill=red!80}]{X}     & \slice{$T=1$} & \ctrl{3}& \slice{\quad $T=2$}    &  & \slice{\quad $T=3$}  &   & \slice{\quad $T=4$} &  &&  \\
 Q_2  &       &  & &  & \ctrl{3} &  & & & & & \\
 Q_3  &     &   &         &  &   &  & \ctrl{3} & & & & \\
  Q_4 &       &  & \targ{}&  \gate[style={fill=red!80}]{X} &          & \gate[style={fill=red!80}]{X}& & \gate[style={fill=red!80}]{X} &  & & \\
  Q_5 &       &  &  &      &\targ{}   &  & & &  & &   \\
  Q_6 &      &  &  &      &          &  & \targ{}& &    & &    
\end{quantikz}
}
     \caption{In this example, a random bitflip error at $Q_1[1]$ will propagate to $Q_4[2]$,$Q_4[3]$ and $Q_4[4]$.}\label{fig:CircuitProp}
   \end{minipage}}\hfill
\fbox{
   \begin{minipage}{0.48\textwidth}
     \centering
\resizebox{\columnwidth}{!}
{
\begin{tikzpicture}
  \tikzstyle{every node}=[draw, minimum size=1cm, inner sep=0.3cm]

  \node[fill=red] (Q1) at (0,0) {$Q_1[1]$};
  \node (Q2) at (0,-1.5) {$Q_1[2]$};
  \node (Q3) at (0,-3) {$Q_1[3]$};
  \node (Q4) at (0,-4.5) {$Q_1[4]$};

  \node (Q6) at (2,0) {$Q_2[1]$};
  \node (Q7) at (2,-1.5) {$Q_2[2]$};
  \node (Q8) at (2,-3) {$Q_2[3]$};
  \node (Q9) at (2,-4.5) {$Q_2[4]$};

  \node (Q11) at (4,0) {$Q_3[1]$};
  \node (Q12) at (4,-1.5) {$Q_3[2]$};
  \node (Q13) at (4,-3) {$Q_3[3]$};
  \node (Q14) at (4,-4.5) {$Q_3[4]$};

  \node (Q16) at (6,0) {$Q_4[1]$};
  \node[fill=red] (Q17) at (6,-1.5) {$Q_4[2]$};
  \node[fill=red] (Q18) at (6,-3) {$Q_4[3]$};
  \node[fill=red] (Q19) at (6,-4.5) {$Q_4[4]$};

  \node (Q21) at (8,0) {$Q_5[1]$};
  \node (Q22) at (8,-1.5) {$Q_5[2]$};
  \node (Q23) at (8,-3) {$Q_5[3]$};
  \node (Q24) at (8,-4.5) {$Q_5[4]$};

  \node (Q26) at (10,0) {$Q_6[1]$};
  \node (Q27) at (10,-1.5) {$Q_6[2]$};
  \node (Q28) at (10,-3) {$Q_6[3]$};
  \node (Q29) at (10,-4.5) {$Q_6[4]$};

  \draw[->] (Q1) -- (Q2);
  \draw[->] (Q2) -- (Q3);
  \draw[->] (Q3) -- (Q4);
  
  \draw[->,red] (Q1) -- (Q17);

  \draw[->] (Q6) -- (Q7);
  \draw[->] (Q7) -- (Q8);
  \draw[->] (Q8) -- (Q9);

  \draw[->] (Q7) -- (Q23);

  \draw[->] (Q11) -- (Q12);
  \draw[->] (Q12) -- (Q13);
  \draw[->] (Q13) -- (Q14);

  \draw[->] (Q13) -- (Q29);

  \draw[->] (Q16) -- (Q17);
  \draw[->,red] (Q17) -- (Q18);
  \draw[->,red] (Q18) -- (Q19);

  \draw[->] (Q21) -- (Q22);
  \draw[->] (Q22) -- (Q23);
  \draw[->] (Q23) -- (Q24);

  \draw[->] (Q26) -- (Q27);
  \draw[->] (Q27) -- (Q28);
  \draw[->] (Q28) -- (Q29);
\end{tikzpicture}
}
\caption{The error propagation of $\textbf{C}$ in \autoref{fig:CircuitProp} has its equivalent path in its \textbf{EPSTG} graph.}\label{fig:Prop}
   \end{minipage}
}
\end{figure}

\begin{theorem}(Propagation through connected path)
    \label{thm:proppath}
    In a propagation graph $\mathbf{G}(\mathbf{C},n,m,T)[t_1,t_2,p]$, any error occurring at $Q_i[t_1]$ can propagate to $Q_j[t_2]$ if and only if there is a connection path between $Q_i[t_1]$ and $Q_j[t_2]$ in $\mathbf{G}(\mathbf{C},n,m,T)[t_1,t_2,p]$.
    \label{lem:proppath}
\end{theorem}
\begin{proof}
Assume the path of length $k$ connecting $Q_i[t_1]$ and $Q_j[t_2]$ is $e_1=(Q_i[t_1],R_1[t_1+1]),e_2=(R_1[t_1+1],R_1[t_1+2]),\cdots,e_{t_2-t_1}=(R_{t_2-t_1}[t_2-1],Q_j[t_2])$. By applying \autoref{lem:CNOTprop} for each edge in the path, we know that the error will propagate from $Q_i[t_1]$ to $Q_j[t_2]$.  On the other hand, when an error propagates from $Q_i[t_1]$ to $Q_j[t_2]$, such path between neighboring time windows must exist.
\end{proof}

For our problem, we are interested in the error distribution at $t=T$. More specifically, we care about how many errors are left after propagation in the subset of vertices $\{Q_k[T]| 1 \leq k\leq n\}$. From now on we call these vertices syndrome vertices. This implies that syndrome vertices at the end of the propagation should be dealt with in a separate way.

\begin{definition}(Syndrome vertices and Source vertices)
    \label{df:syndromevertices}
    We define the syndrome vertices of $\mathbf{G}(\mathbf{C},n,m,T)[t_1,t_2,p], 1 \leq t_1 < t_2 \leq T$ as $\mathbf{Syn(\mathbf{G})}=\{Q_k[T]| 1 \leq k\leq n\}$. All other vertices are called source vertices, represented as $\mathbf{Source(\mathbf{G})}$ , because they might be the source of error of syndrome vertices. 
\end{definition}

\begin{definition}(Reverse spanning graph(\textbf{RSG}) of a syndrome vertex)
    \label{df:RSG}
    For every syndrome vertex $Q_k[T]$ in $\mathbf{G}(\mathbf{C},n,m,T)[t_1,t_2,p], 1 \leq t_1 < t_2 \leq T$, we define its reverse spanning graph(\textbf{RSG}($Q_k[T]$)) as 
    \begin{enumerate}
        \item The set of vertices is the subset of all the vertices in the space-time propagation graph $\textbf{G}$ that have a path connected to $Q_k[T]$, and the vertex $Q_k[T]$ itself.
        \item The directed edge set is $(Q_i[t],Q_k[T])$ for all $Q_i[t]$ in $\textbf{G}$ that have a path connected to $Q_k[T]$. 
        \end{enumerate}
    We use the term reverse in this definition because the subgraph is spanned by $Q_k[T]$ in the reverse direction of edges.
\end{definition}

\begin{figure}[h!]
\fbox{
\begin{minipage}{0.3\textwidth}
     \centering
\resizebox{0.9\columnwidth}{!}
{
\begin{quantikz}
  Q_1        & \slice{$T=1$} & \ctrl{3}& \slice{\quad $T=2$}    &  & \slice{\quad $T=3$}  &   & \slice{\quad $T=4$} &  &&  \\
 Q_2         &  & &  & \ctrl{3} &  & & & & & \\
 Q_3       &   &         &  &   &  & \ctrl{3} & & & & \\
  Q_4        &  & \targ{} &  &          & & &  &  & & \\
  Q_5        &  &  &      &\targ{}   &  & & &  & &   \\
  Q_6       &  &  &      &          &  & \targ{}& &    & &   
\end{quantikz}
}
 \caption{The six qubit transversal circuit.}\label{fig:CircuitOriginal}
\end{minipage}
}
\hfill
\fbox{
\begin{minipage}{0.3\textwidth}
     \centering
\resizebox{0.9\columnwidth}{!}
{
    \begin{tikzpicture}
    \tikzstyle{every node}=[draw, minimum size=1cm, inner sep=0.1cm]
    \node (A1) at (0,-1) {$Q_1[1]$};
    \node (A2) at (0,-3) {$Q_4[1]$};
    \node (A3) at (1.5,-3) {$Q_4[2]$};
    \node (A4) at (3,-3) {$Q_4[3]$};
    \node (A5) at (4.5,-3) {$Q_4[4]$};
    \node (A6) at (6,-3) {$Q_4[5]$};
    \draw[->]  (A1) -- (A3);
    \draw[->]  (A2) -- (A3);
    \draw[->]  (A3) -- (A4);
    \draw[->]  (A4) -- (A5);
    \draw[->]  (A5) -- (A6);
    \end{tikzpicture}
}
\caption{We are considering $\textbf{RSG}(Q_4[5])$. We take the subgraph of $\textbf{EPSTG}$ with all vertices that have at least one path with $\textbf{RSG}(Q_4[5])$.}
\label{fig:Q45}
\end{minipage}
}
\hfill
\fbox{
\begin{minipage}{0.3\textwidth}
     \centering
\resizebox{0.5\columnwidth}{!}
{
\begin{tikzpicture}
\tikzstyle{every node}=[draw, minimum size=1cm, inner sep=0.1cm]
\node (A1) at (0,-0.5) {$Q_1[1]$};
\node (A2) at (0,-2) {$Q_4[1]$};
\node (A3) at (0,-3.5) {$Q_4[2]$};
\node (A4) at (0,-5) {$Q_4[3]$};
\node (A5) at (0,-6.5) {$Q_4[4]$};
\node (A6) at (3,-3.5) {$Q_4[5]$};
\draw[->]  (A1) -- (A6);
\draw[->]  (A2) -- (A6);
\draw[->]  (A3) -- (A6);
\draw[->]  (A4) -- (A6);
\draw[->]  (A5) -- (A6);
\end{tikzpicture}
}
\caption{$\textbf{RSG}(Q_4[5])$}
\label{fig:RSGOfCircuit}
\end{minipage}
}
\caption{The example of a $\textbf{RSG}$ of a transversal circuit with $6$ qubit.}
\label{fig:RSGQ45}
\end{figure}

For the $\textbf{RSG}$ graph of a transversal circuit, an example of $6$ qubits is provided in \autoref{fig:RSGQ45}. \autoref{fig:CircuitOriginal} depict the $6$-qubit transversal circuit, \autoref{fig:Q45} gives the subgraph in the corresponding $\textbf{EPSTG}$ where all vertices have some path connected with $Q_4[5]$. According to \autoref{df:RSG}, \autoref{fig:RSGQ45} is $\textbf{RSG}(Q_4[5])$. 

Next, we further connect $\textbf{RSG}$ with the physical property of error propagation, that we concern about, in the following lemma.

\begin{lemma}
    \label{lem:RSGContain}
    \textbf{RSG}($Q_k[T]$) contains all vertices in the space-time graph of error propagation from which a bit flip error can propagate to the syndrome vertex $Q_k[T]$.
\end{lemma}
\begin{proof}
By \autoref{thm:proppath}, the set of all vertices that can propagate the error to $Q_k[T]$ is exactly the same set of all vertices that have a path with $Q_k[T]$ in $\mathbf{G}(\mathbf{C},n,m,T)[t_1,t_2,p], 1 \leq t_1 < t_2 \leq T$. By \autoref{df:RSG}, this is exactly our definition of \textbf{RSG}($Q_k[T]$). 
\end{proof}

To decide whether there is a bitflip error on one of the syndrome qubit, we have the following theorem.

\begin{theorem}(The parity of \textbf{RSG} determine the error)
    \label{thm:RSGthm}
    For any syndrome vertex $Q_k[T]$ in $\mathbf{G}(\mathbf{C},n,m,T)[t_1,t_2,p], 1 \leq t_1 < t_2 \leq T$, there is an error on $Q_k[T]$ in after error propagation by  $\mathbf{C}$  if an only if there are odd number of errors in \textbf{RSG}($Q_k[T]$) propagated to $Q_k[T]$, or equivalently, the degree of vertex $Q_k[T]$ in \textbf{RSG}($Q_k[T]$) is odd. 
\end{theorem}
\begin{proof}
By \autoref{lem:RSGContain}, the bit flip error in any vertex in \textbf{RSG}($Q_k[T]$) can propagate to $Q_k[T]$. We use the notation $\textbf{NE}(\textbf{RSG}(Q_k[T])$ as the number of random errors in \textbf{RSG}($Q_k[T]$). The error at $Q_k[T]$ after propagation is

\begin{equation}
        \hat{X}^{\textbf{NE}(\textbf{RSG}(Q_k[T])}=
\begin{cases}
        & \hat{X} \qquad \text{When $\textbf{NE}(\textbf{RSG}(Q_k[T])$ is odd.}\\
        & \hat{I} \qquad \text{When $\textbf{NE}(\textbf{RSG}(Q_k[T])$ is even}    
\end{cases}
\end{equation}
\end{proof}

\begin{figure}[h!]
\fbox{
\begin{minipage}{0.5\textwidth}
     \centering
\resizebox{\columnwidth}{!}
{
\begin{quantikz}
  Q_1        & \gate[style={fill=red!80}]{X}\slice{$T=1$} & \ctrl{3}& \slice{\quad $T=2$}    &   & &\slice{\quad $T=3$} & &   & \slice{\quad $T=4$} & & &&  \\
 Q_2         &  & &  & \ctrl{3} & & & & & & & & & \\
 Q_3       &   &         &  & &   & & \ctrl{3} &  & & & & & \\
  Q_4        &  & \targ{} & \gate[style={fill=red!80}]{X} &          & \gate[style={fill=blue!80}]{X}&\gate[style={fill=red!80}]{X} &&\gate[style={fill=blue!80}]{X} & \gate[style={fill=red!80}]{X} & &  \gate[style={fill=blue!80}]{X}\gategroup[1,steps=2,style={inner
sep=6pt}]{Cancel}& \gate[style={fill=red!80}]{X}& \\
  Q_5        &  &  &      &\targ{}   &  & && & & & & &   \\
  Q_6       &  &  &      &          & &  &\targ{}& & & &    & &   
\end{quantikz}
}
 \caption{The six qubit transversal circuit with two bit flip error at $Q_1[1]$ and $Q_4[3]$. Two error cancel each other after propagated to $Q_4[5]$.}\label{fig:CircuitOriginal2E}
\end{minipage}
}
\hfill
\fbox{
\begin{minipage}{0.45\textwidth}
     \centering
\resizebox{0.5\columnwidth}{!}
{
\begin{tikzpicture}
\tikzstyle{every node}=[draw, minimum size=1cm, inner sep=0.1cm]
\node[fill=red] (A1) at (0,-0.5) {$Q_1[1]$};
\node (A2) at (0,-2) {$Q_4[1]$};
\node (A3) at (0,-3.5) {$Q_4[2]$};
\node[fill=blue] (A4) at (0,-5) {$Q_4[3]$};
\node (A5) at (0,-6.5) {$Q_4[4]$};
\node (A6) at (3,-3.5) {$Q_4[5]$};
\draw[->,red]  (A1) -- (A6);
\draw[->]  (A2) -- (A6);
\draw[->]  (A3) -- (A6);
\draw[->,blue]  (A4) -- (A6);
\draw[->]  (A5) -- (A6);
\end{tikzpicture}
}
\caption{$\textbf{RSG}(Q_4[5])$}
\label{fig:RSGCancel}
\end{minipage}
}
\caption{In $\textbf{RGS}(Q_4[5])$, two errors at $Q_1[1]$ and $Q_4[3]$ cancel because they give a even parity to $Q_4[5]$.}
\label{fig:RSGParity}
\end{figure}

An example of two bit-flip errors canceling each other is shown in \autoref{fig:RSGParity}.

\begin{definition}(Reverse spanning graph(\textbf{RSG}) of an $\textbf{EPSTG}$)
    \label{df:RSGGraph}
    The reverse spanning graph of an $\textbf{EPSTG}$ $\mathbf{G}(\mathbf{C},n,m,T)[t_1,t_2], 1 \leq t_1 < t_2 \leq T \}$ is defined as the union of the reverse spanning graph of all syndrome vertices
    \begin{equation}
       \textbf{RSG}(\mathbf{G})= \cup_{k=1}^n \textbf{RSG}(Q_k[T])
    \end{equation}
\end{definition}

\begin{figure}[h!]
\fbox{
\begin{minipage}{0.3\textwidth}
     \centering
\resizebox{0.9\columnwidth}{!}
{
\begin{quantikz}
  Q_1    & \slice{$T=1$} & \ctrl{2}& \slice{\quad $T=2$}    &  &     &          \\
 Q_2     &               &         &  & \ctrl{2} &  & \\
 Q_3     &               & \targ{}        &  &   &  &  \\
  Q_4    &               &         &  &  \targ{}        & & 
\end{quantikz}
}
 \caption{The 4-qubit transversal circuit.}\label{fig:CircuitOrigina4q}
\end{minipage}
}
\hfill
\fbox{
\begin{minipage}{0.3\textwidth}
     \centering
\resizebox{0.9\columnwidth}{!}
{
\begin{tikzpicture}
  \tikzstyle{every node}=[draw, minimum size=1cm, inner sep=0.3cm]

  \node (Q1) at (0,0) {$Q_1[1]$};
  \node (Q2) at (0,-1.5) {$Q_1[2]$};
  \node (Q3) at (0,-3) {$Q_1[3]$};

  \node (Q4) at (2,0) {$Q_2[1]$};
  \node (Q5) at (2,-1.5) {$Q_2[2]$};
  \node (Q6) at (2,-3) {$Q_2[3]$};

  \node (Q7) at (4,0) {$Q_3[1]$};
  \node (Q8) at (4,-1.5) {$Q_3[2]$};
  \node (Q9) at (4,-3) {$Q_3[3]$};

  \node (Q10) at (6,0) {$Q_4[1]$};
  \node (Q11) at (6,-1.5) {$Q_4[2]$};
  \node (Q12) at (6,-3) {$Q_4[3]$};

  \draw[->] (Q1) -- (Q2);
  \draw[->] (Q2) -- (Q3);

  \draw[->] (Q1) -- (Q8);

  \draw[->] (Q4) -- (Q5);
  \draw[->] (Q5) -- (Q6);

  \draw[->] (Q5) -- (Q12);

  \draw[->] (Q7) -- (Q8);
  \draw[->] (Q8) -- (Q9);

  \draw[->] (Q10) -- (Q11);
  \draw[->] (Q11) -- (Q12);

\end{tikzpicture}
}
\caption{The $\textbf{EPSTG}$ of the $4$ qubit transversal circuit, denote as $\textbf{G}$}
\label{fig:4qEG}
\end{minipage}
}
\hfill
\fbox{
\begin{minipage}{0.3\textwidth}
     \centering
\resizebox{0.5\columnwidth}{!}
{
\begin{tikzpicture}
\tikzstyle{every node}=[draw, minimum size=1cm, inner sep=0.1cm]
\node (A1) at (0,-0.5) {$Q_1[1]$};
\node (A2) at (0,-2) {$Q_1[2]$};
\node (A3) at (0,-3.5) {$Q_2[1]$};
\node (A4) at (0,-5) {$Q_2[2]$};
\node (A5) at (0,-6.5) {$Q_3[1]$};
\node (A6) at (0,-8) {$Q_3[2]$};
\node (A7) at (0,-9.5) {$Q_4[1]$};
\node (A8) at (0,-11) {$Q_4[2]$};

\node (A9) at (3,-3.5) {$Q_1[3]$};
\node (A10) at (3,-5) {$Q_2[3]$};
\node (A11) at (3,-6.5) {$Q_3[3]$};
\node (A12) at (3,-8) {$Q_4[3]$};

\draw[->,red]  (A1) -- (A9);
\draw[->,red]  (A2) -- (A9);
\draw[->,blue]  (A1) -- (A11);

\draw[->,green]  (A3) -- (A10);
\draw[->,green]  (A4) -- (A10);

\draw[->,pink]  (A3) -- (A12);
\draw[->,pink]  (A4) -- (A12);

\draw[->,blue]  (A5) -- (A11);
\draw[->,blue]  (A6) -- (A11);

\draw[->,pink]  (A7) -- (A12);
\draw[->,pink]  (A8) -- (A12);
\end{tikzpicture}
}
\caption{We get $\textbf{RSG}(\textbf{G})$ of 4-qubit transversal circuit by combining $\textbf{RSG}(Q_1[3]),\textbf{RSG}(Q_2[3]),\textbf{RSG}(Q_3[3])$ and $\textbf{RSG}(Q_4[3])$ together. }
\label{fig:RSGOfCircuit}
\end{minipage}
}
\caption{The example of a $\textbf{RSG}$ of a transversal circuit with $4$ qubit.}
\label{fig:RSGTotal}
\end{figure}

An example of $\textbf{RSG}(\textbf{G})$ is presented in \autoref{fig:RSGTotal}. After converting the original circuit with $T=3$ in \autoref{fig:CircuitOrigina4q} to its $\textbf{EPSTG}$ in \autoref{fig:4qEG}, we construct the reverse spanning graph of all four syndrome vertices we care about: $\textbf{RSG}(Q_1[3]),\textbf{RSG}(Q_2[3]),\textbf{RSG}(Q_3[4])$ and $\textbf{RSG}(Q_4[3])$, and combine them together to get $\textbf{RSG}(\textbf{G})=\textbf{RSG}(Q_1[3]) \cup \textbf{RSG}(Q_2[3]) \cup \textbf{RSG}(Q_3[3]) \cup \textbf{RSG}(Q_4[3])$. Clearly, \autoref{fig:RSGOfCircuit} is a direct representation of the accumulated bit-flip errors propagate to different syndrome vertices.

\begin{theorem}(An efficient algorithm to construct \textbf{RSG})
    \label{thm:RSGAlg}
    There is an efficient algorithm to construct \textbf{RSG}(\textbf{G}) for any \textbf{EPSTG} with total time window $T$ and $n$ qubits with $m$ $CNOT$ gates run in $poly(T,m,n)$.
\end{theorem}
\begin{proof}
  The idea of the algorithm goes as follows:
  \begin{enumerate}
      \item Construct a new graph $\textbf{G'}$ with all vertices in $\textbf{G}$ and no edges. 
      \item For each syndrome vertex $Q_i(T)$ in $\textbf{Syn}(\textbf{G})$, run a reverse depth first(\textbf{DFS}) search.
      \item For any vertex $Q_j(t)$ found in $\textbf{DFS}$, add edge $(Q_j(t),Q_i(T))$ to the constructed graph $\textbf{G'}$.
  \end{enumerate}
  $\textbf{G'}$ is exactly the \textbf{RSG}(\textbf{G}). Since each \textbf{DFS} has worse-case time complexity $O(Tn+m)$, the total algorithm has time complexity $O(Tn^2+mn)$.
\end{proof}

\begin{lemma}
    \label{lem:RSGBipart}
    $\textbf{RSG}(\mathbf{G})$ is a bipartite graph.
\end{lemma}
\begin{proof}
We put all the vertices of the syndrome $\textbf{Syn}(\textbf{G})$ defined in \autoref{df:syndromevertices} on the right side and all the other vertices, or $\textbf{Source}(\textbf{G})$, on the left side. By \autoref{df:RSGGraph}, there are only edges from the left to the right. The edge between one source vertex $Q_i(t_1)$ and the other syndrome vertex $Q_j(T)$ means that the parity of $Q_i(t_1)$ can affect the parity of $Q_j(T)$ after error propagation. 
\end{proof}

\begin{theorem}(Error Subgraph of $\textbf{RSG}(\mathbf{G})$)
    \label{thm:spanningsubgraphparity}
    For each subgraph $\textbf{G'}$ in $\textbf{RSG}(\mathbf{G})$, we call  $\textbf{G'}$ an error subgraph of $\textbf{RSG}(\mathbf{G})$ if every vertex in $\textbf{G'} \cap \textbf{Syn}(\mathbf{G})$ has odd degree.    
    Each error subgraph $\textbf{G'}$ uniquely corresponds to a pattern for the random error that occurred before the syndrome and the result of the parity of syndrome. The probability of each pattern to occur under a random single-qubit bitflip model of constant probability $p$ is \begin{equation}
        P(\textbf{G'} \cap \textbf{Source}(\mathbf{G}))= p^{|\textbf{G'} \cap \textbf{Source}(\mathbf{G})|}(1-p)^{(n-1)T-|\textbf{G'} \cap \textbf{Source}(\mathbf{G})|}. 
        \label{eq:probability}
    \end{equation}

\end{theorem}
\begin{proof}
    By \autoref{lem:RSGBipart}, $\textbf{RSG}(\mathbf{G})$ is a bipartite graph, the left side is the $\textbf{Source}(\textbf{G})$ which determine the random errors and the right side is the $\textbf{Syn}(\textbf{G})$ which determine the syndrome at time $T$ after propagation. 
    By \autoref{thm:RSGthm}, the parity of the degree of a vertex $Q_k(T)$ in $\textbf{RSG}(\mathbf{G})$ determine whether there is an error in $\textbf{RSG}(\mathbf{G})$ after propagation. Since every vertex in $\textbf{G'} \cap \textbf{Syn}(\mathbf{G})$ has odd degree, it is exactly the set of syndrome qubit at time $T$ which still have errors after propagation. 

    For each pattern, there are $\# \mathcal{V}=|\textbf{G'} \cap \textbf{Source}(\mathbf{G})|$ random errors in the error propagation space-time graph, with probability $p^{\# \mathcal{V}}$. The probability for the rest $(T-1)n-\# \mathcal{V}$ to have no error is $(1-p)^{(T-1)n-\# \mathcal{V}}$.
    So the overall probability of each random error pattern correspond to the subgraph is $p^{\# \mathcal{V}}(1-p)^{(T-1)n-\# \mathcal{V}}$.
\end{proof}

\begin{figure}[h!]
\fbox{
\begin{minipage}{0.3\textwidth}
     \centering
\resizebox{0.9\columnwidth}{!}
{
\begin{quantikz}
  Q_1  & \gate[style={fill=red!80}]{X} & \slice{$T=1$} & \ctrl{2}& \slice{\quad $T=2$}    &  &   \gate[style={fill=red!80}]{X}  &          \\
 Q_2   &\gate[style={fill=red!80}]{X}  &               &         &  & \ctrl{2} & \gate[style={fill=red!80}]{X} & \\
 Q_3   &  &               & \targ{}        & \gate[style={fill=red!80}]{X} &   &  &  \\
  Q_4  &  &               &         & \gate[style={fill=red!80}]{X} &  \targ{}        & & 
\end{quantikz}
}
 \caption{The 4-qubit transversal circuit with four bit-flip error at $Q_1[1]$,$Q_2[2]$,$Q_3[2]$ and $Q_4[2]$.}\label{fig:4error4qTransversal}
\end{minipage}
}
\hfill
\fbox{
\begin{minipage}{0.3\textwidth}
     \centering
\resizebox{0.5\columnwidth}{!}
{
    \begin{tikzpicture}
    \tikzstyle{every node}=[draw, minimum size=1cm, inner sep=0.1cm]
    \node[fill=red] (A1) at (0,-0.5) {$Q_1[1]$};
    \node (A2) at (0,-2) {$Q_1[2]$};
    \node[fill=red] (A3) at (0,-3.5) {$Q_2[1]$};
    \node (A4) at (0,-5) {$Q_2[2]$};
    \node (A5) at (0,-6.5) {$Q_3[1]$};
    \node[fill=red] (A6) at (0,-8) {$Q_3[2]$};
    \node (A7) at (0,-9.5) {$Q_4[1]$};
    \node[fill=red] (A8) at (0,-11) {$Q_4[2]$};
    
    \node[fill=red] (A9) at (3,-3.5) {$Q_1[3]$};
    \node[fill=red] (A10) at (3,-5) {$Q_2[3]$};
    \node (A11) at (3,-6.5) {$Q_3[3]$};
    \node (A12) at (3,-8) {$Q_4[3]$};
    
    \draw[->,red]  (A1) -- (A9);
    \draw[->,red,dash dot]  (A2) -- (A9);
    \draw[->,blue]  (A1) -- (A11);
    
    \draw[->,green]  (A3) -- (A10);
    \draw[->,green,dash dot]  (A4) -- (A10);
    
    \draw[->,pink]  (A3) -- (A12);
    \draw[->,pink,dash dot]  (A4) -- (A12);
    
    \draw[->,blue,dash dot]  (A5) -- (A11);
    \draw[->,blue]  (A6) -- (A11);
    
    \draw[->,pink,dash dot]  (A7) -- (A12);
    \draw[->,pink]  (A8) -- (A12);
    \end{tikzpicture}
}
 \caption{$\textbf{RSG}$ of \autoref{fig:4error4qTransversal}. The vertices with bitflip error is filled in red. And I use dashed line to indicate there is no error propagation between the error vertex with the syndrome vertex.}\label{fig:probRSG}
\end{minipage}
}
\hfill
\fbox{
\begin{minipage}{0.3\textwidth}
     \centering
\resizebox{0.5\columnwidth}{!}
{
    \begin{tikzpicture}
    \tikzstyle{every node}=[draw, minimum size=1cm, inner sep=0.1cm]
    \node[fill=red] (A1) at (0,-0.5) {$Q_1[1]$};
    \node[fill=red] (A3) at (0,-2) {$Q_2[1]$};
    \node[fill=red] (A6) at (0,-3.5) {$Q_3[2]$};
    \node[fill=red] (A8) at (0,-5) {$Q_4[2]$};
    \node (A2) at (0,-7.5) {$Q_1[2]$};
    \node (A4) at (0,-9) {$Q_2[2]$};
    \node (A5) at (0,-10.5) {$Q_3[1]$};
    \node (A7) at (0,-12) {$Q_4[1]$};

    \node[fill=red] (A9) at (3,-3.5) {$Q_1[3]$};
    \node[fill=red] (A10) at (3,-5) {$Q_2[3]$};
    \node (A11) at (3,-7.5) {$Q_3[3]$};
    \node (A12) at (3,-9) {$Q_4[3]$};
    
    \draw[->,red]  (A1) -- (A9);
    \draw[->,red,dash dot]  (A2) -- (A9);
    \draw[->,blue]  (A1) -- (A11);
    
    \draw[->,green]  (A3) -- (A10);
    \draw[->,green,dash dot]  (A4) -- (A10);
    
    \draw[->,pink]  (A3) -- (A12);
    \draw[->,pink,dash dot]  (A4) -- (A12);
    
    \draw[->,blue,dash dot]  (A5) -- (A11);
    \draw[->,blue]  (A6) -- (A11);
    
    \draw[->,pink,dash dot]  (A7) -- (A12);
    \draw[->,pink]  (A8) -- (A12);
    \end{tikzpicture}
}
\caption{Rearrage \autoref{fig:probRearrange}. The upper half is the error subgraph of $\textbf{RSG}(\textbf{G})=\textbf{G}'$, as defined in \autoref{thm:spanningsubgraphparity}.}
\label{fig:probRearrange}
\end{minipage}
}
\caption{An illustration for \autoref{thm:spanningsubgraphparity}.}
\label{fig:RSGProgability}
\end{figure}

Now comes the main theorem of this section. We can reduce our problem of calculating the distribution of number of errors to the problem of counting the number of a certain type of subgraph,

\begin{theorem}(Probability distribution is weighted counting of error subgraph)
    \label{thm:weightedCouting}
    Given a quantum circuit of $n$ qubit $\mathcal{C}$ composed of only $CNOT$ gates under a random bitflip error with error probability $p$ in each discretized time window, the syndrome error distribution at time window $T$ is given by counting the weighted error subgraph of $\textbf{RSG}(\textbf{G})$, where $\textbf{G}$ is error propagation space-time graph of the circuit.
    \begin{equation}
        P(k)=\sum_{\substack{\textbf{G'} \subset \textbf{RSG}(\textbf{G}) \\ \textbf{G'} \text{is error subgraph}  \\|\textbf{G'} \cap \textbf{Syn}(G)|=k}} P(\textbf{G'} \cap \textbf{Source}(\mathbf{G}))
    \end{equation}
    We make the domain of probability space $\mathbb{Z}$ as when $k > n=|\textbf{Syn}(\textbf{G}) \cap \textbf{G'}|$, since there are at most $n$ errors, we just make $P(k)=0$. 
\end{theorem}
\begin{proof}
    Proof is trivial by \autoref{thm:spanningsubgraphparity}.
\end{proof}

Now we have reduced the problem to a weighted subgraph counting problem on a bipartite graph. When $p=\frac{1}{2}$, all error subgraph $\textbf{G'}$ share the same weight $(\frac{1}{2})^{(n-1)T}$, the problem is reduced to a subgraph counting problem. 

In general, subgraph counting is hard \citep{subgraphCounting}. But we can simplify the problem when the $\textbf{RSG}$ can be divided into many smaller connected directed graph, defined as follows:

\begin{definition}(Connected directed graph)
    \label{def:connectedDirected}
    A connected directed graph $G$ is defined as a connected graph such that there exists at least one vertex $Q$ such that there is a path connected each other vertices in $G$ to $Q$. 
\end{definition}

On the contrary, we define the concept of disjoined graph as:

\begin{definition}(Disjoint directed graph)
    \label{def:disjoint}
    We call two connected directed graph $G_1$,$G_2$ disjoint if there exists no path connected from any pair of vertices chosen one from $G_1$ and the other from $G_2$.
\end{definition}

\begin{figure}[h!]
\fbox{
\begin{minipage}{0.3\textwidth}
     \centering
\resizebox{0.9\columnwidth}{!}
{
\begin{quantikz}
  Q_1    & \slice{$T=1$} & \ctrl{2}& \slice{\quad $T=2$}    &  &     &          \\
 Q_2     &               &         & \ctrl{2} &  &  & \\
 Q_3     &               & \targ{}        &  &   &  &  \\
  Q_4    &               &         &  \targ{} &        & & 
\end{quantikz}
}
 \caption{The 4-qubit  parallel transversal circuit.}\label{fig:CircuitOrigina4q}
\end{minipage}
}
\hfill
\fbox{
\begin{minipage}{0.3\textwidth}
     \centering
\resizebox{0.5\columnwidth}{!}
{
\begin{tikzpicture}
\tikzstyle{every node}=[draw, minimum size=1cm, inner sep=0.1cm]
\node (A1) at (0,-0.5) {$Q_1[1]$};
\node (A2) at (0,-2) {$Q_1[2]$};
\node (A3) at (0,-3.5) {$Q_2[1]$};
\node (A4) at (0,-5) {$Q_2[2]$};
\node (A5) at (0,-6.5) {$Q_3[1]$};
\node (A6) at (0,-8) {$Q_3[2]$};
\node (A7) at (0,-9.5) {$Q_4[1]$};
\node (A8) at (0,-11) {$Q_4[2]$};

\node (A9) at (3,-3.5) {$Q_1[3]$};
\node (A10) at (3,-5) {$Q_2[3]$};
\node (A11) at (3,-6.5) {$Q_3[3]$};
\node (A12) at (3,-8) {$Q_4[3]$};

\draw[->,red]  (A1) -- (A9);
\draw[->,red]  (A2) -- (A9);
\draw[->,blue]  (A1) -- (A11);

\draw[->,green]  (A3) -- (A10);
\draw[->,green]  (A4) -- (A10);

\draw[->,pink]  (A3) -- (A12);

\draw[->,blue]  (A5) -- (A11);
\draw[->,blue]  (A6) -- (A11);

\draw[->,pink]  (A7) -- (A12);
\draw[->,pink]  (A8) -- (A12);
\end{tikzpicture}
}
\caption{The \textbf{RSG} of the four-qubit parallel transversal circuit.}
\label{fig:4qRSGOriginalTrans}
\end{minipage}
}
\hfill
\fbox{
\begin{minipage}{0.3\textwidth}
     \centering
\resizebox{0.5\columnwidth}{!}
{
\begin{tikzpicture}
\tikzstyle{every node}=[draw, minimum size=1cm, inner sep=0.1cm]
\node (A1) at (0,-0.5) {$Q_1[1]$};
\node (A2) at (0,-2) {$Q_1[2]$};
\node (A5) at (0,-3.5) {$Q_3[1]$};
\node (A6) at (0,-5) {$Q_3[2]$};
\node (A3) at (0,-6.5) {$Q_2[1]$};
\node (A4) at (0,-8) {$Q_2[2]$};
\node (A7) at (0,-9.5) {$Q_4[1]$};
\node (A8) at (0,-11) {$Q_4[2]$};

\node (A9) at (3,-2) {$Q_1[3]$};
\node (A11) at (3,-3.5) {$Q_3[3]$};

\node (A10) at (3,-8) {$Q_2[3]$};
\node (A12) at (3,-9.5) {$Q_4[3]$};

\draw[->,red]  (A1) -- (A9);
\draw[->,red]  (A2) -- (A9);
\draw[->,blue]  (A1) -- (A11);

\draw[->,green]  (A3) -- (A10);
\draw[->,green]  (A4) -- (A10);

\draw[->,pink]  (A3) -- (A12);

\draw[->,blue]  (A5) -- (A11);
\draw[->,blue]  (A6) -- (A11);

\draw[->,pink]  (A7) -- (A12);
\draw[->,pink]  (A8) -- (A12);
\end{tikzpicture}
}
\caption{After rearrange the vertices in \autoref{fig:4qRSGOriginalTrans}, the $\textbf{RSG}$ for 4-qubit parallel transversal circuit is divided into two disjoint subgraph.}
\label{fig:RSGTrans}
\end{minipage}
}
\caption{The 4-qubit parallel transversal circuit with its $\textbf{RSG}$. }
\label{fig:Trans}
\end{figure}

\begin{figure}[h!]
\fbox{
\begin{minipage}{0.3\textwidth}
     \centering
\resizebox{0.9\columnwidth}{!}
{
\begin{quantikz}
  Q_1    & \slice{$T=1$} &  & \slice{\quad $T=2$}    &  &     &          \\
 Q_2     &               &         &  &  &  & \\
 Q_3     &               &        &  &   &  &  \\
  Q_4    &               &         &  &         & & 
\end{quantikz}
}
 \caption{The 4-qubit gate with no $CNOT$ gate.}\label{fig:CircuitOrigina4q}
\end{minipage}
}
\hfill
\fbox{
\begin{minipage}{0.3\textwidth}
     \centering
\resizebox{0.5\columnwidth}{!}
{
\begin{tikzpicture}
\tikzstyle{every node}=[draw, minimum size=1cm, inner sep=0.1cm]
\node (A1) at (0,-0.5) {$Q_1[1]$};
\node (A2) at (0,-2) {$Q_1[2]$};
\node (A3) at (0,-3.5) {$Q_2[1]$};
\node (A4) at (0,-5) {$Q_2[2]$};
\node (A5) at (0,-6.5) {$Q_3[1]$};
\node (A6) at (0,-8) {$Q_3[2]$};
\node (A7) at (0,-9.5) {$Q_4[1]$};
\node (A8) at (0,-11) {$Q_4[2]$};

\node (A9) at (3,-3.5) {$Q_1[3]$};
\node (A10) at (3,-5) {$Q_2[3]$};
\node (A11) at (3,-6.5) {$Q_3[3]$};
\node (A12) at (3,-8) {$Q_4[3]$};

\draw[->,red]  (A1) -- (A9);
\draw[->,red]  (A2) -- (A9);

\draw[->,green]  (A3) -- (A10);
\draw[->,green]  (A4) -- (A10);

\draw[->,blue]  (A5) -- (A11);
\draw[->,blue]  (A6) -- (A11);

\draw[->,pink]  (A7) -- (A12);
\draw[->,pink]  (A8) -- (A12);
\end{tikzpicture}
}
\caption{The $\textbf{RSG}$ of 4-qubit circuit without CNOT gate.}
\label{fig:4qRSGOriginal}
\end{minipage}
}
\hfill
\fbox{
\begin{minipage}{0.3\textwidth}
     \centering
\resizebox{0.5\columnwidth}{!}
{
\begin{tikzpicture}
\tikzstyle{every node}=[draw, minimum size=1cm, inner sep=0.1cm]
\node (A1) at (0,-0.5) {$Q_1[1]$};
\node (A2) at (0,-2) {$Q_1[2]$};
\node (A3) at (0,-3.5) {$Q_2[1]$};
\node (A4) at (0,-5) {$Q_2[2]$};
\node (A5) at (0,-6.5) {$Q_3[1]$};
\node (A6) at (0,-8) {$Q_3[2]$};
\node (A7) at (0,-9.5) {$Q_4[1]$};
\node (A8) at (0,-11) {$Q_4[2]$};

\node (A9) at (3,-1.5) {$Q_1[3]$};
\node (A10) at (3,-4.5) {$Q_2[3]$};
\node (A11) at (3,-7.5) {$Q_3[3]$};
\node (A12) at (3,-10) {$Q_4[3]$};

\draw[->,red]  (A1) -- (A9);
\draw[->,red]  (A2) -- (A9);

\draw[->,green]  (A3) -- (A10);
\draw[->,green]  (A4) -- (A10);

\draw[->,blue]  (A5) -- (A11);
\draw[->,blue]  (A6) -- (A11);

\draw[->,pink]  (A7) -- (A12);
\draw[->,pink]  (A8) -- (A12);
\end{tikzpicture}
}
\caption{Divide $\textbf{RSG}$ in \autoref{fig:4qRSGOriginal} into $4$ disjoint subgraph. }
\label{fig:RSGDisjoint}
\end{minipage}
}
\caption{The circuit, \textbf{RSG} graph of a $4$ qubit circuit without CNOT gates.}
\label{fig:Disjoint}
\end{figure}

It should be point out that we can assume that we only consider the situation that error propagation space-time graph(\textbf{EPSTG}) is a connected directed graph defined in \autoref{def:connectedDirected}.

\begin{theorem}(Combined error distribution of two disjoint subgraphs)
    Now assume that we have two disjoint subgraphs $G_1$, $G_2$ in the same \textbf{RSG}. The error number distribution of $G_1$ is $\{ P_1(k)|1 \leq k\leq n \}$, and $\{ P_2(k)|1 \leq k\leq n \}$ for $G_2$. The joint error distribution of $G_1$ and $G_2$ is
    \begin{equation}
        P(k)=\sum_{r=0}^k P_1(r)P_2(k-r)
        \label{eq:jointdistribution}
    \end{equation}
    \autoref{eq:jointdistribution} can be computed with time complexity $O(k)$ given $P_1$,$P_2$. The time complexity to get the whole distribution is $O(n^2)$.
\end{theorem}
\begin{proof}
    By \autoref{thm:proppath}, errors can only propagate through paths. So there is no error propagation between $G_1$ and  $G_2$. $P_1(k)$ and $P_2(k)$ are thus two independent probability distributions. Let $X$,$Y$ be two random variables for the instances in two distribution.  
    \begin{equation}
         P(|X|=r, |Y|=(k-r))=P_1(r)P_2(k-r)
    \end{equation}
    \begin{equation}
        P(k)=\sum_{r=0}^k P(|X|=r, |Y|=(k-r))= \sum_{r=0}^k P_1(r)P_2(k-r).
    \end{equation}
    This complete the proof.
\end{proof}

It easy to get the following corollary:

\begin{corollary}(Combined error distribution of many disjoint subgraphs)
\label{cor:subgraphs}
Assume that we have $m$ disjoint subgraphs $G_1$, $G_2$,$G_3$,$\cdots$ $G_m$ in the same \textbf{RSG}. Each $G_r$ has the probability distribution $P_r$ for $1 \leq r \leq n$. The final error joint distribution is given by   
\begin{equation}
    P(k)= \sum_{\substack{a_1,a_2,\cdots,a_m\\ a_1+a_2+\cdots+a_m=k}}P_1(a_1)P_2(a_2)\cdots P_m(a_m)
\end{equation}
\end{corollary}

\autoref{cor:subgraphs} suggests that we just need to analyze the cases when we have a connected graph, because we can use dynamic programming to calculate the probability efficiently, with the following dynamic algorithm:

\begin{enumerate}
    \item Define $P(k,r)$ as the probability of getting $k$ error within the first $r$ subgraphs. We want to calculate $P(k,n)$ for all $1 \leq k\leq n$.
    \item Calculate $P(0,1)=P_1(0)$, $P(1,1)=P_1(1)$, $P(k,1)=0$ for $2 \leq  k \leq n$.
    \item For $2 \leq k \leq n$, we can calculate $P(k,r)$ by $P(k,r)=\sum_{q=0}^r P(k-1,r-q)P_r(q)$, which takes $O(\sum_{r=0}^n)=O(n^2)$ steps.
\end{enumerate}

The time complexity for the above dynamic algorithm is $O(n^3)$. 
Since our goal is to analyze the impact of circuit on the error distribution, a very intuitive measure is the expectation value of error distribution as well as its variation, defined as 

\begin{equation}
    E(\textbf{RSG})=\sum_{k=1}^n kP(k), \qquad Var(\textbf{RSG})=\sum_{k=1}^n (k-E(\textbf{RSG}))^2P(k)=\sum_{k=1}^n k^2P(k)-E(\textbf{RSG})^2
    \label{eq:EandVar}
\end{equation}

\begin{lemma}(Zero probability of oversize disjoint subgraphs)
    \label{lem:zeroprobability}
    If $\textbf{RSG}$ for some error propagation graph $\textbf{G}$ is composed of $m$ disjoint subgraphs $\{G_k | 1 \leq k \leq m \}$, then if $a_1+a_2+\cdots+a_m >n$, we must have:
    \begin{equation}
        P_1(a_1)P_2(a_2)\cdots P_m(a_m)=0
        \label{eq:zeroprobability}
    \end{equation}
\end{lemma}
\begin{proof}
Denote that $|G_k \cap \textbf{Syn}(\textbf{G})|=b_k$ for $1 \leq k \leq n$, because they disjoint, $\sum_{k=1}^n b_k=n$. According to the pigeonhole principle, there is at least one $a_k > b_k$.
As proved in \autoref{thm:weightedCouting}, when $a_k > b_k=|G_k \cap \textbf{Syn}(\textbf{G})|$, $P(k)=0$, so the product in \autoref{eq:zeroprobability} must be $0$.
\end{proof}

\begin{theorem}(Expectation of many disjoint subgraphs)
    \label{thm:ExpecDisj}
    If $\textbf{RSG}$ is composed of $m$ disjoint subgraph $\{G_k | 1 \leq k \leq m \}$, then:
    \begin{equation}
        E(\textbf{RSG})=\sum_{k=1}^m  E(G_k)
        \label{eq:DisjointSubgraph}
    \end{equation}
\end{theorem}
\begin{proof}
By \autoref{cor:subgraphs}, assume $G_k$ has independent distribution $P_k$ for each $1\leq k \leq m$, with $E(G_k)=\sum_{a_k=0}^n P_k(a_k)$, the expectation value is given by:
\begin{equation}
\begin{split}
        E(\textbf{RSG})&=\sum_{k=1}^n kP(k)\\
                       &=\sum_{k=1}^n k \sum_{\substack{a_1,a_2,\cdots,a_m\\ a_1+a_2+\cdots+a_m=k}}P_1(a_1)P_2(a_2)\cdots P_m(a_m)\\
                       &=\sum_{\substack{a_1,a_2,\cdots,a_m\\ a_1+a_2+\cdots+a_m \leq n}}(a_1+a_2+\cdots+a_m)P_1(a_1)P_2(a_2)\cdots P_m(a_m)\\
                       &=\sum_{a_1=0}^na_1P_1(a_1) \big(\sum_{\substack{a_2,\cdots,a_m\\ a_2+\cdots+a_m \leq n-a_1}}P_2(a_2)\cdots P_m(a_m)\big)+\\
                       &\sum_{a_2=0}^na_2P_2(a_2) \big(\sum_{\substack{a_1,a_3,\cdots,a_m\\ a_1+a_3+\cdots+a_m \leq n-a_1}}P_1(a_1)P_3(a_3)\cdots P_m(a_m)\big)+\cdots\\
                       &=\sum_{a_1=0}^na_1P_1(a_1) \big( \prod_{k=2}^m \sum_{r=0}^n P_k(r)\big)+\sum_{a_2=0}^na_2P_2(a_2) \big( \prod_{\substack{k=1\\k\neq 2}}^m \sum_{r=0}^n P_k(r)\big)+\cdots\\
                       &=\sum_{a_1=0}^na_1P_1(a_1)+\sum_{a_2=0}^na_2P_2(a_2)+\cdots +\sum_{a_m=0}^na_mP_m(a_m)\\
                       &=\sum_{k=1}^m E(G_k)
\end{split}
\label{eq:Expectation}
\end{equation}
In the derivation of \autoref{eq:Expectation}, we first expand $k$ as $k=a_1+a_2+\cdots+a_m$, and then extract $a_kP_k(aP_k)$ in front of each summation. Since $G_1$,$G_2$,$\cdots$,$G_m$ are a set of disjoint subgraphs of $\textbf{RSG}$, by \autoref{lem:zeroprobability}, we can remove the conditions for the rest of $a_i$, such as $a_2+\cdots+a_m \leq n -a_1$, since the product with $a_2+\cdots+a_m > n -a_1$ must be $0$. Next, we use the transformation $\sum_{a_2,\cdots,a_m} \prod_{k=2}^m P_k(a_k)=\prod_{k=2}^m \sum_{a_k} P_k(a_k)=1$ to simplify the equation.
\end{proof}

\begin{lemma}(Even and odd sum of binomial distribution)
    \label{lem:EvenOddSum}
    We have the following two formula:
    \begin{equation}
        \begin{split}
         \sum_{\substack{k=0 \\ k\equiv 0 \mod 2}}^T  \binom{T}{k}a^kb^{T-k}&=\frac{1}{2}(b-a)^T+\frac{1}{2}(b+a)^T  \\
         \sum_{\substack{k=1 \\ k\equiv 1 \mod 2}}^T  \binom{T}{k}a^kb^{T-k}&=\frac{1}{2}(a+b)^T-\frac{1}{2}(b-a)^T 
        \end{split}
    \label{eq:EvenOddSum}
    \end{equation}
    Specifically, when $a=p$, $b=(1-p)$:
    \begin{equation}
        \begin{split}
         \sum_{\substack{k=0 \\ k\equiv 0 \mod 2}}^T  \binom{T}{k}p^k(1-p)^{T-k}&=\frac{1}{2}(1-2p)^T+\frac{1}{2} \\
         \sum_{\substack{k=1 \\ k\equiv 1 \mod 2}}^T  \binom{T}{k}p^k(1-p)^{T-k}&=\frac{1}{2}-\frac{1}{2}(1-2p)^T
        \end{split}
    \label{eq:EvenOddSum}
    \end{equation}
\end{lemma}

\begin{theorem}(Variance of many disjoint subgraphs)
    \label{thm:VarDisj}
    If $\textbf{RSG}$ is composed of $m$ disjoint subgraphs $\{G_k | 1 \leq k \leq m \}$, then:
    \begin{equation}
        \begin{split}
        Var(\textbf{RSG})&=\sum_{k=1}^m  Var(G_k)+2\sum_{ 1\leq i\leq j\leq m } Cov(G_i,G_j)\\
        &=\sum_{k=1}^m  Var(G_k)
        \end{split}
        \label{eq:DisjointSubgraphVar}
    \end{equation}
    Where $Cov(G_i,G_j)$ is the covariance of two disjoint subgraphs $G_i,G_j$, defined as:
    \begin{equation}
          Cov(G_i,G_j)=\sum_{a_i,a_j}a_ia_j P_i(a_i)P_j(a_j)-E(G_i)E(G_j)=0
          \label{eq:Cov}
    \end{equation}
    
\end{theorem}
\begin{proof}
    \begin{equation}
    \begin{split}
              Var(\textbf{RSG})&=\sum_{k=1}^n k^2P(k)-E(\textbf{RSG})^2\\
              &=\sum_{k=1}^n k^2 \sum_{\substack{a_1,a_2,\cdots,a_m\\ a_1+a_2+\cdots+a_m=k}}P_1(a_1)P_2(a_2)\cdots P_m(a_m)-(\sum_{k=1}^m  E(G_k))^2\\
              &=\sum_{k=1}^n  \sum_{\substack{a_1,a_2,\cdots,a_m\\ a_1+a_2+\cdots+a_m=k}}(a_1+a_2+\cdots+a_m)^2P_1(a_1)P_2(a_2)\cdots P_m(a_m)-(\sum_{k=1}^m  E(G_k))^2\\
              &= \sum_{\substack{a_1,a_2,\cdots,a_m\\ a_1+a_2+\cdots+a_m \leq n}}(a_1^2+a_2^2+\cdots+a_m^2+2\sum_{1\leq i<j\leq m}a_ia_j)P_1(a_1)P_2(a_2)\cdots P_m(a_m)-(\sum_{k=1}^m  E(G_k))^2\\
              &=\sum_{a_1=0}^n a_1^2P_a(a_1)\sum_{a_k} \prod_{k=2,k\neq 1}^m P_k(a_k)+\sum_{a_2=0}^n a_2^2P_2(a_2)\sum_{a_k} \prod_{k=1,k\neq 2}^m P_k(a_k)+\cdots+\\
              &\sum_{a_m=0}^n a_m^2P_m(a_m)\sum_{a_k} \prod_{k=1,k\neq m}^m P_k(a_k)+2\sum_{1 \leq i\leq j\leq m}\sum_{a_1,a_2,\cdots,a_m}a_ia_j \prod_{k=1}^m P_k(a_k)-(\sum_{k=1}^m  E(G_k))^2\\
              &=\sum_{k=1}^m \sum_{a_k=0}^n a_k^2 P_k(a_k)+2\sum_{1 \leq i\leq j\leq m}\sum_{a_1,a_2,\cdots,a_m}a_ia_j \prod_{k=1}^m P_k(a_k)-\sum_{k=1}^m  E(G_k)^2-\sum_{1 \leq i\leq j\leq m}2E(G_i)E(G_j)\\
              &=\sum_{k=1}^m(\sum_{a_k=0}^n a_k^2 P_k(a_k)-E(G_k)^2)+2\sum_{1 \leq i\leq j\leq m}(\sum_{a_i,a_j}a_ia_j P_i(a_i)P_j(a_j)\sum_{a_k,k\neq i,j.}\prod_{k=1}^m P_k(a_k)-E(G_i)E(G_j))\\
              &=\sum_{k=1}^m Var(G_k)+2\sum_{1 \leq i\leq j\leq m}(\sum_{a_i,a_j}a_ia_j P_i(a_i)P_j(a_j)-E(G_i)E(G_j))\\
              &=\sum_{k=1}^m  Var(G_k)+2\sum_{ 1\leq i\leq j\leq m } Cov(G_i,G_j)\\
              &=\sum_{k=1}^m  Var(G_k)
    \end{split}
    \end{equation}
\end{proof}

A directly application with the knowledge of expectation value of error number distribution is estimating the logical error rate under error propagation.

\begin{theorem}(Bounded logical error rate by expectation value)
    \label{thm:logica}
    For a quantum circuit $\textbf{C}$ with quantum error correction code with code parameter $[[n,k,d]]$, where $d$ is the code distance, and $d$ is also the maximum number of error that the code can correct, with the knowledge of the expectation value $
     E(\textbf{RSG}(\mathbf{G}(\mathbf{C},n,0,T)[1,T,p]))$ of error number on syndrome qubit after some error propagation in $\textbf{C}$, the logical error rate is bounded by
     \begin{equation}
         P_L(\textbf{RSG}(\mathbf{G}(\mathbf{C},n,0,T)[1,T,p]),d) \leq \frac{ E(\textbf{RSG}(\mathbf{G}(\mathbf{C},n,0,T)[1,T,p]))}{d+1}
         \label{eq:ProbLogical}
     \end{equation}
\end{theorem}
\begin{proof}
    By the definition of code distance $d$,
    \begin{equation}
        P_L(\textbf{RSG}(\mathbf{G}(\mathbf{C},n,0,T)[1,T,p]),d)=\sum_{k=d+1}^n P(k)
        \label{eq:DefProbLog}
    \end{equation}
    By Markov inequality, 
    \begin{equation}
        \sum_{k=d+1}^n P(k) \leq \frac{E(k)}{d+1}
        \label{eq:Markov}
    \end{equation}
    Combine \autoref{eq:DefProbLog} with \autoref{eq:Markov} and we have the result in \autoref{eq:ProbLogical}.
\end{proof}

\autoref{thm:Edisjoint} and \autoref{thm:VarDisj} indicate that it is easy to calculate and understand the expectation and variance of some $\textbf{RSG}$ once each disjoint component of $\textbf{RSG}$ is known.

Next, we start to analyze the expectation value and variance of the distribution with bounded propagation. A simple observation is through the degree of source vertices $\textbf{Source}(\textbf{G})$ in $\textbf{RSG}(\textbf{G})$:
\begin{enumerate}
    \item When a source vertex $V$ has degree $d$ equals $1$, it only connects with on syndrome vertex in $\textbf{Syn}(\textbf{G})$, there is not error propagation caused by $V$.
    \item When a source vertex $V$ has degree $d>1$, the bit-flip error propagates to $d$ different syndrome vertices.
\end{enumerate}

\begin{definition}(Shift of distribution)
    \label{df:Shift}
    For a propagation graph $\mathbf{G}(\mathbf{C},n,m,T)[t_1,t_2]$,  we define the shift of error number distribution caused by $\textbf{G}$ as the difference between its expectation value of number of bit-flip errors with that in $\mathbf{G}(\mathbf{C}_0,n,0,T)[t_1,t_2]$, where $\mathbf{C}_0$ is another circuit without any quantum gates.
    \begin{equation}
        \mathcal{T}(\mathbf{G}(\mathbf{C},n,m,T)[t_1,t_2,p])=E(\textbf{RSG}(\mathbf{G}(\mathbf{C},n,m,T)[1,T,p]))-E(\textbf{RSG}(\mathbf{G}(\mathbf{C}_0,n,0,T)[1,T,p]))
        \label{eq:ShiftOfDistribution}
    \end{equation}
\end{definition}

\begin{lemma}(Expectation value for $\textbf{C}_0$)
    \label{lem:expC0}
    The expectation value and variance for error distribution when there is no $CNOT$ gate is:
    \begin{equation}
       E(\textbf{RSG}(\mathbf{G}(\mathbf{C}_0,n,0,T)[1,T,p]))=\frac{n}{2}-\frac{n}{2}(1-2p)^T
        \label{eq:expC0}
    \end{equation}
    \begin{equation}
       Var(\textbf{RSG}(\mathbf{G}(\mathbf{C}_0,n,0,T)[1,T,p]))=\frac{n}{4}-\frac{n}{4}(1-2p)^{2T}
        \label{eq:varC0}
    \end{equation}    
    
\end{lemma}
\begin{proof}
An $n$-qubit $\textbf{RSG}(\mathbf{C}_0)$ can be viewed as $n$ disjoint subgraphs, each composed of exactly $T$ source vertices and $1$ syndrome vertex. An example when $n=4$ is shown in \autoref{fig:Disjoint}. According to \autoref{thm:ExpecDisj} and \autoref{thm:VarDisj}, we just need to calculate the expectation value and variance of each of these subgraphs.
According to \autoref{thm:weightedCouting}, the probability $P(k)$ is given by
\begin{equation}
    P(k)=\sum_{\substack{\textbf{G'} \subset \textbf{RSG}(\textbf{G}) \\ \textbf{G'} \text{is error subgraph}  \\|\textbf{G'} \cap \textbf{Syn}(G)|=k}} P(\textbf{G'} \cap \textbf{Source}(\mathbf{G}))
\end{equation}
By \autoref{thm:spanningsubgraphparity},  $\textbf{G'}$ is an error subgraph if and only if every vertex in $\textbf{G'} \cap \textbf{Syn}(\mathbf{G})$ has odd degree. Since there are only $1$ syndrome vertex, we can simply count the possibility of choosing an odd number of source vertices, given by:

\begin{equation}
    P(1)=P(\text{odd})=\sum_{k=0}^{\lfloor T/2 \rfloor} \binom{T}{2k+1}p^{2k+1}(1-p)^{T-2k-1}
\end{equation}
The expectation value is:
\begin{equation}
    \begin{split}
          E(k)&=0\times(1-P(1))+1\times P(1)\\
          &=P(1)\\
          &=\sum_{k=0}^{\lfloor T/2 \rfloor} \binom{T}{2k+1}p^{2k+1}(1-p)^{T-2k-1}  
    \end{split}
\end{equation}
The value of the variance is:
\begin{equation}
    \begin{split}
            Var(k)&=E(k^2)-E(k)^2\\
                  &=(1^2\times P(1))-P(1)^2\\
                  &=P(1)(1-P(1))\\
                  &=P(1)P(0)\\
                  &=\sum_{k_1=0}^{\lfloor T/2 \rfloor} \binom{T}{2k_1+1}p^{2k_1+1}(1-p)^{T-2k_1-1}  \sum_{k_2=0}^{\lfloor T/2 \rfloor} \binom{T}{2k_2}p^{2k_2}(1-p)^{T-2k_2}
    \end{split}
\end{equation}

By \autoref{lem:expC0} and \autoref{thm:VarDisj}, 
\begin{equation}
     E(\textbf{RSG}(\mathbf{G}(\mathbf{C}_0,n,0,T)[1,T,p]))=n\sum_{k=0}^{\lfloor T/2 \rfloor} \binom{T}{2k+1}p^{2k+1}(1-p)^{T-2k-1}  
\end{equation}
By \autoref{lem:EvenOddSum}, the above summation is:

\begin{equation}
   E(\textbf{RSG}(\mathbf{G}(\mathbf{C}_0,n,0,T)[1,T,p]))=\frac{n}{2}-\frac{n}{2}(1-2p)^T
\end{equation}

\begin{equation}
     Var(\textbf{RSG}(\mathbf{G}(\mathbf{C}_0,n,0,T)[1,T,p]))=n\sum_{k_1=0}^{\lfloor T/2 \rfloor} \binom{T}{2k_1+1}p^{2k_1+1}(1-p)^{T-2k_1-1}  \sum_{k_2=0}^{\lfloor T/2 \rfloor} \binom{T}{2k_2}p^{2k_2}(1-p)^{T-2k_2}
\end{equation}

By \autoref{lem:EvenOddSum}, the above summation is:

\begin{equation}
\begin{split}
       Var(\textbf{RSG}(\mathbf{G}(\mathbf{C}_0,n,0,T)[1,T,p]))&=n(\frac{1}{2}+\frac{1}{2}(1-2p)^T)(\frac{1}{2}-\frac{1}{2}(1-2p)^T)\\
       &=\frac{n}{4}-\frac{n}{4}(1-2p)^{2T}
\end{split}
\label{eq:varC0}
\end{equation}    

\end{proof}

\begin{corollary}(Convergence of Expectation value for $\textbf{C}_0$)
    \label{cor:convergeC0}
    
    \begin{equation}
       \lim_{T\rightarrow \infty} E(\textbf{RSG}(\mathbf{G}(\mathbf{C}_0,n,0,T)[1,T,p]))=\frac{n}{2}
        \label{eq:expC0Converg}
    \end{equation}
    \begin{equation}
       \lim_{T\rightarrow \infty} Var(\textbf{RSG}(\mathbf{G}(\mathbf{C}_0,n,0,T)[1,T,p]))=\frac{n}{4}
        \label{eq:varC0Converg}
    \end{equation}   
\end{corollary}

\begin{table}[h!]
    \centering
    \begin{tabular}{|c|c|c|c|c|c|c|}
        \hline
        Index & $Q_1[1]$ & $Q_1[2]$  & $Q_3[1]$ & $Q_3[2]$ & $\#$ Errors in  $Q_1[3],Q_2[3]$ & Probability \\ \hline
        1 &  0 &  0 &  0 &  0& 0 &  $(1-p)^4$\\ \hline
        2 &  0 &  0 &  0 &  1& 1 &  $p(1-p)^3$\\ \hline
        3 &  0 &  0 &  1 &  0& 1 &  $p(1-p)^3$\\ \hline
        4 &  0 &  0 &  1 &  1& 0 &  $p^2(1-p)^2$\\ \hline
        5 &  0 &  1 &  0 &  0& 1 &  $p(1-p)^3$\\ \hline
        6 &  0 &  1 &  0 &  1& 2 &  $p^2(1-p)^2$\\ \hline
        7 &  0 &  1 &  1 &  0& 2 &  $p^2(1-p)^2$\\ \hline
        8 &  0 &  1 &  1 &  1& 1 &  $p^3(1-p)$\\ \hline
        9 &  1 &  0 &  0 &  0& 2 &  $p(1-p)^3$\\ \hline
        10 & 1 &  0 &  0 &  1& 1 &  $p^2(1-p)^2$\\ \hline
        11 & 1 &  0 &  1 &  0& 1 &  $p^2(1-p)^2$\\ \hline
        12 & 1 &  0 &  1 &  1& 2 &  $p^3(1-p)$\\ \hline
        13 & 1 &  1 &  0 &  0& 1 &  $p^2(1-p)^2$\\ \hline
        14 & 1 &  1 &  0 &  1& 0 &  $p^3(1-p)$\\ \hline
        15 & 1 &  1 &  1 &  0& 0 &  $p^3(1-p)$\\ \hline
        16 & 1 &  1 &  1 &  1& 1 &  $p^4$\\ \hline

    \end{tabular}
    \caption{All $16$ cases for the number of error for parallel transversal circuit. }
    \label{tab:transversal}
\end{table}

\begin{lemma}(Expectation value for parallel transversal circuit)
    \label{lem:expTrans}
    The expectation value and variance for the error distribution for the parallel transverse circuit $CNOT$, as shown in \autoref{fig:Trans},$\textbf{C}_T$ are \begin{equation}
       E(\textbf{RSG}(\mathbf{G}(\mathbf{C}_T,n,0,T=2)[1,T=2,p]))=\frac{n}{2}(5p-8p^2+4p^3)
       \label{eq:ExpTransversal}
    \end{equation}
    The variance of the error number distribution is given by:
    \begin{equation}
               Var(\textbf{RSG}(\mathbf{G}(\mathbf{C}_T,n,0,T=2)[1,T=2,p]))=\frac{n}{2}(7p-35p^2+84p^3-104p^4+64p^5-16p^6)
    \end{equation}
    The expectation value when there are not transversal $CNOT$ is
    \begin{equation}
       E(\textbf{RSG}(\mathbf{G}(\mathbf{C}_0,n,0,T=2)[1,T=2,p]))=n(2p-2p^2)
       \label{eq:ExpNoCircuit}
    \end{equation}    
    The variance of the error number distribution with empty circuit is given by:
    \begin{equation}
       Var(\textbf{RSG}(\mathbf{G}(\mathbf{C}_0,n,0,T=2)[1,T=2,p]))=n(2p-6p^2+8p^3-4p^4)
    \end{equation}
\end{lemma}
\begin{proof}
For each disjoint subgraph with $4$ source vertices and two syndrome vertices, We can enumerate all $16$ error pattern and calculate the probability:
\begin{equation}
\begin{split}
      P(0)&=(1-p)^4+p^2(1-p)^2+2p^3(1-p)=1-4p+7p^2-4p^3\\
      P(1)&=3p(1-p)^3+p^3(1-p)+3p^2(1-p)^2+p^4=3p-6p^2+4p^3\\
      P(2)&=2p^2(1-p)^2+p(1-p)^3+p^3(1-p)=p-p^2
\end{split}
\end{equation}
By \autoref{thm:ExpecDisj}, the expectation value for the number of error is:
\begin{equation}
\begin{split}
        E(\textbf{RSG})&=\frac{n}{2}(0\times P(0)+1\times P(1)+2 \times  P(2))\\
        &=\frac{n}{2}(5p-8p^2+4p^3)
\end{split}
\end{equation}
To calculate the variance, we need $\sum_{k=0}^2 k^2P(k)$, which is given by:
\begin{equation}
\begin{split}
      \sum_{k=0}^2 k^2P(k)&=0^2P(0)+1^2P(1)+2^2P(2)\\
                          &=7p-10p^2+4p^3
\end{split}
\end{equation}
The variance for parallel transversal circuit is thus:
\begin{equation}
\begin{split}
    Var(\textbf{RSG})&=\frac{n}{2}(\sum_{k=0}^2 k^2P(k)-(5p-8p^2+4p^3)^2)\\
    &=\frac{n}{2}(7p-35p^2+84p^3-104p^4+64p^5-16p^6)
\end{split}
\end{equation}

On the other hand, if there are no $CNOT$ gates, and $T=2$, we just need to analyze $n$ disjoint subgraph, each only with $2$ source vertices and one syndrome vertices:

\begin{equation}
\begin{split}
      P(0)&=p^2+(1-p)^2=2p^2-2p+1\\
      P(1)&=2p(1-p)=2p-2p^2\\
      E(\textbf{RSG})&=1\times P(1)=2p-2p^2
\end{split}
\end{equation}

By \autoref{thm:ExpecDisj}, the expectation value for the number of error is:
\begin{equation}
\begin{split}
        E(\textbf{RSG}(\textbf{C}_0))=n(2p-2p^2)
\end{split}
\end{equation}
$\sum_{k=0}^1 k^2P(k)$ is given as follows:
\begin{equation}
\begin{split}
      \sum_{k=0}^1 k^2P(k)&=0^2P(0)+1^2P(1)\\
                          &=2p-2p^2
\end{split}
\end{equation}
The variance for an empty circuit with $T=2$ is thus:
\begin{equation}
\begin{split}
    Var(\textbf{RSG})&=n(\sum_{k=0}^1 k^2P(k)-(2p-2p^2)^2)\\
    &=n(2p-6p^2+8p^3-4p^4)
\end{split}
\end{equation}
\end{proof}
The most trivial cases of error propagation is the accumulation of errors throughout time in the circuit, as we have already observed through a numerical experiment. Now, we also give an exact measure of the change in error distribution with the accumulation of errors by \autoref{df:Shift}.

\begin{theorem}(Shift of distribution with time in empty circuit)
    \label{thm:shifttime}
    In an empty circuit with $n$ qubit, we define the shift of distribution from time $T_1$ to $T_2$ as
    \begin{equation}
        \mathcal{T}(p,n,T_1,T_2)=E(\textbf{RSG}(\mathbf{G}(\mathbf{C},n,0,T_2)[1,T_2,p]))-E(\textbf{RSG}(\mathbf{G}(\mathbf{C},n,0,T_1)[1,T_1,p]))
        \label{eq:ShiftOfDistributionwithTime}
    \end{equation}
    Then 
    \begin{equation}
        \mathcal{T}(p,n,T_1,T_2)=\frac{n}{2}[(1-2p)^{T_1}-(1-2p)^{T_2}]
        \label{eq:TimeShift}
    \end{equation}
    $\mathcal{T}(p,n,T_1,T_2)$ is always positive when $T_2 \geq T_1$ and $p<\frac{1}{2}$. In other words, the function $E(\textbf{RSG}(\mathbf{G}(\mathbf{C}_0,n,0,T)[1,T,p]))$ increases monotonically with respect to $T$ when we set $n,p$ fixed.
\end{theorem}
\begin{proof}
By \autoref{eq:expC0} and the definition of \autoref{df:Shift}, we have 
\begin{equation}
\begin{split}
            \mathcal{T}(p,n,T_1,T_2)&=E(\textbf{RSG}(\mathbf{G}(\mathbf{C},n,0,T_2)[1,T_2,p]))-E(\textbf{RSG}(\mathbf{G}(\mathbf{C},n,0,T_1)[1,T_1,p]))\\
            &=n(\frac{1}{2}-\frac{1}{2}(1-2p)^{T_2})-n(\frac{1}{2}-\frac{1}{2}(1-2p)^{T_1})\\
            &=\frac{1}{2}[(1-2p)^{T_1}-(1-2p)^{T_2}]
\end{split}
\end{equation}
Which is always positive when $T_2>T_1$ and $p<\frac{1}{2}$.
\end{proof}

\begin{theorem}(Shift of distribution with error probability in empty circuit)
    \label{thm:shiftp}
For two empty circuit with different error rate $p_2$ and $p_1$, the probability distribution shift caused by the change of error rate is defined as:
    \begin{equation}
        \mathcal{T}(T,n,p_1,p_2)=E(\textbf{RSG}(\mathbf{G}(\mathbf{C},n,0,T)[1,T,p_2]))-E(\textbf{RSG}(\mathbf{G}(\mathbf{C},n,0,T)[1,T,p_1]))
        \label{eq:ShiftOfDistributionwithprob}
    \end{equation}
    By \autoref{lem:expC0}, \autoref{eq:ShiftOfDistributionwithprob} is:
    \begin{equation}
        \mathcal{T}(T,n,p_1,p_2)=\frac{1}{2}n[(1-2p_1)^{T}-(1-2p_2)^{T}]
        \label{eq:probShift}
    \end{equation}
    When $p_1<p_2<\frac{1}{2}$, $\mathcal{T}(T,n,p_1,p_2)$ is strictly positive, the shift of distribution with time increase with error rate $p$.
\end{theorem}

\begin{figure}[h!]
\fbox{
\begin{minipage}{0.49\textwidth}
     \centering
\resizebox{0.9\columnwidth}{!}
{
\begin{quantikz}
  Q_1    & \slice{$T=1$} & \ctrl{1}& \ctrl{2}   & \ctrl{3}  &     &      &   &    &    & \targ{} &     &   \targ{}   &  \targ{}&\slice{$T=2$} &\gate[4,disable auto
height]{\verticaltext{Repeat}} &&\\
 Q_2     &               & \targ{}  &  &   &\ctrl{1}  & \ctrl{2} &  &   & \targ{}  &  &  \targ{} &  &  \ctrl{-1}  &&&&\\
 Q_3     &               &        & \targ{}   &  &  \targ{} &  &\ctrl{1} &\targ{}  &   &      &  \ctrl{-1}    &  \ctrl{-2}  & &&&&\\
  Q_4    &               &         &  &   \targ{}      & & \targ{} & \targ{}&   \ctrl{-1} & \ctrl{-2} &   \ctrl{-3}      & & & &&&&
\end{quantikz}
}
 \caption{The 4-qubit fully connected circuit with $T=3$. In $T=1$ and $T=2$, every two pairs of qubits are connected by a $CNOT$ gate.}\label{fig:CircuitFully}
\end{minipage}
}
\hfill
\fbox{
\begin{minipage}{0.5\textwidth}
     \centering
\resizebox{0.3\columnwidth}{!}
{
\begin{tikzpicture}
\tikzstyle{every node}=[draw, minimum size=1cm, inner sep=0.1cm]
\node (A1) at (0,-0.5) {$Q_1[1]$};
\node (A2) at (0,-2) {$Q_1[2]$};
\node (A3) at (0,-3.5) {$Q_2[1]$};
\node (A4) at (0,-5) {$Q_2[2]$};
\node (A5) at (0,-6.5) {$Q_3[1]$};
\node (A6) at (0,-8) {$Q_3[2]$};
\node (A7) at (0,-9.5) {$Q_4[1]$};
\node (A8) at (0,-11) {$Q_4[2]$};

\node (A9) at (3,-3.5) {$Q_1[3]$};
\node (A10) at (3,-5) {$Q_2[3]$};
\node (A11) at (3,-6.5) {$Q_3[3]$};
\node (A12) at (3,-8) {$Q_4[3]$};

\draw[->,red]  (A1) -- (A9);
\draw[->,red]  (A2) -- (A9);
\draw[->,red]  (A3) -- (A9);
\draw[->,red]  (A4) -- (A9);
\draw[->,red]  (A5) -- (A9);
\draw[->,red]  (A6) -- (A9);
\draw[->,red]  (A7) -- (A9);
\draw[->,red]  (A8) -- (A9);

\draw[->,blue]  (A1) -- (A11);
\draw[->,blue]  (A2) -- (A11);
\draw[->,blue]  (A3) -- (A11);
\draw[->,blue]  (A4) -- (A11);
\draw[->,blue]  (A5) -- (A11);
\draw[->,blue]  (A6) -- (A11);
\draw[->,blue]  (A7) -- (A11);
\draw[->,blue]  (A8) -- (A11);

\draw[->,green]  (A1) -- (A10);
\draw[->,green]  (A2) -- (A10);
\draw[->,green]  (A3) -- (A10);
\draw[->,green]  (A4) -- (A10);
\draw[->,green]  (A5) -- (A10);
\draw[->,green]  (A6) -- (A10);
\draw[->,green]  (A7) -- (A10);
\draw[->,green]  (A8) -- (A10);

\draw[->,pink]  (A1) -- (A12);
\draw[->,pink]  (A2) -- (A12);
\draw[->,pink]  (A3) -- (A12);
\draw[->,pink]  (A4) -- (A12);
\draw[->,pink]  (A5) -- (A12);
\draw[->,pink]  (A6) -- (A12);
\draw[->,pink]  (A7) -- (A12);
\draw[->,pink]  (A8) -- (A12);
\end{tikzpicture}
}
\caption{The \textbf{RSG} of the fully connected graph in \autoref{fig:CircuitFully}.}
\label{fig:4qRSGFully}
\end{minipage}
}
\caption{The circuit and the $\textbf{RSG}$ of a 4-qubit fully connected circuit.}
\label{fig:DisjointFully}
\end{figure}

An interesting but trivial case is the fully connected circuit. The $\textbf{RSG}$ of a fully connected graph, as shown in \autoref{fig:DisjointFully}, contains all edges between any source vertex and any syndrome vertex.

\begin{lemma}(Error distribution for fully connected graph)
    \label{lem:distFully}
    The error distribution in a fully connected graph is:
    \begin{equation}
       P(r)=\begin{cases}
           &0, \qquad   r\neq 0,n\\
           &\sum_{k=0}^{\lfloor Tn/2 \rfloor} \binom{Tn}{2k}p^{2k}(1-p)^{Tn-2k}=\frac{1}{2}+\frac{1}{2}(1-2p)^{Tn}, \qquad   r=0\\
           &\sum_{k=1}^{\lfloor Tn/2 \rfloor} \binom{Tn}{2k+1}p^{2k+1}(1-p)^{Tn-2k-1}=\frac{1}{2}-\frac{1}{2}(1-2p)^{Tn}, \qquad   r=n
       \end{cases}
        \label{eq:distfully}
    \end{equation}
\end{lemma}
\begin{proof}
When the circuit is fully connected, every syndrome vertex has exactly the same behavior, either they all flip with an error or they all stay in the correct state. 
Because the degree of each syndrome vertex is $Tn$, it's trivial to get the probability form in \autoref{eq:distfully}.
\end{proof}

\begin{lemma}(Expectation value for fully connected graph)
    \label{lem:expfully}
    The expectation value and variance for error distribution for a fully connected gate is:
    \begin{equation}
       E(\textbf{RSG}(\mathbf{G}(\mathbf{C},n,m,T)[1,T]))=\frac{n}{2}-\frac{n}{2}(1-2p)^{Tn}
        \label{eq:expfully}
    \end{equation}
    \begin{equation}
       Var(\textbf{RSG}(\mathbf{G}(\mathbf{C},n,m,T)[1,T]))=n^2 (\frac{1}{4}-(1-2n)^{2Tn})
        \label{eq:varfully}
    \end{equation}    
\end{lemma}
\begin{proof}
By \autoref{lem:distFully},  only $P(0)$ and $P(n)$ are non-zero in the distribution. The expectation value is given by:
\begin{equation}
    \begin{split}
             E(\textbf{RSG}(\mathbf{G}(\mathbf{C},n,m,T)[1,T]))&=0P(0)+nP(n)\\
             &=\frac{n}{2}-n(1-2p)^{Tn}
    \end{split}
\end{equation}
The variance is:
\begin{equation}
    \begin{split}
             Var(\textbf{RSG}(\mathbf{G}(\mathbf{C},n,m,T)[1,T]))&=0^2P(0)+n^2P(n)- E(\textbf{RSG}(\mathbf{G}(\mathbf{C},n,m,T)^2\\
             &=n^2P(n)(1-P(n))\\
             &=n^2P(n)P(0)\\
             &=n^2 (\frac{1}{4}-(1-2n)^{2Tn})
    \end{split}
\end{equation}
\end{proof}

\begin{corollary}(Shift of distribution by fully connected graph) The expectation value shift for a fully connected graph is given by:
\begin{equation}
\mathcal{T}(\mathbf{G}(\mathbf{C},n,m,T)[t_1,t_2,p])=\frac{n}{2}(1-2p)^T-\frac{1}{2}n(1-2p)^{Tn}
\label{eq:ShiftByFully}
\end{equation}
When $p>\frac{1}{Tn}$, $\mathcal{T}(\mathbf{G}(\mathbf{C},n,m,T)[t_1,t_2,p])>0$
\label{cor:shiftfully}
\end{corollary}
\begin{proof}
    We can get \autoref{eq:ShiftByFully} by combining \autoref{eq:expfully} with \autoref{eq:expC0}. Since the first term is equivalently $T \rightarrow nT$ in empty circuit with $n$ qubit, by \autoref{thm:shifttime}, when $p < \frac{1}{2}$,  \autoref{eq:ShiftByFully} is strictly positive.
\end{proof} 

\begin{corollary}(Shift of distribution by parallel transversal graph)
\label{cor:shifttransversal}
\begin{equation}
    \mathcal{T}(\mathbf{G}(\mathbf{C}_T,n,n/2,2)[0,T,p])=\frac{n}{2}p(2p-1)^2 \geq 0
    \label{eq:ShiftParallelTrans}
\end{equation}
\autoref{eq:ShiftParallelTrans} is bounded by $\frac{n}{27}$ when $p<\frac{1}{2}$ and the maximum is taken when $p=\frac{1}{6}$.
\end{corollary}
\begin{proof}
By \autoref{lem:expTrans}, the 
\begin{equation}
\begin{split}
        \mathcal{T}(\mathbf{G}(\mathbf{C}_T,n,n/2,2)[0,T,p])&=\frac{n}{2}(5p-8p^2+4p^3)-n(2p-2p^2)\\
        &=\frac{n}{2}p(2p-1)^2
        \label{eq:Pfunc}
\end{split}
\end{equation}
We can also calculate the maximum given by \autoref{eq:ShiftParallelTrans} with respect to $p$:
\begin{equation}
    \frac{d\mathcal{T}(\mathbf{G}(\mathbf{C}_T,n,n/2,2)[0,T,p])}{dp}=6n(p^2-\frac{2}{3}p+\frac{1}{12})=6n(p-\frac{1}{2})(p-\frac{1}{6})
    \label{eq:DerivativeOfP}
\end{equation}
\begin{figure}[h!]
    \centering
    \begin{minipage}{0.45\linewidth}
        \centering
        \includegraphics[width=\linewidth]{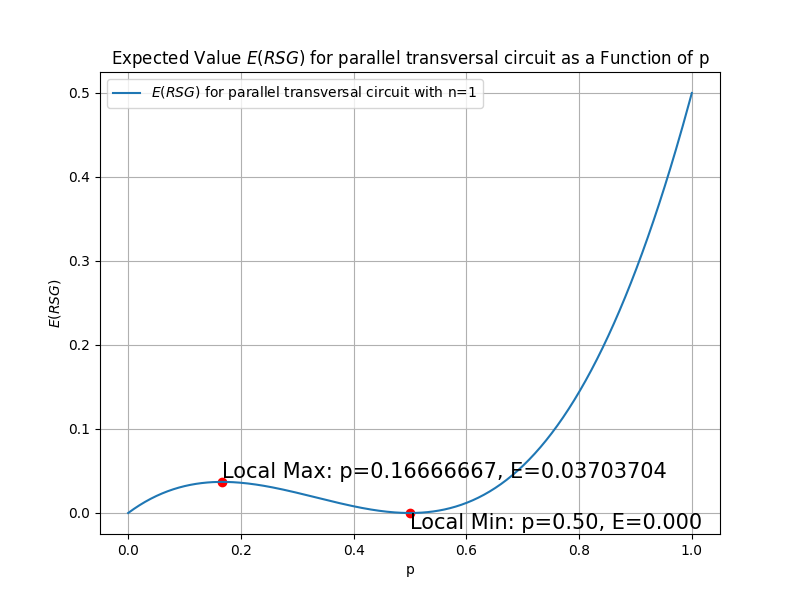}
        \caption{Plot of function \autoref{eq:Pfunc}. When $p=\frac{1}{6}$, the shift of error distribution has the local maximum value $\frac{1}{27}n$. The local minimum is when $p=\frac{1}{2}$, the shift of distribution is $0$.When $p>\frac{1}{2}$, the shift of distribution keeping increasing with $p$.}       \label{fig:distributionShiftTransversal}
    \end{minipage}
    \hfill
    \begin{minipage}{0.5\linewidth}
        \centering
        \includegraphics[width=\linewidth]{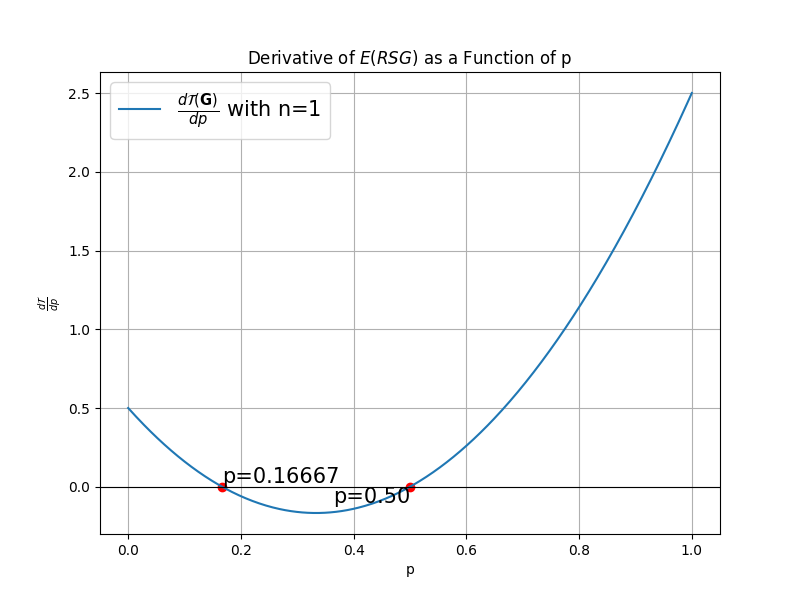}
        \caption{Plot of first order derivative function in \autoref{eq:DerivativeOfP}. When $p=\frac{1}{6}$ or $p=\frac{1}{2}$, the function has local maximum and minimum.}
        \label{fig:Derivative}
    \end{minipage}
\end{figure}
\autoref{eq:Pfunc} and \autoref{eq:DerivativeOfP} are plotted in \autoref{fig:distributionShiftTransversal} and \autoref{fig:Derivative}. We can draw the following important conclusion:
\begin{enumerate}
    \item The parallel transversal circuit always have non-negative shift of error, no matter what value $p$ takes.
    \item The shift of distribution for parallel transversal circuit is bounded by $\frac{n}{27}$ when $p$ is a small value. 
\end{enumerate}
\end{proof}





\section{Decision version of the problem is in $\mathcal{P}$}

To understand the complexity of calculating the error distribution after propagation, we start with a easier dicision version. The decision version of the counting problem is stated as follows: 

\textbf{
Given a circuit \textbf{C} ,its Error propagation space-time graph \textbf{G},  and an error syndrome $\textbf{S} \subset \textbf{Syn}(\textbf{G})$, does there exist an error pattern of the source qubit that generates the error syndrome?
}

The decision version of the problem belongs to $\mathcal{NP}$ because, given a certificate $w$ of the error syndrome, we can verify whether it is a source of the syndrome or not in polynomial time.

For a given $\textbf{RSG}(\textbf{G})$, If we assign each of the $\textbf{Source}(\textbf{G})$ vertices a boolean variable, denoted as $\textbf{Source}(\textbf{G})=\{v_1,v_2,\cdots,v_m\}$. Assume that $|\textbf{Source}(\textbf{G})|=m$, $m=(T-1)n$,$|\textbf{S}|=p$ and the set of vertices with error in  $\textbf{Syn}(\textbf{G})$, \quad $\textbf{S}=\{s_1,s_2,\cdots,s_p\}$, the set of vertices without error $\textbf{Syn}(\textbf{G})/\textbf{S}$. We can make $n$ clauses for all syndrome vertices\, in each clause, we use one different binary variables $v$ for each edge with an edge connecting vertex $v$ to the syndrome vertex. 

The whole boolean form of the satisfiable problem with respect to the syndrome $\textbf{S} \subset \textbf{Syn}(\textbf{G})$ is:

\begin{equation}
    \mathcal{F}(\textbf{RSG}(\textbf{G}),\textbf{S})=\big\{\bigwedge_{ \substack{ k=1 \\ s_k \in \textbf{S}}}^{p}( 0 \oplus \bigoplus_{\substack{v_i \in \textbf{Source}(\textbf{G}) \\s_k \in \textbf{S} \\ (v_i,s_k)\in \mathbf{E}(\textbf{G})} } v_i)\big\}\bigwedge\big\{\bigwedge_{\substack{k=p+1 \\ s_k \in \textbf{G}/\textbf{S}}}^{n}( 1 \oplus \bigoplus_{\substack{v_i \in \textbf{Source}(\textbf{G}) \\s_k \in \textbf{Syn}(\textbf{G})/\textbf{S} \\ (v_i,s_k)\in \mathbf{E}(\textbf{G})} } v_i)\big\}
    \qquad v_i \in \{0,1\}
    \label{eq:Boolean}
\end{equation}

We have the following theorem:

\begin{theorem}(Satisfiability equivalent to error pattern)
    \label{thm:satisfiabilityErrorpattern}
    Given a reverse spanning graph $\textbf{RSG}(\textbf{G})$ of a circuit $\textbf{C}$ and the subset of $\textbf{Syn}(\textbf{G})$ where errors occur, \autoref{eq:Boolean} is true if and only if the assignment of $\textbf{Source}(\textbf{G})=\{v_1,v_2,\cdots,v_m\}$ correspond to one possible error pattern in $\textbf{Source}(\textbf{G})$, such that an independent bitflip error at $v_i$ happens when $v_i=1$, while after error propagation, the error only remains in $\textbf{S}$.  
\end{theorem}
\begin{proof}
    By \autoref{thm:RSGthm}, the parity of the degree of the vertex of any syndrome vertex in $\textbf{RSG}$ determines the error on the vertex of the syndrome. 
    The degrees and their parity of all vertices in $\textbf{Syn}(\textbf{G})$, can be described as follows:

    For all $s_k \in \textbf{S}$ and $v_i \in \textbf{Source}(\textbf{G})$, 
    \begin{equation}
        \textbf{Deg}_{s_k}= \sum_{\substack{v_i \in \textbf{Source}(\textbf{G}) \\s_k \in \textbf{S} \\ (v_i,s_k)\in \mathbf{E}(\textbf{G}) \\ v_i=1 }} v_i, \qquad \textbf{Parity}(\textbf{Deg}_{s_k})=\bigoplus_{\substack{v_i \in \textbf{Source}(\textbf{G}) \\s_k \in \textbf{S} \\ (v_i,s_k)\in \mathbf{E}(\textbf{G}) \\ v_i=1 }} v_i
        \label{eq:oddDegree}
    \end{equation}
     For all $s_k \in \textbf{Syn}(\textbf{G})/\textbf{S}$ and $v_i \in \textbf{Source}(\textbf{G})$, 
    \begin{equation}
        \textbf{Deg}_{s_k}= \sum_{\substack{v_i \in \textbf{Source}(\textbf{G}) \\s_k \in \textbf{Syn}(\textbf{G})/\textbf{S} \\ (v_i,s_k)\in \mathbf{E}(\textbf{G})\\ v_i=1 } } v_i, \qquad \textbf{Parity}(\textbf{Deg}_{s_k})=\bigoplus_{\substack{v_i \in \textbf{Source}(\textbf{G}) \\s_k \in \textbf{Syn}(\textbf{G})/\textbf{S} \\ (v_i,s_k)\in \mathbf{E}(\textbf{G})\\ v_i=1 } } v_i
        \label{eq:evenDegree}
    \end{equation}   

    We denote all $v_i \in \{ v_i| v_i=1, v_i \in \textbf{Source}(\textbf{G}) \}$ together with $\textbf{S}$, a new subgraph $\textbf{G}' \subseteq \textbf{G}$. According to \autoref{thm:spanningsubgraphparity}, $\textbf{G}'$ is an error subgraph in $\textbf{RSG}(\textbf{G})$ if and only if all vertices in $\textbf{S}$ has odd parity, which also uniquely corresponds to a pattern of the random bit-flip error that occurs before the syndrome. The equivalent boolean formula is:
    \begin{equation}
         \big\{\bigwedge_{\substack{k=1 \\ s_k \in \textbf{S}}}^{p}( 0 \oplus \textbf{Parity}(\textbf{Deg}_{s_k}))\big\} \bigwedge \big\{\bigwedge_{\substack{k=p+1 \\ s_k \in \textbf{G}/\textbf{S}}}^{n}( 1 \oplus \textbf{Parity}(\textbf{Deg}_{s_k}))\big\}
         \label{eq:parityFinal}
    \end{equation}
    Combined with \autoref{eq:parityFinal}, \autoref{eq:evenDegree} and \autoref{eq:oddDegree}, we get

    \begin{equation}
    \mathcal{F}(\textbf{RSG}(\textbf{G}),\textbf{S})=\big\{\bigwedge_{ \substack{ k=1 \\ s_k \in \textbf{S}}}^{p}( 0 \oplus \bigoplus_{\substack{v_i \in \textbf{Source}(\textbf{G}) \\s_k \in \textbf{S} \\ (v_i,s_k)\in \mathbf{E}(\textbf{G})} } v_i)\big\}\bigwedge\big\{\bigwedge_{\substack{k=p+1 \\ s_k \in \textbf{G}/\textbf{S}}}^{n}( 1 \oplus \bigoplus_{\substack{v_i \in \textbf{Source}(\textbf{G}) \\s_k \in \textbf{Syn}(\textbf{G})/\textbf{S} \\ (v_i,s_k)\in \mathbf{E}(\textbf{G})} } v_i)\big\}
    \qquad v_i \in \{0,1\}
    \end{equation}
    This finish the proof.
\end{proof}

We use the $4$-qubit transversal circuit depicted in \autoref{fig:RSGProgability} as an example of transforming a $\textbf{RGS}(\textbf{G})$ and the syndrome $\textbf{S}$ to equation \autoref{eq:Boolean}. Now $\textbf{Source}(\textbf{G})=\{v_{Q_1[1]},v_{Q_1[2]},v_{Q_2[1]},v_{Q_2[2]},v_{Q_3[1]},v_{Q_3[2]},v_{Q_4[1]},v_{Q_4[2]}\}$, $\textbf{Syn}(\textbf{G})=\{s_{Q_1[3]},s_{Q_2[3]},s_{Q_3[3]},s_{Q_4[3]} \}$ and $\textbf{S}=\{s_{Q_1[3]},s_{Q_2[3]}\}$. \autoref{eq:Boolean} is:

\begin{equation}
    \begin{split}
          \mathcal{F}(\textbf{RSG}(\textbf{G}),\textbf{S})&=\big\{(0\oplus v_{Q_1[1]}\oplus v_{Q_1[2]}) \wedge (0\oplus v_{Q_2[1]}\oplus v_{Q_2[2]})\big\} \\
          &\bigwedge \big\{(1\oplus v_{Q_1[1]}\oplus v_{Q_3[1]} \oplus v_{Q_3[2]})  \wedge (1\oplus v_{Q_2[1]}\oplus v_{Q_4[1]} \oplus v_{Q_4[2]})\big\}
    \qquad v_i \in \{0,1\}  
    \end{split}
    \label{eq:BooleanExample}
\end{equation}
A satisfying solution for \autoref{eq:BooleanExample} is as provided in \autoref{fig:RSGProgability}, let $\{v_{Q_1[1]},v_{Q_2[1]},v_{Q_3[2]},v_{Q_4[2]}\}$ all be $1$ and the rest be $0$. The value of \autoref{eq:BooleanExample} after such assignment is 
\begin{equation}
      \begin{split}
          \mathcal{F}(\textbf{RSG}(\textbf{G}),\textbf{S})&=\big\{(0\oplus 1 \oplus 0) \wedge (0\oplus 1\oplus 0)\big\} \\
          &\bigwedge \big\{(1\oplus 1\oplus 0 \oplus 1)  \wedge (1\oplus 1 \oplus 0 \oplus 1)\big\}\\
          &=1 \wedge 1  \wedge 1 \wedge 1\\
          &=1
    \end{split}  
\end{equation}
Which means that we get a satisfying assignment for all $v_i$.

By \autoref{thm:spanningsubgraphparity}, we also know that the probability weight of each solution is given as
\begin{equation}
    p^{\sum_{i=1}^n v_i}(1-p)^{n-\sum_{i=1}^n v_i}=\prod_{i=1}^n p^{v_i}(1-p)^{1-v_i}
    \label{eq:weight}
\end{equation}

The decision version of the problem is finding an assignment of $\textbf{Source}(\textbf{G})=\{v_1,v_2,\cdots,v_m\}$ to make \autoref{eq:Boolean} true. 

The language is defined as $\mathcal{L}=\{ v_1v_2\cdots v_m |  \mathcal{F}(\textbf{RSG}(\textbf{G}),\textbf{S})=1 \}$

\begin{theorem}(Equivalent linear equation)
\label{thm:linearequation}
For any $\textbf{G},\textbf{S}$, where $|\textbf{S}|=p$, and $|\textbf{Source}(\textbf{G})|=m$. We can create a matrix $\mathcal{M}(\textbf{G}) \in \mathbb{F}_2^{n \times m}$ such that $\mathcal{M}[i][j]=1$ if and only if $(v_j,s_i) \in \textbf{E}(\textbf{G})$, and a vector $\Vec{\textbf{b}} \in \mathbb{F}_2^{n \times 1}$, such that $\textbf{b}[i]=1$ if and only if $s_i \in \textbf{S}$. 

The decision problem of   $\mathcal{L}=\{ v_1v_2\cdots v_m |  \mathcal{F}(\textbf{RSG}(\textbf{G}),\textbf{S})=1 \}$ is equivalent as finding one solution $\Vec{v} \in \mathbb{F}_2^m$ such that 
\begin{equation}
    \mathcal{M}(\textbf{G}) \Vec{v}=\Vec{\textbf{b}}
    \label{eq:lineareq}
\end{equation}
\autoref{eq:lineareq} is exactly the kernel of the matrix $\mathcal{M}(\textbf{G})$. The kernel space is not empty if and only if $\textbf{rank}(\mathcal{M}(\textbf{G}))<m$ because $\textbf{dim}(\textbf{Ker}(\mathcal{M}(\textbf{G})))=m-\textbf{rank}(\mathcal{M}(\textbf{G}))$.
\end{theorem}
\begin{proof}
    We use $\mathcal{M}(i)$ to represent the $i$-th row of the matrix $\mathcal{M}(\textbf{G})$.
    \autoref{eq:lineareq} is equivalent as
    \begin{equation}
       \forall  1 \leq i\leq m,\quad  <\mathcal{M}(i),\Vec{v}>=\bigoplus_{j=1}^n (\mathcal{M}[i][j] \cdot v_j)= \textbf{b}[i]
        \label{eq:innerproduct}
    \end{equation}
    According to the definition of $\mathcal{M}$, $\mathcal{M}[i][j]=1$ if and only if $(v_j,s_i) \in \textbf{E}(\textbf{G})$, and $\textbf{b}[i]=1$ if and only if $s_i \in \textbf{S}$, so \autoref{eq:innerproduct} is the same as 
    \begin{equation}
        \forall  1 \leq i\leq m,\quad <\mathcal{M}(i),\Vec{v}>=\bigoplus_{\substack{j=1 \\ (v_j,s_i) \in \textbf{E}(\textbf{G})}}^n  v_j= \begin{cases}
              & 0 ,\qquad  s_i \in \textbf{G}/\textbf{S}\\
              & 1 ,\qquad s_i \in \textbf{S}
        \end{cases}
        \label{eq:innerproduct}
    \end{equation}
    After separate $s_i$ into two groups $\textbf{S}=\{s_1,s_2,\cdots,s_p\}$ and $\textbf{Syn}(\textbf{G})/\textbf{S}=\{s_{p+1},s_{p+1},\cdots,s_n\}$, we have
    \begin{equation}
        1=\bigwedge_{\substack{i=1 \\ s_i \in \textbf{S}}}^{p}  (0 \oplus<\mathcal{M}(i),\Vec{v}>)=\bigwedge_{\substack{i=p+1 \\ s_i \in \textbf{G}/\textbf{S}}}^{n} (1 \oplus <\mathcal{M}(i),\Vec{v}>)
        \label{eq:innerproductsepa}
    \end{equation}

     Take \autoref{eq:innerproduct} into \autoref{eq:innerproductsepa}, we arrive at:
    \begin{equation}
    1=\big\{\bigwedge_{ \substack{ k=1 \\ s_k \in \textbf{S}}}^{p}( 0 \oplus \bigoplus_{\substack{v_i \in \textbf{Source}(\textbf{G}) \\s_k \in \textbf{S} \\ (v_i,s_k)\in \mathbf{E}(\textbf{G})} } v_i)\big\}\bigwedge\big\{\bigwedge_{\substack{k=p+1 \\ s_k \in \textbf{G}/\textbf{S}}}^{n}( 1 \oplus \bigoplus_{\substack{v_i \in \textbf{Source}(\textbf{G}) \\s_k \in \textbf{Syn}(\textbf{G})/\textbf{S} \\ (v_i,s_k)\in \mathbf{E}(\textbf{G})} } v_i)\big\}
    \qquad v_i \in \{0,1\}
    \end{equation}
    This is exactly \autoref{eq:Boolean}. By \autoref{thm:satisfiabilityErrorpattern}, we know that any assignment of $\Vec{v}$ corresponds uniquely to a valid error pattern in $\textbf{Source}(\textbf{G})$, and in the meantime is a solution to the linear equation \autoref{eq:lineareq}. This complete the proof.
\end{proof}

\autoref{thm:linearequation} implies the algorithm of finding a solution of  the decision problem.

\begin{enumerate}
    \item Generate the matrix $\mathcal{M}(\textbf{G})$ and $\Vec{\textbf{b}}$.  
    \item Use Gaussian elimination algorithm to get a solution of $\mathcal{M}(\textbf{G}) \Vec{\textbf{v}}=\Vec{\textbf{b}}$.
\end{enumerate}

We still use the $4$-qubit transversal circuit depicted in \autoref{fig:RSGProgability} as an example. Now $\textbf{Source}(\textbf{G})=\{v_{Q_1[1]},v_{Q_1[2]},v_{Q_2[1]},v_{Q_2[2]},v_{Q_3[1]},v_{Q_3[2]},v_{Q_4[1]},v_{Q_4[2]}\}$\footnote{The order of index is given by index of $Q_{k_1}[i]$ smaller than index of $Q_{k_2}[j]$ if $k_1 < k_2$. Otherwise, if $k_1=k_2$, compare $i$ with $j$.}, $\textbf{Syn}(\textbf{G})=\{s_{Q_1[3]},s_{Q_2[3]},s_{Q_3[3]},s_{Q_4[3]} \}$ and $\textbf{S}=\{s_{Q_1[3]},s_{Q_2[3]}\}$. By the definition in \autoref{thm:linearequation}, $n=4$ and $m=8$,the matrix $\mathcal{M}(\textbf{G},\textbf{S}) \in \mathbb{F}_2^{4 \times 8}$. 
\begin{equation}
   \mathcal{M}(\textbf{G}) =\begin{pmatrix}
       1 & 1 & 0 & 0& 0 & 0& 0 & 0\\
       0 & 0 & 1 & 1& 0 & 0& 0 & 0\\
       1 & 0 & 0 & 0& 1 & 1& 0 & 0\\
       0 & 0 & 1 & 0& 0 & 0& 1 & 1\\
   \end{pmatrix}, \qquad \Vec{\textbf{b}}=\begin{pmatrix}
       1\\
       1\\
       0\\
       0
   \end{pmatrix}
\end{equation}
We can verify that if we assign $\Vec{\textbf{v}}=[v_{Q_1[1]},v_{Q_1[2]},v_{Q_2[1]},v_{Q_2[2]},v_{Q_3[1]},v_{Q_3[2]},v_{Q_4[1]},v_{Q_4[2]}]=[1,0,1,0,0,1,0,1]$, then 
\begin{equation}
    \mathcal{M}(\textbf{G}) \Vec{\textbf{v}}=\begin{pmatrix}
       1 & 1 & 0 & 0& 0 & 0& 0 & 0\\
       0 & 0 & 1 & 1& 0 & 0& 0 & 0\\
       1 & 0 & 0 & 0& 1 & 1& 0 & 0\\
       0 & 0 & 1 & 0& 0 & 0& 1 & 1\\
   \end{pmatrix} \times \begin{pmatrix}
       1\\
       0\\
       1\\
       0\\
       0\\
       1\\
       0\\
       1
   \end{pmatrix}=\begin{pmatrix}
       1\\
       1\\
       0\\
       0
   \end{pmatrix}
\end{equation}

On the other hand, when $\Vec{\textbf{v}}$ is unknown, we can use gaussian elimination algorithm on the augmented matrix $\big(  \mathcal{M}(\textbf{G}) | \Vec{\textbf{b}} \big)$ to get the row echelon form:

\begin{equation}
    \begin{split}
           \big(  \mathcal{M}(\textbf{G}) | \Vec{\textbf{b}} \big)&= \begin{pmatrix}
       1 & 1 & 0 & 0& 0 & 0& 0 & 0 & 1\\
       0 & 0 & 1 & 1& 0 & 0& 0 & 0 & 1\\
       1 & 0 & 0 & 0& 1 & 1& 0 & 0 & 0\\
       0 & 0 & 1 & 0& 0 & 0& 1 & 1 & 0\\
   \end{pmatrix}\\
   \xrightarrow{\text{Row}_3-\text{Row}_1} &\begin{pmatrix}
       1 & 1 & 0 & 0& 0 & 0& 0 & 0 & 1\\
       0 & 0 & 1 & 1& 0 & 0& 0 & 0 & 1\\
       0 & 1 & 0 & 0& 1 & 1& 0 & 0 & 1\\
       0 & 0 & 1 & 0& 0 & 0& 1 & 1 & 0\\
   \end{pmatrix}\\
   \xrightarrow{\text{Row}_1-\text{Row}_3} &\begin{pmatrix}
       1 & 0 & 0 & 0& 1 & 1& 0 & 0 & 0\\
       0 & 0 & 1 & 1& 0 & 0& 0 & 0 & 1\\
       0 & 1 & 0 & 0& 1 & 1& 0 & 0 & 1\\
       0 & 0 & 1 & 0& 0 & 0& 1 & 1 & 0\\
   \end{pmatrix}\\
   \xrightarrow{\text{Row}_4-\text{Row}_2} &\begin{pmatrix}
       1 & 0 & 0 & 0& 1 & 1& 0 & 0 & 0\\
       0 & 0 & 1 & 1& 0 & 0& 0 & 0 & 1\\
       0 & 1 & 0 & 0& 1 & 1& 0 & 0 & 1\\
       0 & 0 & 0 & 1& 0 & 0& 1 & 1 & 1\\
   \end{pmatrix}\\
   \xrightarrow{\text{Row}_2-\text{Row}_4} &\begin{pmatrix}
       1 & 0 & 0 & 0& 1 & 1& 0 & 0 & 0\\
       0 & 0 & 1 & 0& 0 & 0& 1 & 1 & 0\\
       0 & 1 & 0 & 0& 1 & 1& 0 & 0 & 1\\
       0 & 0 & 0 & 1& 0 & 0& 1 & 1 & 1\\
   \end{pmatrix}
    \end{split}
\label{eq:Gaussian}
\end{equation}
In the final row echelon form in \autoref{eq:Gaussian}, because the left $4 \times 4$ sub-matrix has row echelon form, we can set  $[v_{Q_1[1]},v_{Q_1[2]},v_{Q_2[1]},v_{Q_2[2]}]=[0,1,0,1]$ and the rest part $0$, this must also be a solution to the equation $\mathcal{M}(\textbf{G},\textbf{S}) \Vec{\textbf{v}}=\Vec{\textbf{b}}$.

\begin{theorem}($\mathcal{L}=\{ v_1v_2\cdots v_m |  \mathcal{F}(\textbf{RSG}(\textbf{G}),\textbf{S})=1 \}$ is in $\mathcal{P}$)
\label{thm:LinP}
\end{theorem}
\begin{proof}
    The proof is trivial from \autoref{thm:linearequation}. In complexity theory, this is exactly the \textbf{XORSAT} problem.
\end{proof}

\section{Counting version of the problem}

In the context of quantum error correction with error propagation, we care about the distribution of error number. In this context, the decision version of the problem, when there are $k$ bitflip errors in total, is rephrased as \textit{Is $\{v_1v_2,\cdots v_m\}$ a true assignment to $\mathcal{L}=\{ v_1v_2\cdots v_m |  \mathcal{F}(\textbf{RSG}(\textbf{G}),\textbf{S})=1 \}$'', such that $|\textbf{S}|=k$?}
The constraint can be translated to:

\begin{equation}
k=\mathcal{N}(\textbf{RSG}(\textbf{G}))=\sum_{ \substack{ k=1 \\ s_k \in \textbf{Syn}(\textbf{G})}}^{n}( \bigoplus_{\substack{v_i \in \textbf{Source}(\textbf{G}) \\s_k \in  \textbf{Syn}(\textbf{S}) \\ (v_i,s_k)\in \mathbf{E}(\textbf{G}} } v_i)
\qquad v_i \in \{0,1\}
\label{eq:countingconstraint}
\end{equation}

The new language can be defined as $\mathcal{L}(k)=\{ v_1v_2\cdots v_m |  \mathcal{N}(\textbf{RSG}(\textbf{G}))=k\}$. Apparantly, when $k$ is a small constant or a value with a small constant gap with $n$, $\mathcal{L}(k)\in \mathcal{P}$, because there are at most polynomial way of different arrangement of syndrome $\textbf{S}$, and we just need to enumerate all possible $\textbf{S}$.

Now we define a new language $\mathcal{L}_{max}$, which is defined as

\begin{equation}
  \mathcal{L}_{max}=\{ v_1v_2\cdots v_m |  \mathcal{N}(\textbf{RSG}(\textbf{G}))=\max_{v_1v_2\cdots v_m}(\textbf{RSG}(\textbf{G}))) \}  
  \label{eq:Lmax}
\end{equation}

First, I claim that if $\mathcal{L}(k)$ is in $\mathcal{P}$, then $\mathcal{L}_{max}$ is also in $\mathcal{P}$.  
\begin{lemma}($\mathcal{L}(k)$ in $\mathcal{P}$ implies $\mathcal{L}_{max}$ in $\mathcal{P}$)
\label{lem:LmaxLk}
If $\mathcal{L}(k) \in \mathcal{P}$ for $1 \leq   k \leq n$, then $\mathcal{L}_{max} \in \mathcal{P}$.
\end{lemma}
\begin{proof}
  Now since we have oracle $\mathcal{O}(k)$ for language $\mathcal{L}(k)$, which runs in polynomial time and space and returns one witness $\Vec{v} \in \mathcal{L}(k)$ or halt and return an empty string if $\mathcal{L}(k)$ is empty. We can simply enumerate $k$, can call $\mathcal{O}(k)$ from $k=n$,$k=n-1$,$\cdots$, to $k=1$, so that if some $\mathcal{O}(k)$ returns a nonempty string, use this as the output for the witness of $\mathcal{L}_{max}$. Otherwise, the maximum number of satisfiable clauses must be less than $k$ and we keep calling $\mathcal{O}(k-1)$.  The algorithm is still polynomial. We thus find a polynomial reduction from $\mathcal{L}(max)$ to $\mathcal{L}(k)$. 
\end{proof}

The complexity of $\mathcal{L}_{max}$ is, however, $NP-hard$, because this is exactly the \textbf{MAX-XORSAT} problem, which is a famous hard problem in combinitorial optimization. By \autoref{lem:LmaxLk}, we know that $\mathcal{L}(k)$ must be $\mathcal{NP}$ complete in the worst case.

\FloatBarrier
\section{Random sampling algorithm}

From the previous sections we know that exact calculation for error distribution propagation is at least $NP-complete$ in the worst case. In this section, we introduce a random sampling algorithm. The algorithm can be described briefly by the following steps:
\begin{enumerate}
    \item For the input circuit $\mathcal{C}$, construct the $\textbf{EPSTG}$ of $\mathcal{C}$ as $\textbf{G}$ and $\textbf{RSG}(\textbf{G})$.  Construct the matrix $\mathcal{M}(\textbf{G}) \in \mathbb{F}_2^{n \times m}$ defined in \autoref{thm:linearequation}. Initialize a list of probability $P(0),P(1),P(2),\cdots,P(n)$ as $0$ in memory. 
    \item Sample all $m=(n-1)T$ errors by Bernoulli distribution for a given bit-flip error rate $p$. Generate the random error pattern vector $\Vec{v} \in \mathbb{F}_2^m$.
    \item Use equation \autoref{eq:lineareq} $\mathcal{M}(\textbf{G}) \Vec{v}=\Vec{\textbf{b}}$  to calculate the error syndrome vector $\Vec{\textbf{b}}$.
    \item Count the number of $1$ in $\Vec{\textbf{b}}$, denoted as $m$. Update $P(m) \rightarrow P(m)+1$
    \item Repeat step $2$,$3$,$4$ for $N$ times. 
    \item Normalize $P(0),P(1),P(2),\cdots,P(n)$ by $P(k) \rightarrow \frac{P(k)}{N}$ for $1\leq  k \leq n$. 
 \end{enumerate}

It is obvious that the time complexity for the above algorithm is $O(Tn^2+2mn+N(nm^2+n))$, which is polynomial with $n$. By \autoref{thm:RSGAlg}, we can construct $\textbf{RSG}$ given $\mathcal{C}$ in time $O(Tn^2+mn)$. Given $\textbf{RSG}$, we just need $O(mn)$ time to construct $\mathcal{M}(\textbf{G})$. In each of the $N$ iteration in the above algorithm, we need $O(nm^2)$ time to do matrix multiplication and $O(n)$ time to count the number of $1$ in $\Vec{b}$.

\section{Probability space for the problem}

In this section, we use probability theory to analyze the error propagation.  The probability space is defined as follows:

\begin{definition}(Probability space)
    \label{def:probabilitySpace}
    The probability space of error propagation is defined as:
    \begin{enumerate}
        \item The \textbf{sample space} $\Omega$ is the set of all possible outcome of error syndrome.
        \item The \textbf{event space} $\mathcal{F}$ is the set of all events of  error syndrome that we care about.
        \item The probability function $\textbf{P}$ which assign a value to each event in the event space $\mathcal{F}$. 
    \end{enumerate}
We use $X_i$ as the instance of an error occurring at location $i$, which is denoted by the $i$ th column in the matrix $\mathcal{M}(\textbf{G},\textbf{S}) \in \mathbb{F}_2^{m \times n}$, and $Y_i$ as the instance of an error happen at the syndrome qubit $i$, $\overline{Y_i}$ as the instance of no error occurring at the syndrome qubit $i$. $Y_i \cap Y_j$ represent the event that two errors occur in the qubit of the syndrome $i$ and $j$ together.   $Y_i \cup Y_j$ represent the event that at least one error occurs in qubit $i$ or qubit $j$.
\end{definition}

In the following part, we use $\textbf{deg}(Y_i)$ to represent the degree of syndrome qubit in $\textbf{RSG}$ represent $Y_i$. We use $\textbf{Source}(Y_i) \subset \Omega $ to represent the set ot source qubits that has one edge connected with the syndrome qubit represented by $Y_i$. Clearly, $|\textbf{Source}(Y_i)|=\textbf{deg}(Y_i)$.We use $\textbf{deg}(Y_i \cap Y_j)$ to represent the number of edges that $Y_i$  and $Y_j$ share the same source vertices, $\textbf{deg}(Y_i \cap Y_j)=|\textbf{Source}(Y_i) \cap \textbf{Source}(Y_j)|$. For example, in an empty circuit, because all $\{Y_i\}$ are independent, $\textbf{deg}(Y_i \cap Y_j)=0$ for $i \neq j$. In a transversal circuit, only two pair of syndrome qubits with a $CNOT$ gate has  $\textbf{deg}(Y_i \cap Y_j)=1$. 

The benefit of defining such a probability space in \autoref{def:probabilitySpace}, is that we unify the probability measure of multiple syndrome qubit outcome with single syndrome qubit outcome, by taking $\cup$ and $\cap$ between element in event space $\mathcal{F}$. Specifically, we already know that $\textbf{P}(Y_i)$ is easy to calculate:
\begin{equation}
    \begin{split}
            \textbf{P}(Y_i)&=\sum_{k=0}^{\lfloor \textbf{deg}(Y_i)/2 \rfloor} \binom{\textbf{deg}(Y_i)}{2k+1}p^{2k+1}(1-p)^{\textbf{deg}(Y_i)-2k-1}=\frac{1}{2}-\frac{1}{2}(1-2p)^{\textbf{deg}(Y_i)}\\
            \textbf{P}(\overline{Y_i})&=1-\textbf{P}(Y_i)=\sum_{k=0}^{\lfloor \textbf{deg}(Y_i)/2 \rfloor} \binom{\textbf{deg}(Y_i)}{2k}p^{2k}(1-p)^{\textbf{deg}(Y_i)-2k}=\frac{1}{2}+\frac{1}{2}(1-2p)^{\textbf{deg}(Y_i)}
    \end{split}
    \label{eq:probY}
\end{equation}
The tricky part of our problem, however, is that $Y_1,Y_2,\cdots,Y_n$ are not independent so its not trivial to measure the event such as $\textbf{P}(Y_1 \cap Y_3 \cap \cdots Y_n)$.

\begin{theorem}(Inclusion exclusion principle)
    \label{thm:IncluExclu}
    \begin{equation}
        \textbf{P}(\cup_{i=1}^n Y_i)=\sum_{k=1}^n (-1)^{k-1} \sum_{\substack{ I \subseteq \{1,2,\cdots,n\} \\ |I|=k }}\textbf{P}(\cap_{i \in I} Y_i)
        \label{eq:IncluExclu}
    \end{equation}
\end{theorem}

\begin{theorem}(Joint probability of two syndrome)
    The probability of there being at least one error in one of $Y_i$, $Y_j$ is
    \begin{equation}
        \textbf{P}(Y_i \cup Y_j)=\textbf{P}(Y_i)+\textbf{P}(Y_j)-\textbf{P}(Y_i \cap Y_j)
        \label{eq:Joint}
    \end{equation}
    Where $\textbf{P}(Y_i \cap Y_j)$ is the probability of both $Y_i$, $Y_j$ have errors:
    \begin{equation}
        \textbf{P}(Y_i \cap Y_j)=\frac{1}{4}-\frac{1}{4}(1-2p)^{\textbf{deg}(Y_i)}-\frac{1}{4}(1-2p)^{\textbf{deg}(Y_j)}+\frac{1}{4}(1-2p)^{\textbf{deg}(Y_i)+\textbf{deg}(Y_j)-2\textbf{deg}(Y_i \cap Y_j)}
        \label{eq:JointCap}
    \end{equation}
\end{theorem}
\begin{proof}
    \autoref{eq:Joint} is a direct result of inclusion exclusion principle when $n=2$. For \autoref{eq:JointCap},  we just need to calculate the probability for two cases:
    \begin{enumerate}
        \item There are odd number of errors in $\textbf{Source}(Y_i)/\textbf{Source}(Y_j)$ and $\textbf{Source}(Y_j)/\textbf{Source}(Y_i)$, and even number of errors in $\textbf{Source}(Y_j)\cap \textbf{Source}(Y_i)$
        \item There are even number of errors in $\textbf{Source}(Y_i)/\textbf{Source}(Y_j)$ and $\textbf{Source}(Y_j)/\textbf{Source}(Y_i)$, and odd number of errors in $\textbf{Source}(Y_j)\cap \textbf{Source}(Y_i)$
    \end{enumerate}
    Let $\textbf{deg}(Y_i)=x$,\quad $\textbf{deg}(Y_j)=x$,\quad $\textbf{deg}(Y_i \cap Y_j)=a$. $|\textbf{Source}(Y_i)/\textbf{Source}(Y_j)|=x-a$,$|\textbf{Source}(Y_j)/\textbf{Source}(Y_i)|=y-a$

    Then the probability of the two cases are:
    \begin{equation}
        \begin{split}
               \textbf{P}(Y_i \cap Y_j)&=[\frac{1}{2}-\frac{1}{2}(1-2p)^{x-a}]\times[\frac{1}{2}+\frac{1}{2}(1-2p)^{a}]\times [\frac{1}{2}-\frac{1}{2}(1-2p)^{y-a}]\\
               &+[\frac{1}{2}+\frac{1}{2}(1-2p)^{x-a}]\times[\frac{1}{2}-\frac{1}{2}(1-2p)^{a}]\times [\frac{1}{2}+\frac{1}{2}(1-2p)^{y-a}]\\
               &=\frac{1}{4}-\frac{1}{4}(1-2p)^x-\frac{1}{4}(1-2p)^y+\frac{1}{4}(1-2p)^{x+y-2a}
        \end{split}
    \end{equation}
\end{proof}

It remains open how to bound the error propagation with inclusion exclusion principle in the general case.

\section{From the perspective of Entropy}

In this section, we discuss the entropy increase caused by error propagation and try to connect it with the question that we are interested in. 

\begin{definition}(Entropy of error distribution)
    \label{def:entropy}
     The entropy of the syndrome error distribution $P(k)$ is defined as 
     \begin{equation}
         \mathcal{H}(P)=-\sum_{k=1}^n P(k)\log(P(k))
         \label{eq:entropy}
     \end{equation}
     If $P$ can be described by a continuous function, then the entropy can also be defined in the continuous integral format:
      \begin{equation}
         \mathcal{H}(P)=-\int_{k=0}^\infty P(x)\log(P(x)) dx
         \label{eq:entropyInt}
     \end{equation}    
     Because there is a one-one correspondence between the error probability $p$, the reverse spanning graph $\textbf{RSG}$ of the circuit, with the final error distribution $P(k)$, we can also define the entropy of $\textbf{RSG}$ as \autoref{eq:entropy}.
\end{definition}

\begin{theorem}(Entropy increase with time)
\label{thm:Eincrease}
Given a quantum circuit $\textbf{C}$ with $n$ qubits and not quantum gates. The entropy of error distribution defined in \autoref{def:entropy} always increase.
\end{theorem}

\begin{theorem}(Entropy increase with propagation)
\label{thm:Eincrease}
Given a quantum circuit $\textbf{C}$ with $n$ qubits and some $CNOT$ quantum gates. The entropy of error distribution defined in \autoref{def:entropy} always increase.
\end{theorem}

\begin{theorem}(Entropy in disjoint subgraph)
\label{thm:Edisjoint}
Given a quantum circuit $\textbf{C}$ with $n$ qubits and some $CNOT$ quantum gates. If the $\textbf{RSG}$ of the circuit has disjoin subgraphs $\textbf{G}_1$,$\textbf{G}_2$, then 
\begin{equation}
    \mathcal{H}(\textbf{RSG})=\mathcal{H}(\textbf{G}_1)+\mathcal{H}(\textbf{G}_2)
\end{equation}
\end{theorem}

\section{Simulation and experiment result}

We implement the numerical simulation in python. The source code is stored in the following github repository \url{https://github.com/yezhuoyang/Qerrorprop}. First, to justify that our model is mathematically correct and the phenomena of error propagation exists, we implement experiment for empty circuit, transversal CNOT circuit and random circuit. Because exact calculation for error number distribution is hard, we use sampling method to approximate the error number distribution, expectation value and variance.

\subsection{Experiment with probabilistic model}

\begin{figure}[h!]
    \centering
\fbox{
\begin{minipage}{0.45\textwidth}
     \centering
    \includegraphics[width=\linewidth]{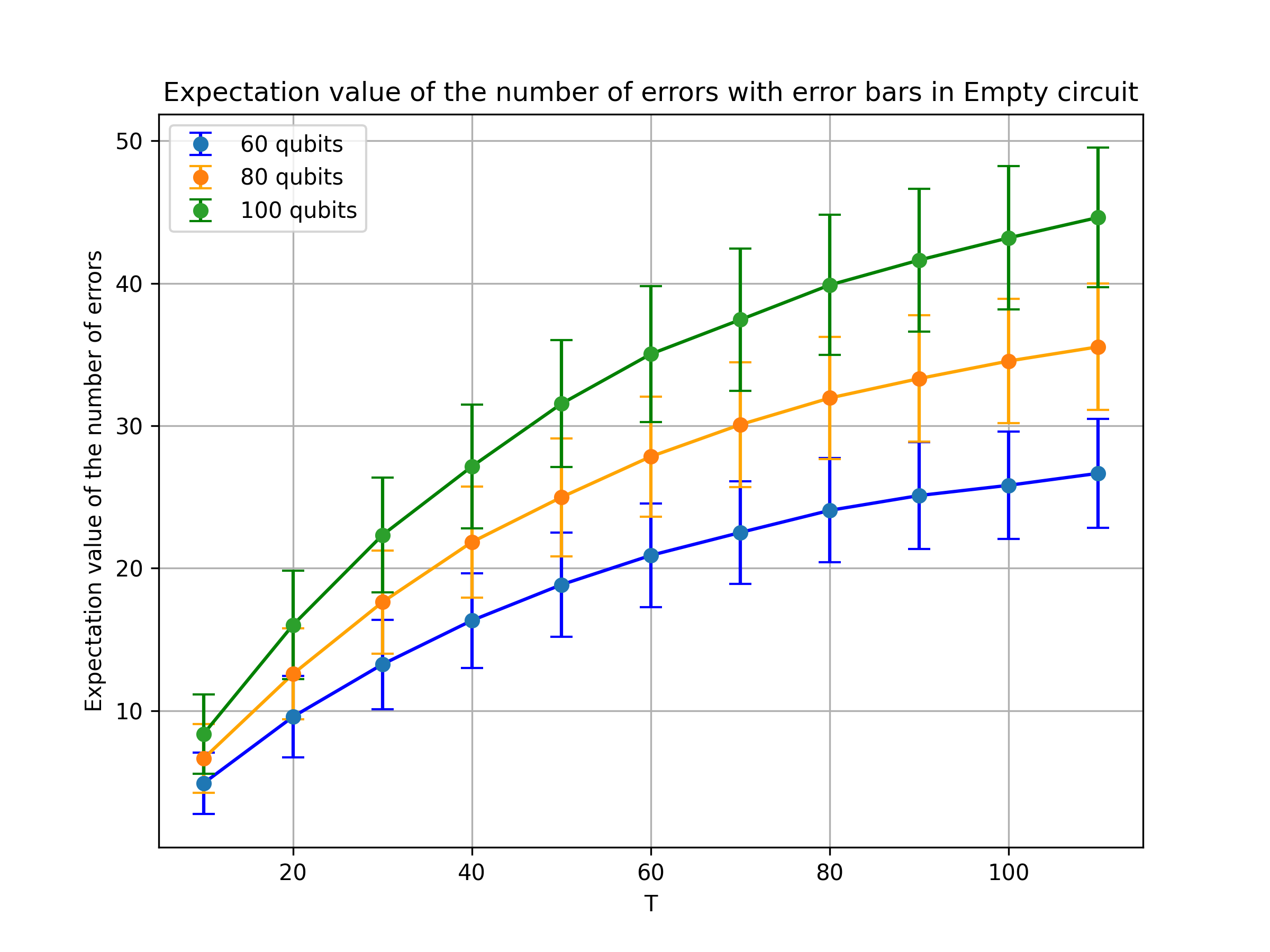}
    \caption{We do experiment for empty circuit with the random sampling algorithm for qubit number $n=60,80,100$, bit-flip rate $p=0.01$, and $T$ from $0$ to $100$. All three curve converges to $\frac{n}{2}$.}
    \label{fig:expErrorEmpty}
\end{minipage}
}
\fbox{
\begin{minipage}{0.45\textwidth}
     \centering
    \includegraphics[width=\linewidth]{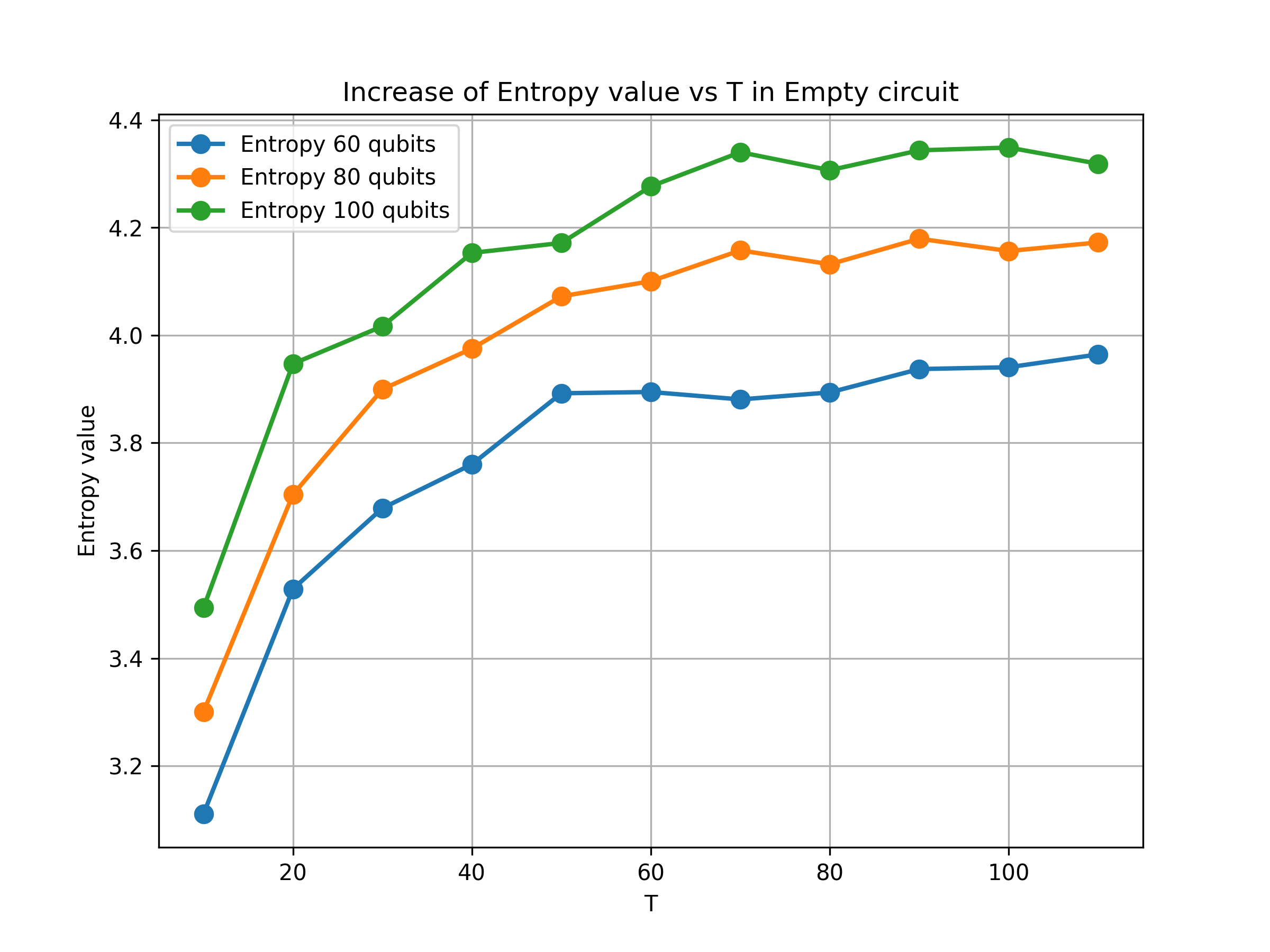}
    \caption{We plot the entropy of the distribution with respect to $T$.}
    \label{fig:entropy}
\end{minipage}
}
\end{figure}

We implement the experiment for an empty circuit with different $T$ using random sampling algorithm. As shown in \autoref{fig:expErrorEmpty}, the result is consistent with what we have proved in \autoref{lem:expC0}, \autoref{thm:shifttime} and \autoref{thm:shiftp}, that the expectation value increase monotonically with time and finally converge to $\frac{n}{2}$. We also calculate the entropy defined in \autoref{def:entropy} for different $T$. As expected by intuition, the entropy also increase. 

\begin{figure}[h!]
\centering
\begin{minipage}{0.45\textwidth}
     \centering
    \includegraphics[width=\linewidth]{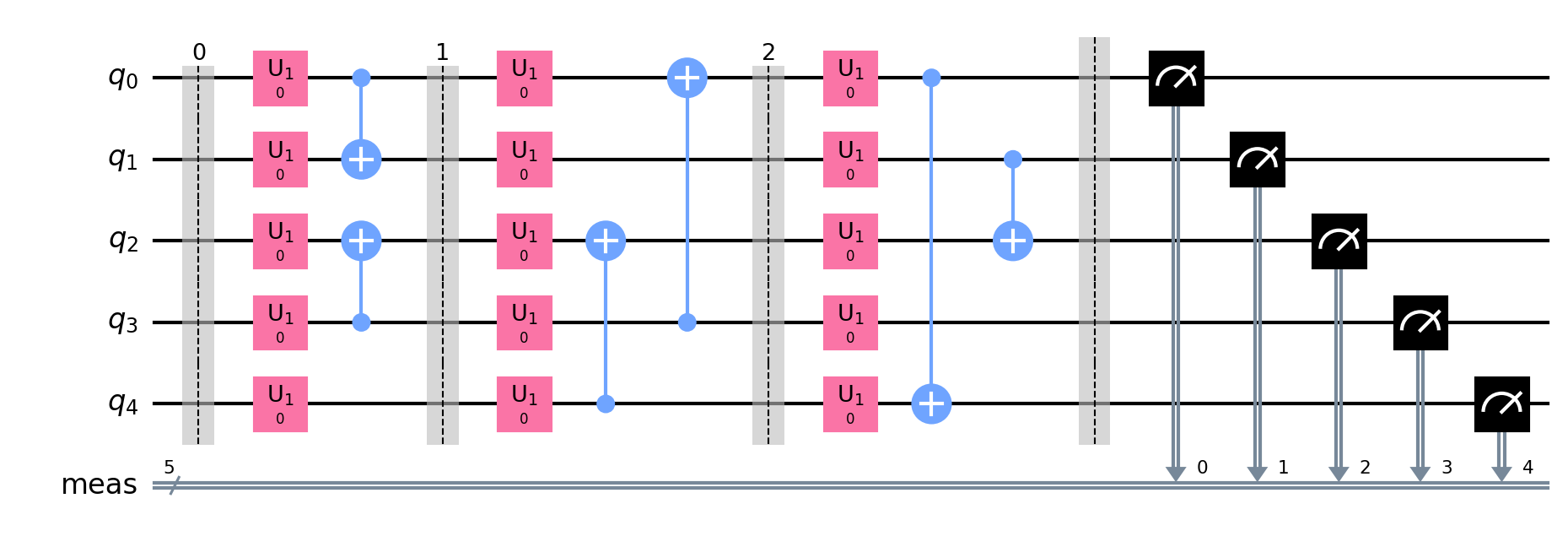}
\end{minipage}
\begin{minipage}{0.45\textwidth}
     \centering
    \includegraphics[width=\linewidth]{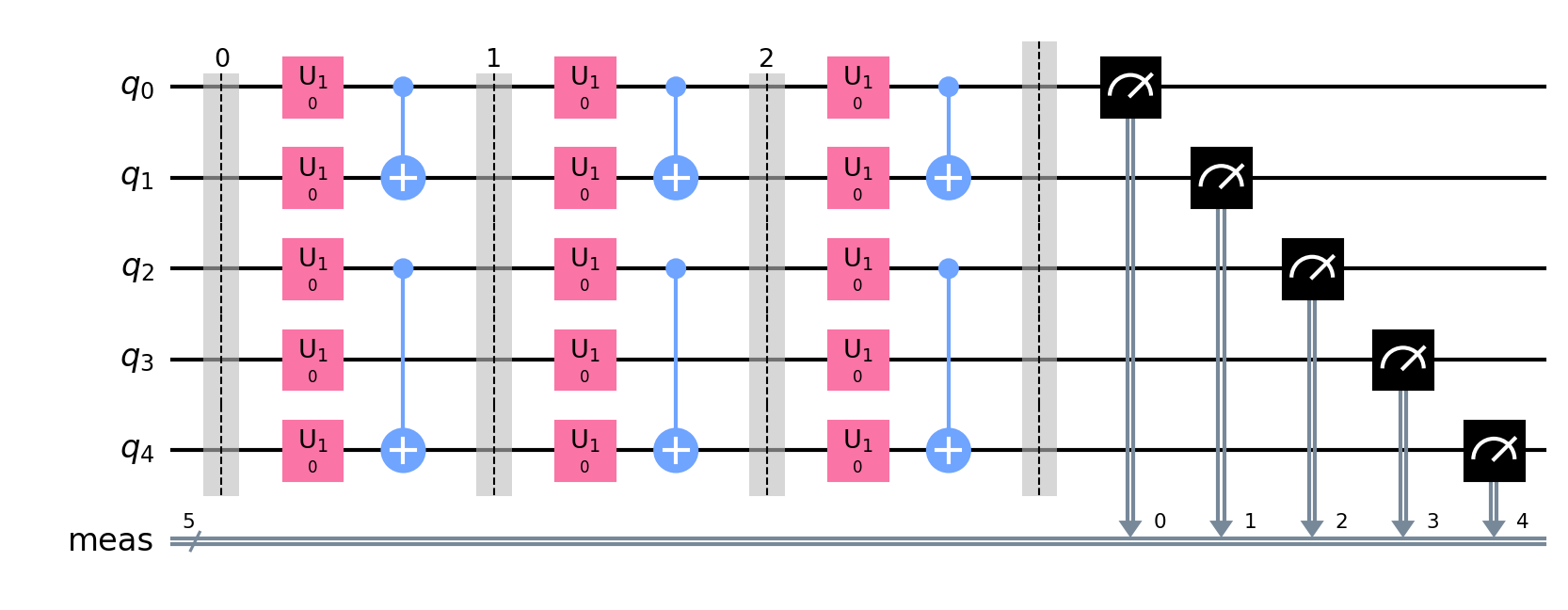}
\end{minipage}
    \caption{Illustration of the two different types of random circuit we use. The left figure is a random circuit with global connection, the right figure is a random circuit with only fixed local connection. }
\label{fig:RandomCircuit}
\end{figure}

\begin{figure}
    \centering
    \includegraphics[width=0.9\linewidth]{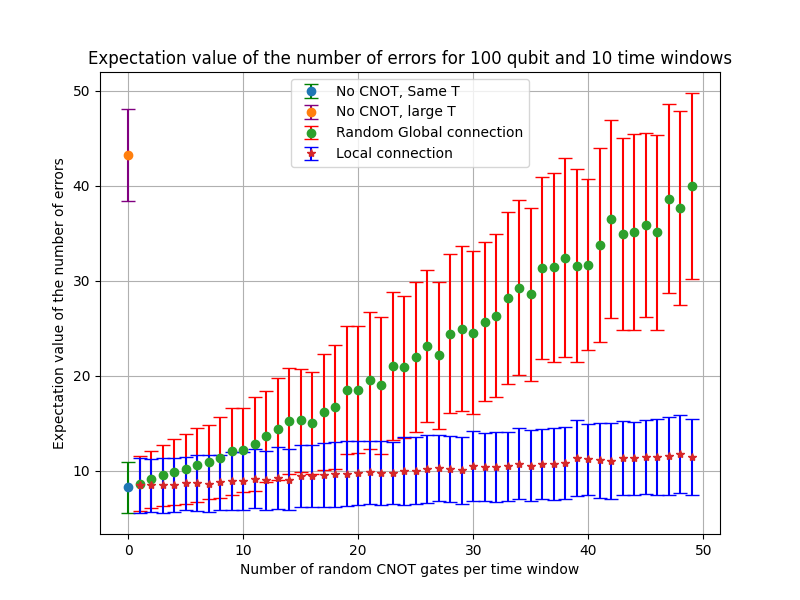}
    \caption{The simulation result of the error number distribution after propagation in random generated circuit. For each fixed qubit number, we generate a fixed number of random $CNOT$ gates in each time window. We also generate another circuit with the same number of $CNOT$ gates, but only has local fixed connection, as shown in \autoref{fig:RandomCircuit}. The expectation value for error number increases with respect to the density of $CNOT$ gates, as expect by intuition. And there is a sharp contrast between the error propagation caused by globally connected circuit with locally connected circuit. }
    \label{fig:randomCircuitResult}
\end{figure}

To justify that the error propagation has a large impact on globally connected circuit. We design the following experiment of random circuit:
\begin{enumerate}
    \item Fix $T$, $n$ and a value $k$ represent the number of random $CNOT$ gates to be sampled in each time windows.
    \item Within each time window, we randomly sample $k$ $CNOT$ gates on $k$ different random pairs of qubit. We call this circuit $\mathcal{C}_1$. An example when $n=5,k=2,T=3$ is the left figure in \autoref{fig:RandomCircuit}. 
    \item Within each time window, we fix the same set of qubit pair, and add one $CNOT$ gate to each of these qubit pair. We call this circuit $\mathcal{C}_2$.An example is the right figure in \autoref{fig:RandomCircuit}.
    \item We run the random error sample algorithm and estimate $E(\textbf{RSG}(\mathbf{G}(\mathcal{C}_1,n,0,T)[1,T,p]))$,
    $Var(\textbf{RSG}(\mathbf{G}(\mathcal{C}_1,n,0,T)[1,T,p])),E(\textbf{RSG}(\mathbf{G}(\mathcal{C}_2,n,0,T)[1,T,p])),Var(\textbf{RSG}(\mathbf{G}(\mathcal{C}_2,n,0,T)[1,T,p]))$.
\end{enumerate}

The result of the above experiment is shown in \autoref{fig:randomCircuitResult}. The error propagation in random circuit with global connection clearly has a severe error propagation problem, while the circuit with fixed local connection has bounded error propagation.

\begin{figure}
    \centering
    \includegraphics[width=0.8\linewidth]{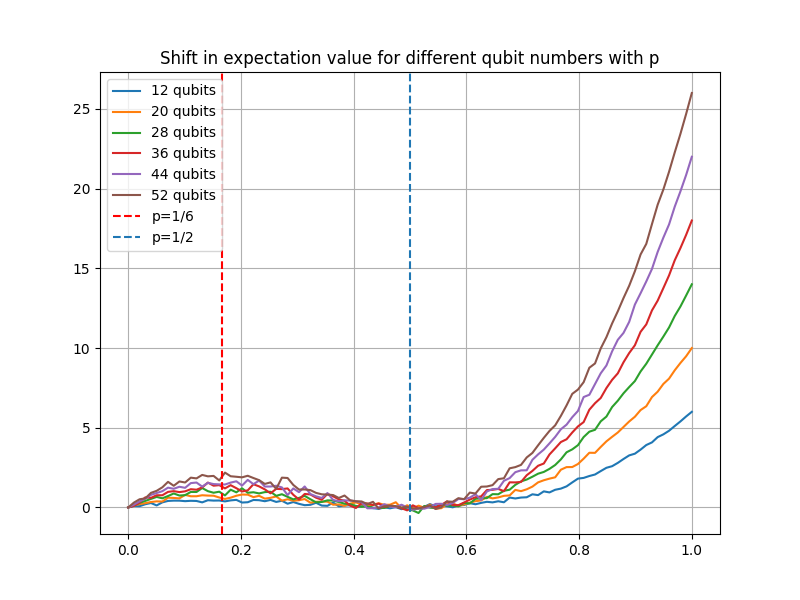}
    \caption{The simulation result for parallel transversal qubit for $n=12,20,28,36,44,52$, with respect to different $p$. We use the random sampling algorithm with $1000$ shots to estimate the expectation value and the shift. The simulation result justify the correctness of  \autoref{cor:shifttransversal}. }
    \label{fig:ShiftExpectationTransP}
\end{figure}

In \autoref{cor:shifttransversal}, we theoretically prove that the shift of error propagation is bounded by $\frac{n}{27}$ when $p=\frac{1}{6}$. Here we also implement numerical simulation. As shown in \autoref{fig:ShiftExpectationTransP}.

\begin{figure}
    \centering
    \includegraphics[width=0.8\linewidth]{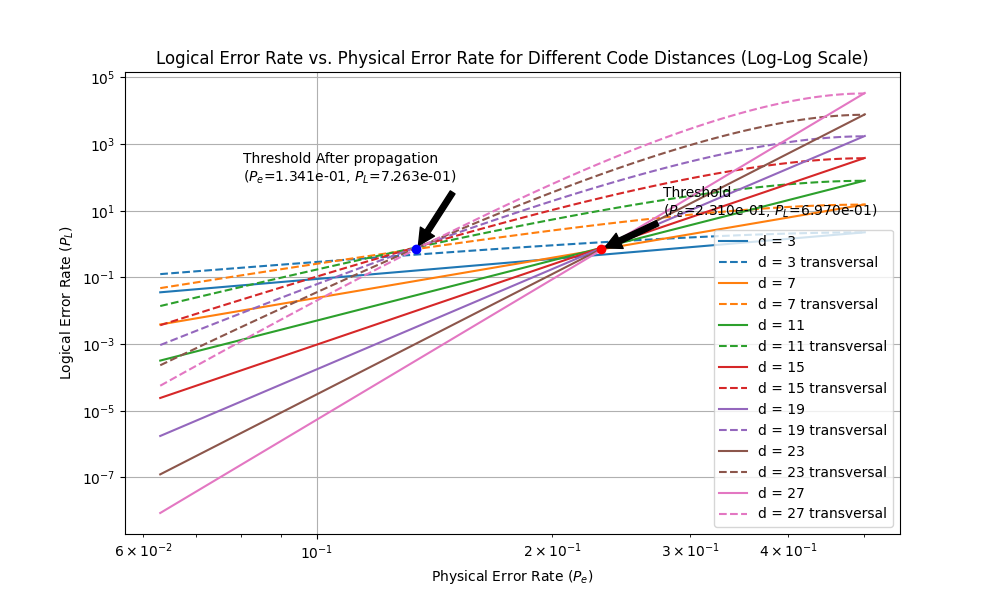}
    \caption{The error threshold before and after error propagation for surface code with transversal logical $CNOT$ gate. The threshold decreases from $0.231$ to $0.134$ because of propagation.}
    \label{fig:surfaceThres}
\end{figure}

Moreover, we apply out framework in the estimation of logical error rate in surface code with parallel transversal $CNOT$ gates. The threshold after error propagation, as demonstrated in \autoref{fig:surfaceThres}, decrease from $0.231$ to $0.134$, after parallel transversal $CNOT$. The significance of the simulation result is that even if the error propagation when circuit only has local connection, the threshold can still be affected a lot. This will case more overhead in a fault tolerant circuit, as $d \approx (\frac{p}{p_{thres}})^{\frac{1}{\epsilon}}$.

\subsection{Qiskit experiment}

\begin{figure}[h!]
    \centering
\fbox{    
\begin{minipage}{0.45\textwidth}
     \centering
    \includegraphics[width=\linewidth]{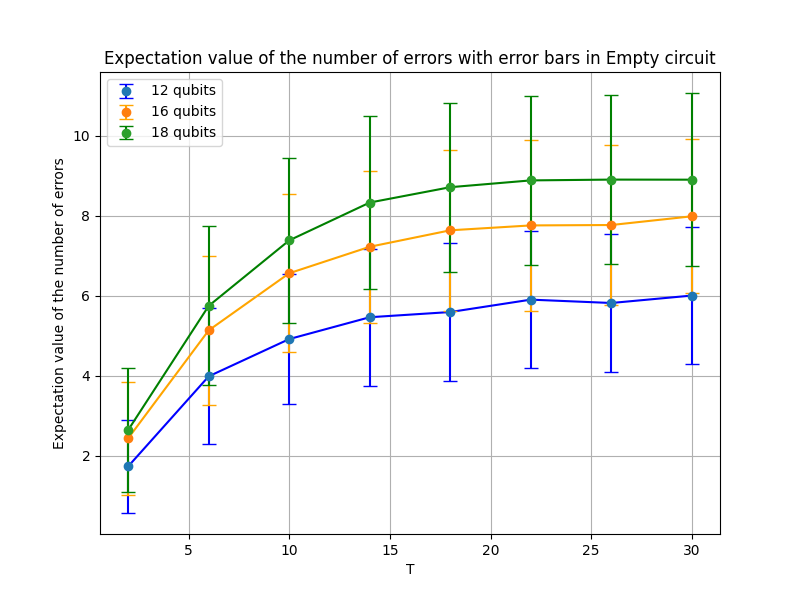}
    \caption{The experiment result done in qiskit Aersimulator. We plot the expectation value of the total number of error for empty circuit for qutbit number $n=12,16,18$, and $T$ from $1$ to $30$. All three curves converge to $\frac{n}{2}$ when $T$ becomes large enough.}
    \label{fig:expErrorEmptyqiskit}
\end{minipage}
}
\fbox{
\begin{minipage}{0.45\textwidth}
     \centering
    \includegraphics[width=\linewidth]{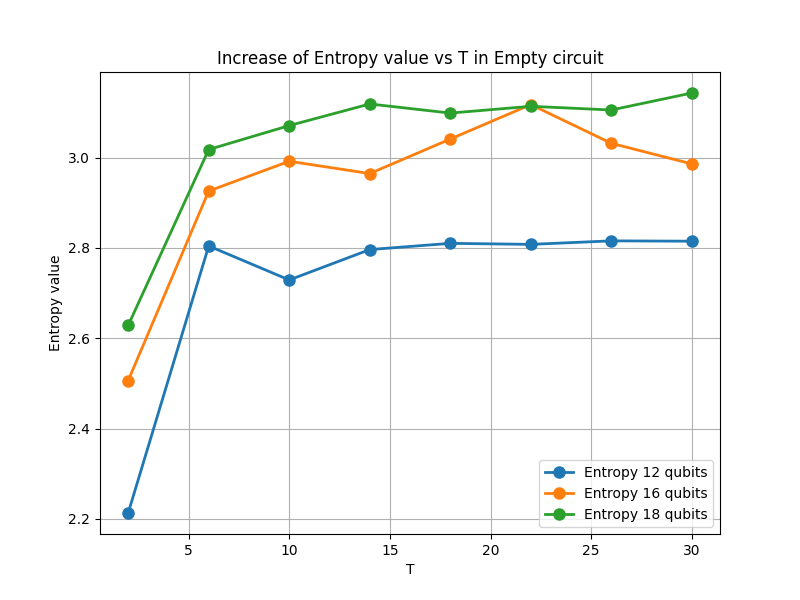}
    \caption{We plot the entropy of empty circuit for different time $T$.}
    \label{fig:entropyqiskit}
\end{minipage}
}
\end{figure}

\begin{figure}
    \centering
    \includegraphics[width=0.9\linewidth]{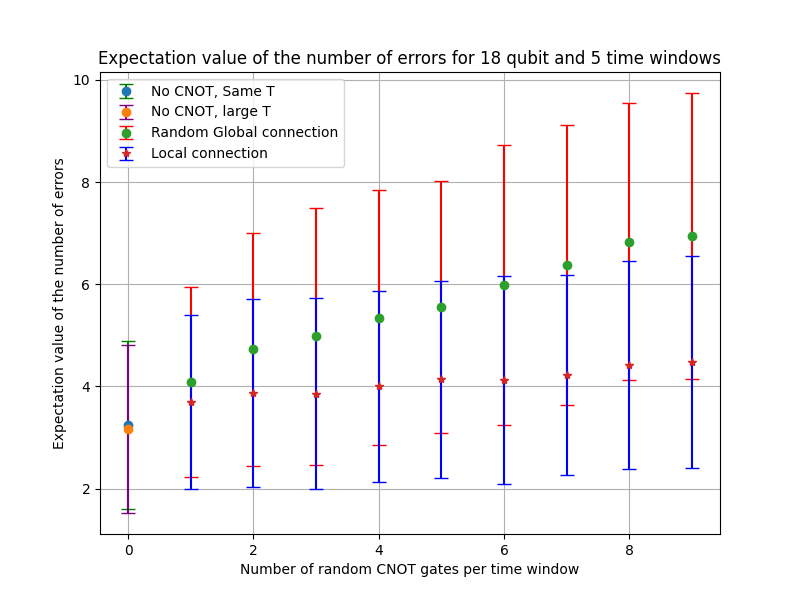}
    \caption{The simulation result of the error number distribution after propagation in random generated circuit with qiskit simulator.}
    \label{fig:randomCircuitResultqiskit}
\end{figure}

We also run some experiment with qiskit-Aer simulator. In qiskit simulation, we use bitflip noise model attached to single qubit $U_1$ gate with bit-flip probability $p$,  where we set $U_1=\hat{I}$. 
In \autoref{fig:expErrorEmptyqiskit} and \autoref{fig:entropyqiskit}, we do the same numerical simulation with empty circuit, with qubit number $n=12,16,18$. And all three curves show the same convergence trend with time $T$. We also run the the same random circuit experiment with qiskit for $n=18$ and $T=5$.  There is also sharp contrast between the circuit with random circuit with global connection and local connection. The consistency between qiskit simulation result and the random sampling result of our model justify the correctness of our modeling.

\FloatBarrier

\section{Conclusion and future direction}

In this paper, a new theoretical framework of measuring error propagation statistically has been made. 

The construction of such model from a general circuit can be done in polynomial time, but the exact calculation of error number distribution is hard. From numerical result, however, we observe that random sampling algorithms give us rather accurate approximate distribution. We have done numerical simulation to show that error propagation get severe with increasing error accumulation time and circuit connectivity. Even though the distribution shift of locally connected transversal circuit is bounded by $\frac{n}{27}$, the error threshold can still decrease a lot after error propagation. All evidence clearly demonstrate that we should put emphasis on error propagation. Not only should we bound the connectivity of the circuit, but we should also optimize the measure of propagation as much as possible to reduce the overhead of error correction.  

We list some future direction for this research:
\begin{enumerate}
    \item Is there a better and more general theoretical bound for error distribution shift after propagation? 
    \item We can apply this framework to more code construction.
    \item We should generalize the framework to circuit that has both $Z$ and $X$ error propagation.
    \item We could explore the possibility of designing a better threshold estimation algorithm with the knowledge of error propagation. 
\end{enumerate}

\section{Acknowledgment}
The author thanks the help from the discussion with Sunny He from MIT and the guidence from professor Jens Palsberg from UCLA.


\bibliographystyle{abbrvnat} 
\bibliography{references}  


\end{document}